\documentclass[a4paper,english,cleveref,autoref,thm-restate]{lipics-v2019}

\usepackage{graphicx}
\usepackage{etoolbox}

\newtoggle{fullpaper}
\settoggle{fullpaper}{true}

\iftoggle{fullpaper}{%
}{}

\nolinenumbers

\usepackage{booktabs}
\usepackage{subcaption}
\usepackage{microtype}
\usepackage{amsmath}
\usepackage{nccmath}
\usepackage{mathtools}
\usepackage{amssymb,amsfonts}
\usepackage{stmaryrd}
\usepackage{hyperref}
\usepackage{multicol}
\usepackage{color}
\usepackage{algorithm}
\usepackage[noend]{algpseudocode}
\usepackage{graphicx}
\usepackage{tikz}
\usepackage{thm-restate}
\usepackage{mathpartir}
\usepackage{xspace}
\usepackage{wrapfig}
\usepackage{ifthen}
\usepackage{xcolor}
\usepackage{stackengine}
\usepackage{listings}

\usepackage[all]{xy}\UseComputerModernTips

\usepackage[capitalise]{cleveref}

\definecolor{StringRed}{rgb}{.637,0.082,0.082}
\definecolor{CommentGreen}{rgb}{0.0,0.55,0.3}
\definecolor{KeywordBlue}{rgb}{0.0,0.3,0.55}
\definecolor{LinkColor}{rgb}{0.55,0.0,0.3}
\definecolor{CiteColor}{rgb}{0.55,0.0,0.3}
\definecolor{HighlightColor}{rgb}{0.0,0.0,0.0}

\definecolor{grey}{rgb}{0.5,0.5,0.5}
\definecolor{red}{rgb}{1,0,0}
\definecolor{darkgreen}{rgb}{0.0,0.7,0.0}

\hypersetup{%
  linktocpage=true, pdfstartview=FitV,
  breaklinks=true, pageanchor=true, pdfpagemode=UseOutlines,
  plainpages=false, bookmarksnumbered, bookmarksopen=true, bookmarksopenlevel=3,
  hypertexnames=true, pdfhighlight=/O,
  colorlinks=true,linkcolor=LinkColor,citecolor=CiteColor,
  urlcolor=LinkColor
}

\usetikzlibrary{calc,shapes,shapes.misc,shadows,arrows,decorations.pathreplacing,decorations.markings,decorations.pathmorphing,positioning,automata,fadings,fit,arrows}
\usetikzlibrary{shapes.multipart,matrix}
\usetikzlibrary{shapes.callouts}
\usetikzlibrary{patterns,backgrounds}
\usetikzlibrary{snakes}

\newlength{\hatchspread}
\newlength{\hatchthickness}
\newlength{\hatchshift}
\newcommand{\hatchcolor}{}
\tikzset{hatchspread/.code={\setlength{\hatchspread}{#1}},
         hatchthickness/.code={\setlength{\hatchthickness}{#1}},
         hatchshift/.code={\setlength{\hatchshift}{#1}},
         hatchcolor/.code={\renewcommand{\hatchcolor}{#1}}}
\tikzset{hatchspread=10pt,
         hatchthickness=4pt,
         hatchshift=0pt,
         hatchcolor=black}
\pgfdeclarepatternformonly[\hatchspread,\hatchthickness,\hatchshift,\hatchcolor]
   {custom north west lines}
   {\pgfqpoint{\dimexpr-2\hatchthickness}{\dimexpr-2\hatchthickness}}
   {\pgfqpoint{\dimexpr\hatchspread+2\hatchthickness}{\dimexpr\hatchspread+2\hatchthickness}}
   {\pgfqpoint{\dimexpr\hatchspread}{\dimexpr\hatchspread}}
   {
    \pgfsetlinewidth{\hatchthickness}
    \pgfpathmoveto{\pgfqpoint{0pt}{\dimexpr\hatchspread+\hatchshift}}
    \pgfpathlineto{\pgfqpoint{\dimexpr\hatchspread+0.15pt+\hatchshift}{-0.15pt}}
    \ifdim \hatchshift > 0pt
      \pgfpathmoveto{\pgfqpoint{0pt}{\hatchshift}}
      \pgfpathlineto{\pgfqpoint{\dimexpr0.15pt+\hatchshift}{-0.15pt}}
    \fi
    \pgfsetstrokecolor{\hatchcolor}
    \pgfusepath{stroke}
   }

\pgfdeclarepatternformonly[\hatchspread,\hatchthickness,\hatchshift,\hatchcolor]
   {custom north east lines}
   {\pgfqpoint{\dimexpr-2\hatchthickness}{\dimexpr-2\hatchthickness}}
   {\pgfqpoint{\dimexpr\hatchspread+2\hatchthickness}{\dimexpr\hatchspread+2\hatchthickness}}
   {\pgfqpoint{\dimexpr\hatchspread}{\dimexpr\hatchspread}}
   {
    \pgfsetlinewidth{\hatchthickness}
    \pgfpathmoveto{\pgfqpoint{\dimexpr\hatchshift-0.15pt}{-0.15pt}}
    \pgfpathlineto{\pgfqpoint{\dimexpr\hatchspread+0.15pt}{\dimexpr\hatchspread-\hatchshift+0.15pt}}
    \ifdim \hatchshift > 0pt
      \pgfpathmoveto{\pgfqpoint{-0.15pt}{\dimexpr\hatchspread-\hatchshift-0.15pt}}
      \pgfpathlineto{\pgfqpoint{\dimexpr\hatchshift+0.15pt}{\dimexpr\hatchspread+0.15pt}}
    \fi
    \pgfsetstrokecolor{\hatchcolor}
    \pgfusepath{stroke}
   }

\tikzstyle{issuedStyle}=[pattern=custom north east lines, hatchcolor=colorISS, rounded corners]
\tikzstyle{coveredStyle}=[]

\newcommand{\setBox}[2]{
    \draw[#1] ($(#2)  + (-1.2,-0.4)$) rectangle ++(2.4,0.8);
}
\newcommand{\bigSetBox}[2]{
    \draw[#1] ($(#2)  + (-1.3,-0.5)$) rectangle ++(2.6,1.0);
}

\newcommand{\coveredBox}[1]{
  \setBox{coveredStyle}{#1}
}
\newcommand{\issuedBox}[1]{
  \setBox{issuedStyle}{#1}
}
\newcommand{\issuedCoveredBox}[1]{
  \bigSetBox{coveredStyle}{#1}
  \issuedBox{#1}
}

\newcommand{\nlSetBox}[2]{
    \draw[#1] ($(#2)  + (-0.85,-0.55)$) rectangle ++(1.7,1.1);
}
\newcommand{\nlBigSetBox}[2]{
    \draw[#1] ($(#2)  + (-0.95,-0.65)$) rectangle ++(1.9,1.3);
}
\newcommand{\nlIssuedBox}[1]{
  \nlSetBox{issuedStyle}{#1}
}
\newcommand{\nlCoveredBox}[1]{
  \nlSetBox{coveredStyle}{#1}
}
\newcommand{\nlIssuedCoveredBox}[1]{
  \nlBigSetBox{coveredStyle}{#1}
  \nlIssuedBox{#1}
}

\newcommand{\nlCoveredDoubleBox}[2]{
    \draw[coveredStyle] ($(#1)  + (0.95,0.6)$) rectangle ($(#2)  + (-0.95,-0.6)$);
}

\tikzstyle{extractStyle}=[color=black,rounded corners=3pt,dashed,fill=green!10]

\newcommand{\smallIssuedBoxText}{{\protect\tikz \protect\draw[issuedStyle] (0,0) rectangle ++(0.35,0.35);}}
\newcommand{\smallCoveredBoxText}{{\protect\tikz \protect\draw[coveredStyle] (0,0) rectangle ++(0.35,0.35);}}
\newcommand{\smallXBoxText}{{\protect\tikz \protect\draw[extractStyle] (0,0) rectangle ++(0.35,0.35);}}

\iftoggle{fullpaper}{%
  \newcommand{\citeapp}[2]{\cite[\cref{#1}]{appendix}}  
}{%
  \newcommand{\citeapp}[2]{\cite[#2]{appendix}}
}

\crefformat{section}{#2\S{}#1#3}
\Crefname{section}{Section}{Sections}
\Crefformat{section}{Section #2#1#3}

\crefformat{subsection}{#2\S{}#1#3}
\Crefname{subsection}{Section}{Sections}
\Crefformat{subsection}{Section #2#1#3}

\crefname{figure}{\text{Fig.}}{\text{Figures}}
\Crefname{figure}{\text{Figure}}{\text{Figures}}
\crefname{corollary}{\text{Corollary}}{\text{corollaries}}
\Crefname{corollary}{\text{Corollary}}{\text{Corollaries}}
\crefname{lemma}{\text{Lemma}}{\text{Lemmas}}
\Crefname{lemma}{\text{Lemma}}{\text{Lemmas}}
\crefname{proposition}{\text{Prop.}}{\text{Prop.}}
\Crefname{proposition}{\text{Proposition}}{\text{Propositions}}
\crefname{definition}{\text{Def.}}{\text{Definitions}}
\Crefname{definition}{\text{Definition}}{\text{Definitions}}
\crefname{notation}{\text{Notation}}{\text{Notations}}
\Crefname{notation}{\text{Notation}}{\text{Notations}}
\crefname{theorem}{\text{Theorem}}{\text{Theorems}}
\Crefname{theorem}{\text{Theorem}}{\text{Theorems}}
\crefname{conjecture}{\text{Conj.}}{\text{Conjectures}}
\Crefname{conjecture}{\text{Conjecture}}{\text{Conjectures}}

\newcommand{\textcode}[1]{\texorpdfstring{\texttt{#1}}{#1}}
\newcommand{\kw}[1]{\textbf{\textcode{#1}}}

\newcommand{\ie}{\emph{i.e.,} }
\newcommand{\eg}{\emph{e.g.,} }

\newcommand{\sth}{\emph{s.t.} }
\newcommand{\etal}{\emph{et~al.}}

\newcommand{\inarrC}[1]{\begin{array}{@{}c@{}}#1\end{array}}

\newcommand{\inarr}[1]{\begin{array}{@{}l@{}}#1\end{array}}
\newcommand{\inarrII}[2]{\begin{array}{@{}l@{~~}||@{~~}l@{}}\inarr{#1}&\inarr{#2}\end{array}}
\newcommand{\inarrIII}[3]{\begin{array}{@{}l@{~~}||@{~~}l@{~~}||@{~~}l@{}}\inarr{#1}&\inarr{#2}&\inarr{#3}\end{array}}

\renewcommand{\comment}[1]{\color{teal}{~~\texttt{/\!\!/}\textit{#1}}}
\newcommand{\nocomment}[1]{\color{red!60!black}{~~\texttt{/\!\!/}\textit{#1}}}

\newcommand{\set}[1]{\{{#1}\}}

\newcommand{\fn}{\rightarrow}

\renewcommand{\st}{\; | \;}
\newcommand{\N}{{\mathbb{N}}}

\newcommand{\dom}[1]{\textit{dom}{({#1})}}
\newcommand{\codom}[1]{\textit{codom}{({#1})}}

\newcommand{\tup}[1]{{\langle{#1}\rangle}}
\newcommand{\nin}{\not\in}
\newcommand{\suq}{\subseteq}

\newcommand{\rst}[1]{|_{#1}}

\newcommand{\defeq}{\triangleq}

\newcommand{\seq}{\mathbin{;}}

\newcommand\TidSet{T}

\colorlet{colorPO}{gray!60!black}
\colorlet{colorCF}{red!60!black}
\colorlet{colorECF}{red!60!black}
\colorlet{colorJF}{blue!60!black}
\colorlet{colorRF}{green!60!black}
\colorlet{colorEW}{brown}
\colorlet{colorMO}{orange}
\colorlet{colorFR}{purple}
\colorlet{colorECO}{red!80!black}
\colorlet{colorSYN}{green!40!black}
\colorlet{colorHB}{blue}
\colorlet{colorPPO}{magenta}
\colorlet{colorPB}{olive}
\colorlet{colorSBRF}{olive}
\colorlet{colorRMW}{olive!70!black}
\colorlet{colorRS}{blue}
\colorlet{colorRELEASE}{blue!70!black}
\colorlet{colorSC}{olive!40!black}
\colorlet{colorPSC}{olive!40!black}
\colorlet{colorREL}{olive}
\colorlet{colorCONFLICT}{olive}
\colorlet{colorRACE}{olive}
\colorlet{colorWB}{orange!70!black}
\colorlet{colorSCB}{violet}
\colorlet{colorDETOUR}{teal}
\colorlet{colorDEPS}{violet}
\colorlet{colorFENCE}{olive}
\colorlet{colorCOV}{magenta!20}
\colorlet{colorISS}{blue!10!white}
\colorlet{colorVF}{purple!70!black}

\tikzset{
   every path/.style={>=stealth},
   po/.style={->,color=colorPO,shorten >=-0.5mm,shorten <=-0.5mm},
   sw/.style={->,color=colorSYN,shorten >=-0.5mm,shorten <=-0.5mm},
   sc/.style={->,color=colorSC,dotted,thick,shorten >=-0.5mm,shorten <=-0.5mm},
   rf/.style={->,color=colorRF,dashed,,shorten >=-0.5mm,shorten <=-0.5mm},
   hb/.style={->,color=colorHB,thick,shorten >=-0.5mm,shorten <=-0.5mm},
   mo/.style={->,color=colorMO,dotted,very thick,shorten >=-0.5mm,shorten <=-0.5mm},
   co/.style={->,color=colorMO,dotted,thick,shorten >=-0.5mm,shorten <=-0.5mm},
   no/.style={->,dotted,thick,shorten >=-0.5mm,shorten <=-0.5mm},
   fr/.style={->,color=colorFR,dotted,thick,shorten >=-0.5mm,shorten <=-0.5mm},
   deps/.style={->,color=colorDEPS,dotted,thick,shorten >=-0.5mm,shorten <=-0.5mm},
   ppo/.style={->,color=colorPPO,shorten >=-0.5mm,shorten <=-0.5mm},
   rmw/.style={->,color=colorRMW,thick,shorten >=-0.5mm,shorten <=-0.5mm},
   detour/.style={->,color=colorDETOUR,shorten >=-0.5mm,shorten <=-0.5mm},
   cf/.style={-,snake=zigzag,segment amplitude=1pt,segment length=3pt,colorCF},
   ew/.style={<->,dashed,,shorten >=-0.5mm,shorten <=-0.5mm,color=colorEW},
   jf/.style={->,color=colorJF,dotted,thick,shorten >=-0.5mm,shorten <=-0.5mm},
   vf/.style={->,color=colorVF,dashed,shorten >=-0.5mm,shorten <=-0.5mm},
}

\newcommand{\rlx}{\mathtt{rlx}}
\newcommand{\rel}{{\mathtt{rel}}}
\newcommand{\acq}{{\mathtt{acq}}}
\newcommand{\acqrel}{{\mathtt{acqrel}}}
\newcommand{\sco}{{\mathtt{sc}}}

\newcommand{\full}{{\mathtt{sy}}}
\newcommand{\ld}{{\mathtt{ld}}}

\newcommand{\isync}{{\mathtt{isync}}}
\newcommand{\lwsync}{{\mathtt{lwsync}}}
\newcommand{\sync}{{\mathtt{sync}}}

\newcommand{\isex}{{\mathtt{ex}}}

\newcommand{\rlab}[3]{{\lR}^{#1}({#2},{#3})}
\newcommand{\erlab}[4]{
\ifthenelse{\equal{#2}{}}{{\lR}^{#1}_{#4}}
{\ifthenelse{\equal{#3}{}\and\equal{#1}{}\and\equal{#2}{}}{{\lR}^({#2})}
{\ifthenelse{\equal{#3}{}}{{\lR}^{#1}_{#4}({#2})}
{{\lR}^{#1}_{#4}({#2},{#3})}}}}
\newcommand{\ewlab}[4]{
\ifthenelse{\equal{#2}{}}{{\lW}^{#1}_{#4}}
{\ifthenelse{\equal{#3}{}\and\equal{#1}{}\and\equal{#2}{}}{{\lW}^({#2})}
{\ifthenelse{\equal{#3}{}}{{\lW}^{#1}_{#4}({#2})}
{{\lW}^{#1}_{#4}({#2},{#3})}}}}

\newcommand{\prlab}[2]{{\lR}({#1},{#2})}
\newcommand{\wlab}[3]{{\lW}^{#1}({#2},{#3})}
\newcommand{\pwlab}[2]{{\lW}({#1},{#2})}
\newcommand{\flab}[1]{{\lF}^{#1}}

\newcommand{\lE}{{\mathtt{E}}}

\newcommand{\lEo}{\Init}

\newcommand{\lR}{{\mathtt{R}}}
\newcommand{\lW}{{\mathtt{W}}}
\newcommand{\lA}{{\mathtt{A}}}
\newcommand{\lQ}{{\mathtt{Q}}}
\newcommand{\lL}{{\mathtt{L}}}

\newcommand{\lF}{{\mathtt{F}}}

\newcommand{\lRMWc}{{\mathtt{RMW}}}

\newcommand{\lLAB}{{\mathtt{lab}}}
\newcommand{\lTID}{{\mathtt{tid}}}
\newcommand{\lSN}{{\mathtt{sn}}}
\newcommand{\lTYP}{{\mathtt{typ}}}
\newcommand{\lLOC}{{\mathtt{loc}}}
\newcommand{\lMOD}{{\mathtt{mod}}}
\newcommand{\lVAL}{{\mathtt{val}}}

\newcommand{\lPO}{{\color{colorPO}\mathtt{po}}}
\newcommand{\lPOimm}{{\color{colorPO}\mathtt{po_{imm}}}}
\newcommand{\lCF}{{\color{colorCF}\mathtt{cf}}}
\newcommand{\lCFimm}{{\color{colorCF}\mathtt{cf_{imm}}}}
\newcommand{\lECF}{{\color{colorECF}\mathtt{ecf}}}
\newcommand{\lJF}{{\color{colorJF} \mathtt{jf}}}
\newcommand{\lRF}{{\color{colorRF} \mathtt{rf}}}
\newcommand{\lRMW}{{\color{colorRMW} \mathtt{rmw}}}
\newcommand{\lEW}{{\color{colorEW} \mathtt{ew}}}
\newcommand{\lCO}{{\color{colorMO} \mathtt{co}}}

\newcommand{\lFR}{{\color{colorFR} \mathtt{fr}}}

\newcommand{\lECO}{{\color{colorECO} \mathtt{eco}}}

\newcommand{\lRELEASE}{{\color{colorRELEASE}\mathtt{release}}}
\newcommand{\lSW}{{\color{colorSYN}\mathtt{sw}}}
\newcommand{\lHB}{{\color{colorHB}\mathtt{hb}}}
\newcommand{\lDOB}{{\mathtt{dob}}}

\newcommand{\lBOB}{{\mathtt{bob}}}
\newcommand{\lAOB}{{\mathtt{aob}}}
\newcommand{\lOBS}{{\mathtt{obs}}}

\newcommand{\lVF}{{\color{colorVF}\mathtt{vf}}}
\newcommand{\lSRF}{{\color{colorJF}\mathtt{sjf}}}

\newcommand{\lSCB}{{\color{colorSCB} \mathtt{scb}}}
\newcommand{\lPSC}{{\color{colorPSC} \mathtt{psc}}}
\newcommand{\lPSCB}{\lPSC_{\rm base}}
\newcommand{\lPSCF}{\lPSC_\lF}

\newcommand{\lCTRL}{{{\color{colorDEPS}\mathtt{ctrl}}}}
\newcommand{\lDATA}{{{\color{colorDEPS}\mathtt{data}}}}
\newcommand{\lADDR}{{{\color{colorDEPS}\mathtt{addr}}}}
\newcommand{\lPPO}{{{\color{colorPPO}\mathtt{ppo}}}}
\newcommand{\lRMWDEP}{{{\color{colorDEPS}\mathtt{casdep}}}}

\newcommand{\lmakeE}[1]{#1\mathtt{e}}
\newcommand{\lJFE}{\lmakeE{\lJF}}
\newcommand{\lRFE}{\lmakeE{\lRF}}
\newcommand{\lCOE}{\lmakeE{\lCO}}
\newcommand{\lFRE}{\lmakeE{\lFR}}

\newcommand{\lmakeI}[1]{#1\mathtt{i}}

\newcommand{\lRFI}{\lmakeI{\lRF}}
\newcommand{\lCOI}{\lmakeI{\lCO}}

\newcommand{\lFENCE}{\mathtt{fence}}

\newcommand{\Tid}{\mathsf{Tid}}
\newcommand{\Loc}{\mathsf{Loc}}
\newcommand{\Val}{\mathsf{Val}}
\newcommand{\Lab}{\mathsf{Lab}}

\newcommand{\Init}{\mathsf{Init}}

\newcommand{\tsomodel}{_{\TSO}}
\newcommand{\lHBTSO}{\lHB\tsomodel}
\newcommand{\lEHBTSO}{{\color{colorHB}\mathtt{ehb}}\tsomodel}
\newcommand{\lPPOTSO}{\lPPO\tsomodel}
\newcommand{\lFENCETSO}{\lFENCE\tsomodel}
\newcommand{\lMFENCE}{\mathtt{MFENCE}}
\newcommand{\lIFENCE}{\mathtt{implied\_fence}\tsomodel}

\newcommand{\scmodel}{_{\mathsf{SC}}}
\newcommand{\POWER}{\ensuremath{\mathsf{POWER}}\xspace}

\newcommand{\IMM}{\ensuremath{\mathsf{IMM}}\xspace}

\newcommand{\IMMsc}{\ensuremath{\mathsf{IMM\scmodel}}\xspace}
\newcommand{\ARM}{\ensuremath{\mathsf{ARM}}\xspace}
\newcommand{\ARMs}{\ensuremath{\mathsf{ARMv7}}\xspace}
\newcommand{\ARMe}{\ensuremath{\mathsf{ARMv8}}\xspace}
\newcommand{\TSO}{\ensuremath{\mathsf{TSO}}\xspace}

\newcommand{\weakestmo}{\ensuremath{\mathsf{Weakestmo}}\xspace}
\newcommand{\Intel}{\ensuremath{\mathsf{x86}}\xspace}

\newcommand{\prog}{\mathit{prog}}

\newcommand{\compile}[1]{{(\!|} #1 {|\!)}}

\newcommand{\lDMBSY}{\mathtt{F^{\full}}}
\newcommand{\lDMBLD}{\mathtt{F^{\ld}}}

\newcommand{\Promise}{\ensuremath{\mathsf{PS}}\xspace}

\newcommand{\loc}{x}
\newcommand{\locx}{x}
\newcommand{\locy}{y}
\newcommand{\locz}{z}

\newcommand{\tid}{t}

\newcommand{\simrel}{\mathcal{I}}
\newcommand{\simrelCert}{\mathcal{I}^{\rm \bf cert}}

\newcommand{\pstate}{\sigma}

\newcommand{\readInst }[3]{#2 \;{:=}\;[#3]^{#1}}
\newcommand{\fenceInst}[1]{\kw{fence}^{#1}}

\newcommand{\writeInst}[3]{[#2]^{#1}\;{:=}\;#3}

\newcommand{\incInst}[6]{#3 \;{:=}\;\kw{FADD}_{#6}^{#1#2}({#4},{#5})}
\newcommand{\casInst}[7]{#3 \;{:=}\;\kw{CAS}_{#7}^{#1#2}({#4},{#5},{#6})}

\newcommand{\RC}{\ensuremath{\mathsf{RC11}}\xspace}

\newcounter{mylabelcounter}

\makeatletter
\newcommand{\labelAxiom}[2]{%
\hfill{\normalfont\textsc{(#1)}}\refstepcounter{mylabelcounter}
\immediate\write\@auxout{
  \string\newlabel{#2}{{\unexpanded{\normalfont\textsc{#1}}}{\thepage}{{\unexpanded{\normalfont\textsc{#1}}}}{mylabelcounter.\number\value{mylabelcounter}}{}}
}
}
\makeatother

\newcommand{\squishlist}[1][$\bullet$]{
 \begin{list}{#1}
  { \setlength{\itemsep}{0pt}
     \setlength{\parsep}{0pt}
     \setlength{\topsep}{1pt}
     \setlength{\partopsep}{0pt}
     \setlength{\leftmargin}{1.2em}
     \setlength{\labelwidth}{0.5em}
     \setlength{\labelsep}{0.4em} } }
\newcommand{\squishend}{
  \end{list}  }

\newcommand{\IssuedSet}{I}
\newcommand{\CoveredSet}{C}

\newcommand{\determined}{\mathtt{determined}}

\newcommand{\myrightarrow}[1]{
\ifthenelse{\equal{#1}{}}{\xrightarrow{}}{\mathrel{\raisebox{-2pt}{$\xrightarrow{#1}$}}}}

\allowdisplaybreaks

\newcommand{\ese}[3]{e^{#1}_{#2#3}}
\newcommand{\mese}[3]{\ese{#1}{#2}{#3}\colon}

\newcommand{\TC}{TC}
\newcommand\doubleplus{+\kern-1.3ex+\kern0.8ex}

\newcommand{\TCgen}[1]{\ensuremath{\TC_{{\mathsf{#1}}}}\xspace}
\newcommand{\TCa}{\TCgen{a}}
\newcommand{\TCb}{\TCgen{b}}
\newcommand{\TCc}{\TCgen{c}}
\newcommand{\TCd}{\TCgen{d}}
\newcommand{\TCe}{\TCgen{e}}
\newcommand{\TCf}{\TCgen{f}}
\newcommand{\TCinit}[1]{\TC_{\mathrm{init}}({#1})}
\newcommand{\TCfinal}[1]{\TC_{\mathrm{final}}({#1})}

\newcommand{\lVIS}{{\mathtt{Vis}}}

\newcommand{\lCONT}{{\textrm{K}}}

\newcommand{\ea}{\mathtt{s2g}}

\newcommand{\lang}{\ensuremath{\mathsf{L}}\xspace}
\newcommand{\ES}{S}
\newcommand{\ESinit}{S_{\rm init}}

\newcommand{\ESgen}[1]{\ensuremath{\ES_{\mathsf{#1}}}\xspace}
\newcommand{\ESa}{\ESgen{a}}
\newcommand{\ESb}{\ESgen{b}}
\newcommand{\ESc}{\ESgen{c}}
\newcommand{\ESd}{\ESgen{d}}
\newcommand{\ESe}{\ESgen{e}}
\newcommand{\ESf}{\ESgen{f}}

\newcommand{\thread}[1]{{\mathtt{thread}(#1)}}

\newcommand{\fmap}[1]{{\llceil #1 \rrceil}}
\newcommand{\fcomap}[1]{{\llfloor #1 \rrfloor}}

\newcommand{\SX}{X}

\newcommand{\SXgen}[1]{\SX_{\mathsf{#1}}}
\newcommand{\SXinit}{\SXgen{\Init}}
\newcommand{\SXa}{\SXgen{a}}
\newcommand{\SXb}{\SXgen{b}}
\newcommand{\SXc}{\SXgen{c}}
\newcommand{\SXd}{\SXgen{d}}

\newcommand{\SXf}{\SXgen{f}}

\newcommand{\SBr}{Br}
\newcommand{\SBrgen}[1]{\SBr_{\mathsf{#1}}}

\newcommand{\SBrb}{\SBrgen{b}}
\newcommand{\SBrc}{\SBrgen{c}}

\newcommand{\Glb}{G_{\mathtt{LB}}}

\newcommand{\esstepcons}[1]{\xrightarrow{#1}}

\newcommand{\threadstep}[2]{\xrightarrow{#2}_{#1}}
\newcommand{\travstep}[1]{\longrightarrow_{#1}}

\lstdefinelanguage{Coq}{ 
mathescape=true,
texcl=false, 
morekeywords=[1]{Section, Module, End, Require, Import, Export,
  Variable, Variables, Parameter, Parameters, Axiom, Hypothesis,
  Hypotheses, Notation, Local, Tactic, Reserved, Scope, Open, Close,
  Bind, Delimit, Definition, Let, Ltac, Fixpoint, CoFixpoint, Add,
  Morphism, Relation, Implicit, Arguments, Unset, Contextual,
  Strict, Prenex, Implicits, Inductive, CoInductive, Record,
  Structure, Canonical, Coercion, Context, Class, Global, Instance,
  Program, Infix, Theorem, Lemma, Corollary, Proposition, Fact,
  Remark, Example, Proof, Goal, Save, Qed, Defined, Hint, Resolve,
  Rewrite, View, Search, Show, Print, Printing, All, Eval, Check,
  Projections, inside, outside, Def},
morekeywords=[2]{forall, exists, exists2, fun, fix, cofix, struct,
  match, with, end, as, in, return, let, if, is, then, else, for, of,
  nosimpl, when},
morekeywords=[3]{Type, Prop, Set, true, false, option},
morekeywords=[4]{pose, set, move, case, elim, apply, clear, hnf,
  intro, intros, generalize, rename, pattern, after, destruct,
  induction, using, refine, inversion, injection, rewrite, congr,
  unlock, compute, ring, field, fourier, replace, fold, unfold,
  change, cutrewrite, simpl, have, suff, wlog, suffices, without,
  loss, nat_norm, assert, cut, trivial, revert, bool_congr, nat_congr,
  symmetry, transitivity, auto, split, left, right, autorewrite},
morekeywords=[5]{by, done, exact, reflexivity, tauto, romega, omega,
  assumption, solve, contradiction, discriminate},
morekeywords=[6]{do, last, first, try, idtac, repeat},
morecomment=[s]{(*}{*)},
showstringspaces=false,
morestring=[b]",
morestring=[d],
tabsize=3,
extendedchars=false,
sensitive=true,
breaklines=false,
captionpos=b,
columns=[l]flexible,
identifierstyle={\ttfamily},
keywordstyle=[1]{\textbf},
keywordstyle=[2]{\textbf},
keywordstyle=[3]{\textbf},
keywordstyle=[4]{\textbf},
keywordstyle=[5]{\textbf},
stringstyle=\ttfamily,
commentstyle={\ttfamily},
literate=
    {forall}{{{$\forall\;$}}}1
    {exists}{{$\exists\;$}}1
    {<-}{{$\leftarrow\;$}}1
    {=>}{{$\Rightarrow\;$}}1
    {==}{{\code{==}\;}}1
    {==>}{{\code{==>}\;}}1
    {->}{{$\rightarrow\;$}}1
    {<->}{{$\leftrightarrow\;$}}1
    {<==}{{$\leq\;$}}1
    {\#}{{$^\star$}}1 
    {\\o}{{$\circ\;$}}1 
    {\@}{{$\cdot$}}1 
    {\/\\}{{$\wedge\;$}}1
    {\\\/}{{$\vee\;$}}1
    {++}{{\code{++}}}1
    {\@\@}{{$@$}}1
    {\\mapsto}{{$\mapsto\;$}}1
    {\\hline}{{\rule{\linewidth}{0.5pt}}}1
}[keywords,comments,strings]

\lstnewenvironment{coq}{\lstset{language=Coq}}{}

\newcommand{\cshift}{+ 0.1}
\newcommand{\dshift}{- 0.1}
\newcommand{\floor}[1]{\left\lfloor #1 \right\rfloor}

\title{Reconciling Event Structures with Modern Multiprocessors}

\author{Evgenii Moiseenko}{St. Petersburg University, Russia and JetBrains Research, Russia}{e.moiseenko@2012.spbu.ru}{}{}

\author{Anton Podkopaev}{National Research University Higher School of Economics, Russia and MPI-SWS, Germany and JetBrains Research, Russia}{podkopaev@mpi-sws.org}{}{}

\author{Ori Lahav}{Tel Aviv University, Israel}{orilahav@tau.ac.il}{}{}

\author{Orestis Melkonian}{University of Edinburgh, UK}{melkon.or@gmail.com}{}{}

\author{Viktor Vafeiadis}{MPI-SWS, Germany}{viktor@mpi-sws.org}{}{}

\authorrunning{E. Moiseenko et al.}
\Copyright{Evgenii Moiseenko, Anton Podkopaev, Ori Lahav, Orestis Melkonian, Viktor Vafeiadis}



\EventEditors{Robert Hirschfeld and Tobias Pape}
\EventNoEds{2}
\EventLongTitle{34th European Conference on Object-Oriented Programming (ECOOP 2020)}
\EventShortTitle{ECOOP 2020}
\EventAcronym{ECOOP}
\EventYear{2020}
\EventDate{July 13--17, 2020}
\EventLocation{Berlin, Germany}
\EventLogo{}
\SeriesVolume{166}
\ArticleNo{5}

\begin{document}

\maketitle              
\begin{abstract}
\weakestmo is a recently proposed memory consistency model
that uses event structures to resolve the infamous ``out-of-thin-air'' problem
and to enable efficient compilation to hardware.
Nevertheless, this latter property---compilation correctness---has not yet been
formally established.

This paper closes this gap by establishing correctness
of the intended compilation schemes from \weakestmo to a wide range
of formal hardware memory models (\Intel, \POWER, \ARMs, \ARMe)
in the Coq proof assistant.
Our proof is the first that establishes correctness of compilation of
an event-structure-based model that forbids ``out-of-thin-air'' behaviors,
as well as the first mechanized compilation proof of a weak memory model
supporting sequentially consistent accesses to such a range of hardware platforms.
Our compilation proof goes via the recent Intermediate Memory Model (\IMM),
which we suitably extend with sequentially consistent accesses.

\end{abstract}

\section{Introduction}
\label{sec:introduction}


A major research problem in concurrency semantics is to
develop a weak memory model that allows load-to-store reordering
(a.k.a.\ \emph{load buffering}, \textsf{LB})
and compiler optimizations (\eg elimination of fake dependencies),
while forbidding ``out-of-thin-air'' behaviors
\cite{Pichon-al:POPL16,Jeffrey-Riely:LICS16,Boehm-Demsky:MSPC14,Lahav-al:PLDI17}.

The problem can be illustrated with the following two programs,
which access locations $x$ and $y$ initialized to $0$.
The annotated outcome $a=b=1$ ought to be allowed for \ref{ex:LB-fake}
because $1 + a * 0$ can be optimized to $1$ and then the instructions
of thread 1 executed out of order.
In contrast, it should be forbidden for \ref{ex:LB-TA},
since no optimizations are applicable.
\begin{center}
\begin{minipage}{.5\linewidth}
{\small
\begin{equation}
\inarrII{
  \readInst{}{a}{x} \comment{1} \\
  \writeInst{}{y}{1 + a * 0} \\
}{\readInst{}{b}{y} \comment{1} \\
  \writeInst{}{x}{b}  \\
}
\tag{LB-fake}\label{ex:LB-fake}
\end{equation}
}
\end{minipage}
\hfill\vline\hfill
\begin{minipage}{.48\linewidth}
{\small
\begin{equation}
\inarrII{
  \readInst{}{a}{x} \nocomment{1} \\
  \writeInst{}{y}{a} \\
}{\readInst{}{b}{y} \nocomment{1} \\
  \writeInst{}{x}{b}  \\
}
\tag{LB-data}\label{ex:LB-TA}
\end{equation}
}
\end{minipage}
\end{center}


Among the proposed models that correctly distinguish between these two programs
is the recent \weakestmo model \cite{Chakraborty-Vafeiadis:POPL19}.
\weakestmo was developed in response to certain limitations of earlier models,
such as the ``promising semantics'' of Kang \etal~\cite{Kang-al:POPL17}, namely that
(\emph{i}) they did not cover the whole range of C/C++ concurrency features and that
(\emph{ii}) they did not support the intended compilation schemes to hardware.

Being flexible in its design, \weakestmo addresses the former point.
It supports all usual features of the C/C++11 model~\cite{Batty-al:POPL11}
and can easily be adapted to support any new concurrency features
that may be added in the future.
It does not, however, fully address the latter point.
Due to the difficulty of establishing correctness of the intended compilation schemes
to hardware architectures that permit load-store reordering (\ie \POWER, \ARMs, \ARMe),
Chakraborty and Vafeiadis~\cite{Chakraborty-Vafeiadis:POPL19}
only establish correctness of suboptimal schemes
that add (unnecessary) explicit fences to prevent load-store reordering.


In this paper, we address this major limitation of the \weakestmo paper.
We establish in Coq correctness of the intended compilation schemes
to a wide range of hardware architectures that includes the major ones:
\Intel-\TSO~\cite{Owens-al:TPHOL09}, \POWER~\cite{Alglave-al:TOPLAS14},
\ARMs~\cite{Alglave-al:TOPLAS14}, \ARMe~\cite{Pulte-al:POPL18}.
The compilation schemes, whose correctness we prove, do not require any fences
or fake dependencies for relaxed accesses.
Because of a technical limitation of our setup (see \cref{sec:related}), however,
compilation of read-modify-write (RMW) accesses to \ARMe uses a load-reserve/store-conditional loop
(similar to that of \ARMs and \POWER) as opposed to the newly introduced \ARMe
instructions for certain kinds of RMWs.


The main challenge in this proof is to reconcile the
different ways in which hardware models and \weakestmo allow load-store reordering.
Unlike most models at the programming language level, hardware models (such as
\ARMe) do not execute instructions in sequence; they instead keep track of
dependencies between instructions and ensure that no dependency cycles ever arise
in a single execution.
In contrast, \weakestmo executes instructions in order, but simultaneously considers
multiple executions to justify an execution where a load reads a value
that indirectly depends upon a later store.
Technically, these multiple executions together form an \emph{event structure},
upon which \weakestmo places various constraints.


\begin{wrapfigure}{r}{0.45\textwidth}\vspace{-10pt}\centering
\begin{tikzpicture}[auto,nd/.style={inner sep=2pt},lab/.style={inner sep=0pt}]
\node[nd] (ph) {\IMMsc};
\node[nd,below right=-1mm and 8mm of ph] (arm7) {\ARMs};
\node[nd,above right=-1mm and 8mm of ph] (pow)  {\POWER};
\node[nd,above=2.5mm of pow ] (tso)  {\Intel-\TSO};
\node[nd,below=2.5mm of arm7] (arm8) {\ARMe};
\node[nd,above left=2.5mm and 6mm of ph] (wmo) {\weakestmo};
\node[nd,below left=2.5mm and 6mm of ph] (rc11) {$\mathsf{C11}$};
\draw[->,out=  0,in=170,very thick] (wmo) edge (ph);
\draw[->,out=  0,in=190] (rc11) edge (ph);
\draw[->,out= 60,in=180] (ph) edge (tso);
\draw[->,out= 45,in=190] (ph) edge (tso);
\draw[->,out= 10,in=180] (ph) edge  (pow);
\draw[->,out= 4,in=190] (ph) edge  (pow);
\draw[->,out= -3,in=185] (ph) edge (arm7);
\draw[->,out= 3,in=175] (ph) edge (arm7);
\draw[->,out=-60,in=180] (ph) edge (arm8);
\end{tikzpicture}
\caption{Results proved in this paper.}
\label{fig:res}
\end{wrapfigure}
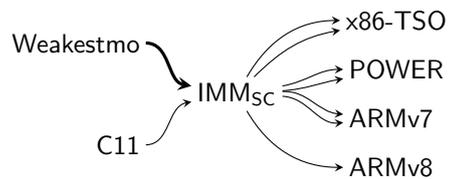

The high-level proof structure is shown in \cref{fig:res}.
We reuse \IMM, an \emph{intermediate memory model},
introduced by Podkopaev \etal~\cite{Podkopaev-al:POPL19}
as an abstraction over all major existing hardware memory models.
To support \weakestmo compilation,
we extend \IMM with \emph{sequentially consistent} (SC) accesses
following the \RC model~\cite{Lahav-al:PLDI17}.
As \IMM is very much a hardware-like model (\eg it tracks dependencies),
the main result is compilation from \weakestmo to \IMM (indicated by the bold arrow).
The other arrows in the figure are extensions of previous results
to account for SC accesses,
while double arrows indicate results for two compilation schemes.


The complexity of the proof is also evident from the size of the Coq development.
We have written about 30K lines of Coq definitions and proof scripts on top of
an existing infrastructure of about another 20K lines (defining \IMM, the
aforementioned hardware models and many lemmas about them).
As part of developing the proof, we also had to mechanize the \weakestmo
definition in Coq and to fix some minor deficiencies in the original definition,
which were revealed by our proof effort.


To the best of our knowledge,
our proof is the first proof of correctness of compilation of an event-structure-based memory model.
It is also the first mechanized compilation proof of a weak memory model
supporting sequentially consistent accesses to such a range of hardware architectures.
The latter, although fairly straightforward in our case,
has had a history of wrong compilation correctness arguments (see \cite{Lahav-al:PLDI17} for details).

\subparagraph*{Outline}
We start with an informal overview of \IMM, \weakestmo, and our compilation proof (\cref{sec:overview}).
We then present a fragment of \weakestmo formally (\cref{sec:weakestmo})
and its compilation proof (\cref{sec:wmoproof}).
Subsequently, we extend these results to cover SC accesses (\cref{sec:sc}),
discuss related work (\cref{sec:related}) and conclude (\cref{sec:conclusion}).
The associated proof scripts and supplementary material for our paper 
are publicly available at \url{http://plv.mpi-sws.org/weakestmoToImm/}. 

\section{Overview of the Compilation Correctness Proof}
\label{sec:overview}

To get an idea about the \IMM and \weakestmo memory models, consider a version of
the \ref{ex:LB-fake} and \ref{ex:LB-TA} programs from \cref{sec:introduction}
with no dependency in thread 1:
\begin{equation}
\inarrII{
  \readInst{}{a}{x} \comment{1} \\
  \writeInst{}{y}{1} \\
}{\readInst{}{b}{y} \comment{1} \\
  \writeInst{}{x}{b}  \\
}
\tag{LB}\label{ex:LB}
\end{equation}
As we will see, the annotated outcome is allowed by both \IMM and \weakestmo,
albeit in different ways.
The different treatment of load-store reordering affects the outcomes
of other programs.
For example, \IMM forbids the annotate outcome of \ref{ex:LB-fake}
by treating it exactly as \ref{ex:LB-TA},
whereas \weakestmo allows the outcome by treating \ref{ex:LB-fake} exactly as \ref{ex:LB}.

\subsection{An Informal Introduction to \IMM}
\label{sec:immscex}

\IMM is a \emph{declarative} (also called \emph{axiomatic}) model identifying
a program's semantics with a set of \emph{execution graphs}, or just \emph{executions}.
As an example, \cref{fig:lbWeak1} contains $\Glb$,
an \IMM execution graph of \ref{ex:LB}
corresponding to an execution yielding the annotated behavior.

\begin{figure}[t]
  \begin{subfigure}[b]{.48\textwidth}\centering
$\inarr{\begin{tikzpicture}[yscale=0.4,xscale=1]
  \node (init) at (1,  1.5) {$\Init$};


  \node (i11) at ( 0,  0) {$\erlab{}{x}{1}{}$};
  \node (i12) at ( 0, -2) {$\ewlab{}{y}{1}{}$};

  \node (i21) at ( 2,  0) {$\erlab{}{y}{1}{}$};
  \node (i22) at ( 2, -2) {$\ewlab{}{x}{1}{}$};

  \draw[po] (i11) edge node[left] {\small$\lPO$} (i12);
  \draw[po] (i21) edge node[left] {\small$\lPO$} (i22);
  \draw[ppo,bend left=10] (i21) edge node[right] {\small$\lPPO$} (i22);

  \draw[rf] (i22) edge node[below] {\small$\lRF$} (i11);
  \draw[rf] (i12) edge node[below] {} (i21);

  \draw[po] (init) edge node[left] {\small$\lPO$} (i11);
  \draw[po] (init) edge node[right] {\small$\lPO$} (i21);


\end{tikzpicture}}$
  \caption{$\Glb$: Execution graph of \ref{ex:LB}.}
  \label{fig:lbWeak1}
  \end{subfigure}\hfill
  \begin{subfigure}[b]{.48\textwidth}\centering
$\inarr{\begin{tikzpicture}[yscale=0.4,xscale=1]
  \node (init) at (1,  1.5) {$\Init$};

  \node (i11) at ( 0,  0) {$\erlab{}{x}{1}{}$};
  \node (i12) at ( 0, -2) {$\ewlab{}{y}{1}{}$};

  \node (i21) at ( 2,  0) {$\erlab{}{y}{1}{}$};
  \node (i22) at ( 2, -2) {$\ewlab{}{x}{1}{}$};

  \draw[po] (i11) edge node[right] {\small$\lPO$} (i12);
  \draw[po] (i21) edge node[left] {\small$\lPO$} (i22);
  \draw[ppo,bend right=10] (i11) edge node[left] {\small$\lPPO$} (i12);
  \draw[ppo,bend left=10] (i21) edge node[right] {\small$\lPPO$} (i22);

  \draw[rf] (i22) edge node[below] {\small$\lRF$} (i11);
  \draw[rf] (i12) edge node[below] {} (i21);

  \draw[po] (init) edge node[left] {\small$\lPO$} (i11);
  \draw[po] (init) edge node[right] {\small$\lPO$} (i21);
\end{tikzpicture}}$
  \caption{Execution of \ref{ex:LB-TA} and \ref{ex:LB-fake}.}
  \label{fig:lbWeak2}
  \end{subfigure}
\caption{Executions of \ref{ex:LB} and \ref{ex:LB-TA}/\ref{ex:LB-fake} with outcome $a = b = 1$.}
\label{fig:lbWeak}
\end{figure}

Vertices of execution graphs, called \emph{events},
represent memory accesses either due to the initialization of memory
or to the execution of program instructions.
Each event is labeled with the type of the access
(\eg $\lR$ for reads, $\lW$ for writes),
the location accessed, and the value read or written.
Memory initialization consists of a set of events labeled
$\ewlab{}{\loc}{0}{}$ for each location $\loc$ used in the program;
for conciseness, however, we depict the initialization events
as a single event with label $\Init$.

Edges of execution graphs represent different relations on events.
In \cref{fig:lbWeak}, three different relations are depicted.
The \emph{program order} relation ($\lPO$) totally orders events
originated from the same thread according to their order in the program,
as well as the initialization event(s) before all other events.
The \emph{reads-from} relation ($\lRF$) relates a write event
to the read events that read from it.
Finally, the \emph{preserved program order} ($\lPPO$) is a subset of the
program order relating events that cannot be executed out of order.
Such $\lPPO$ edges arise whenever there is a dependency chain between the
corresponding instructions (\eg a write storing the value read by a prior read).

Because of the syntactic nature of $\lPPO$,
\IMM conflates the executions of \ref{ex:LB-TA} and \ref{ex:LB-fake}
leading to the outcome $a=b=1$ (see \cref{fig:lbWeak2}).
This choice is in line with hardware memory models;
it means, however, that \IMM is not suitable as a memory model for a programming language
(because, as argued in \cref{sec:introduction},
\ref{ex:LB-fake} can be transformed to \ref{ex:LB} by an optimizing compiler).

The executions of a program are constructed in two steps.%
\footnote{For a detailed formal description of the graphs and
their construction process we refer the reader to
\cite[\S 2.2]{Podkopaev-al:POPL19}.}
First, a thread-local semantics determines the sequential executions of each thread,
where the values returned by each read access are chosen non-deterministically
(among the set of \emph{all} possible values),
and the executions of different threads are combined into a single execution.
Then, the execution graphs are filtered by a \emph{consistency predicate},
which determines which executions are allowed (i.e., are \IMM-consistent).
These \IMM-consistent executions form the program's semantics.

\IMM-consistency checks three basic constraints:
\begin{description}
\item[Completeness:] Every read event reads from precisely
one write with the same location and value;
\item[Coherence:] For each location $x$, there is a total ordering
of $x$-related events extending the program order so that each read of $x$
reads from the most recent prior write according to that total order; and
\item[Acyclic dependency:] There is no cycle consisting only of $\lPPO$ and $\lRF$ edges.
\end{description}
The final constraint disallows executions in which an event
recursively depends upon itself, as this pattern can lead to
``out-of-thin-air'' outcomes.
Specifically, the execution in \cref{fig:lbWeak2}, which represents the
annotated behavior of \ref{ex:LB-fake} and \ref{ex:LB-TA}, is \emph{not}
\IMM-consistent because of the $(\lPPO\cup\lRF)$-cycle.
In contrast, $\Glb$ is \IMM-consistent.

\subsection{An Informal Introduction to \weakestmo}
\label{sec:overview-weakestmo}

We move on to \weakestmo,
which also defines the program's semantics as a set of execution graphs.
However, they are constructed differently---extracted from a final \emph{event structure},
which \weakestmo incrementally builds for a program.

An event structure represents multiple executions of a programs in a single graph.
Like execution graphs, event structures contain a set of events and several relations among them.
Like execution graphs, the \emph{program order} ($\lPO$) orders events according to each thread's control flow.
However, unlike execution graphs, $\lPO$ is not necessarily total among the events of a given thread.
Events of the same thread that are not $\lPO$-ordered are said to be in \emph{conflict} ($\lCF$) with one another,
and cannot belong to the same execution.
Such conflict events arise when two read events originate from the same read instruction
(\eg representing executions where the reads return different values).
Moreover, $\lCF$ ``extends downwards'':
events that depend upon conflicting events (\ie have conflicting $\lPO$-predecessors)
are also in conflict with one other.
In pictures, we typically show only the \emph{immediate conflict} edges
(between reads originating from the same instruction)
and omit the conflict edges between events $\lPO$-after immediately conflicting ones.

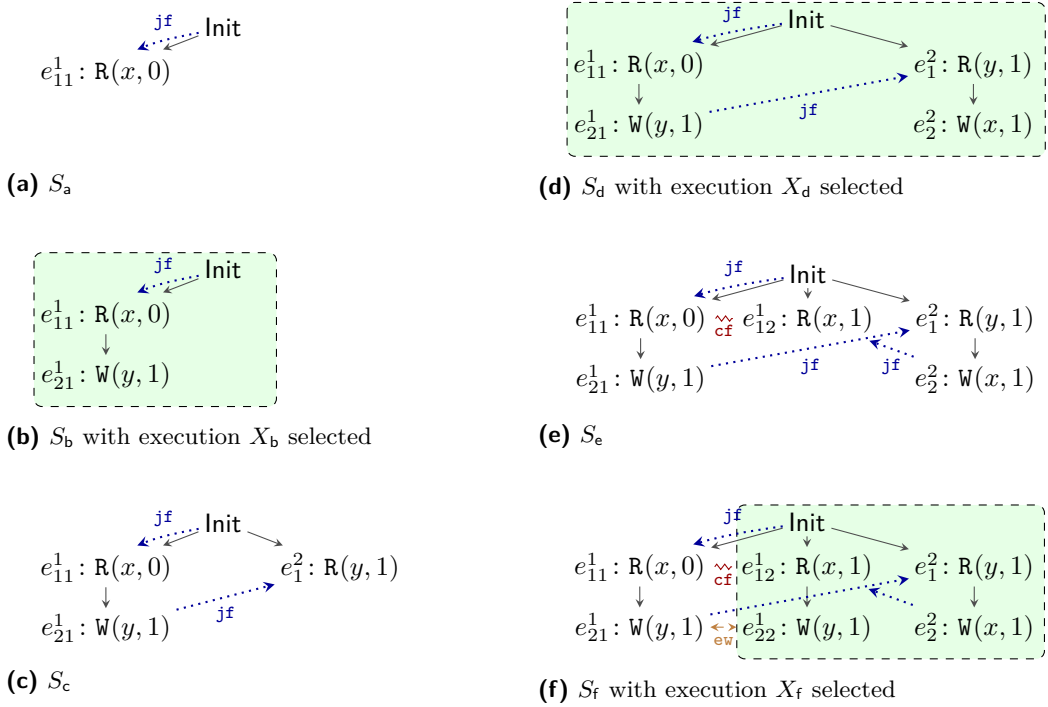
\begin{figure}[t]
\begin{minipage}{.4\textwidth}
  \begin{subfigure}[b]{\textwidth}\centering
\begin{tikzpicture}[yscale=0.4,xscale=2.2]
  \node (init) at (0, 1.5) {$\Init$};

  \node (i01) at (-0.7,  0) {$\mese{1}{1}{1} \erlab{}{x}{0}{}$};
  \node (i02) at (-0.7, -2) {\phantom{$\ewlab{}{y}{1}{}$}};
  \node (i31) at (0.7,  0) {\phantom{$\mese{2}{1}{} \erlab{}{y}{1}{}$}};

  \draw[jf,bend right=10] (init) edge node[above] {\scriptsize$\lJF$} (i01);
  \draw[po]   (init) edge node[right] {} (i01);
\end{tikzpicture}
  \caption{\ESa}
  \label{esc:lb1}
  \end{subfigure}
  \vspace{1ex}

  \begin{subfigure}[b]{\textwidth}\centering
\begin{tikzpicture}[yscale=0.4,xscale=2.2]
  \node (init) at (0, 1.5) {$\Init$};

  \node (i01) at (-0.7,  0) {$\mese{1}{1}{1} \erlab{}{x}{0}{}$};
  \node (i02) at (-0.7, -2) {$\mese{1}{2}{1} \ewlab{}{y}{1}{}$};
  \node (i31) at (0.7,  0) {\phantom{$\mese{2}{1}{} \erlab{}{y}{1}{}$}};

  \draw[po]   (i01) edge node[right] {} (i02);
  \draw[jf,bend right=10] (init) edge node[above] {\scriptsize$\lJF$} (i01);
  \draw[po]   (init) edge node[right] {} (i01);

  \begin{scope}[on background layer]
    \draw[extractStyle] ($(init)  + (-1.13,0.55)$) rectangle ++(1.45,-5.1);

    \phantom{\node at (-0.4, 1.3) {\large $\SXb$};}
  \end{scope}
\end{tikzpicture}
  \caption{\ESb with execution $\SXb$ selected}
  \label{esc:lb2}
  \end{subfigure}
  \vspace{1ex}

  \begin{subfigure}[b]{\textwidth}\centering
\begin{tikzpicture}[yscale=0.4,xscale=2.2]
  \node (init) at (0, 1.5) {$\Init$};

  \node (i01) at (-0.7,  0) {$\mese{1}{1}{1} \erlab{}{x}{0}{}$};
  \node (i02) at (-0.7, -2) {$\mese{1}{2}{1} \ewlab{}{y}{1}{}$};

  \node (i31) at (0.7,  0) {$\mese{2}{1}{} \erlab{}{y}{1}{}$};

  \draw[po] (i01) edge node[right] {} (i02);

  \draw[jf,bend right=10] (init) edge node[above] {\scriptsize$\lJF$} (i01);
  \draw[jf] (i02) edge node[below] {\scriptsize$\lJF$} (i31);

  \draw[po]   (init) edge node[right] {} (i01);
  \draw[po]   (init) edge node[right] {} (i31);
\end{tikzpicture}
  \caption{\ESc}
  \label{esc:lb3}
  \end{subfigure}
\end{minipage}
\hfill
\hfill
\begin{minipage}{.5\textwidth}
  \begin{subfigure}[b]{\textwidth}\centering
\begin{tikzpicture}[yscale=0.4,xscale=2.2]
  \node (init) at (0, 1.5) {$\Init$};

  \node (i01) at (-1,  0) {$\mese{1}{1}{1} \erlab{}{x}{0}{}$};
  \node (i02) at (-1, -2) {$\mese{1}{2}{1} \ewlab{}{y}{1}{}$};

  \node (i31) at ( 1,  0) {$\mese{2}{1}{} \erlab{}{y}{1}{}$};
  \node (i32) at ( 1, -2) {$\mese{2}{2}{} \ewlab{}{x}{1}{}$};

  \draw[po] (i01) edge node[right] {} (i02);
  \draw[po] (i31) edge node[right] {} (i32);

  \draw[jf,bend right=10] (init) edge node[above] {\scriptsize$\lJF$} (i01);
  \draw[jf] (i02) edge node[below] {\scriptsize$\lJF$} (i31);

  \draw[po]   (init) edge node[right] {} (i01);
  \draw[po]   (init) edge node[right] {} (i31);

  \begin{scope}[on background layer]
    \draw[extractStyle] ($(init)  + (-1.43,0.55)$) rectangle ++(2.86,-5.1);
  \end{scope}
\end{tikzpicture}
  \caption{\ESd with execution $\SXd$ selected}
  \label{esc:lb4}
  \end{subfigure}
  \vspace{1ex}

  \begin{subfigure}[b]{\textwidth}\centering
\begin{tikzpicture}[yscale=0.4,xscale=2.2]
  \node (init) at (0, 1.5) {$\Init$};

  \node (i01) at (-1,  0) {$\mese{1}{1}{1} \erlab{}{x}{0}{}$};
  \node (i02) at (-1, -2) {$\mese{1}{2}{1} \ewlab{}{y}{1}{}$};

  \node (i11) at ( 0,  0) {$\mese{1}{1}{2} \erlab{}{x}{1}{}$};

  \node (i31) at ( 1,  0) {$\mese{2}{1}{} \erlab{}{y}{1}{}$};
  \node (i32) at ( 1, -2) {$\mese{2}{2}{} \ewlab{}{x}{1}{}$};

  \draw[cf] (i01) -- (i11);
  \node at ($.5*(i01) + .5*(i11) - (0, 0.4)$) {\scriptsize$\lCF$};

  \phantom{\draw[ew] (i02) edge node[below] {\scriptsize$\lEW$} (i32);}

  \draw[po] (i01) edge node[right] {} (i02);
  \draw[po] (i31) edge node[right] {} (i32);

  \draw[jf] (i32) edge node[below] {\scriptsize$\lJF$} (i11);
  \draw[jf] (i02) edge node[below] {\scriptsize$\lJF$} (i31);

  \draw[jf,bend right=10] (init) edge node[above] {\scriptsize$\lJF$} (i01);


  \draw[po]   (init) edge node[right] {} (i01);
  \draw[po]   (init) edge node[right] {} (i11);
  \draw[po]   (init) edge node[right] {} (i31);
\end{tikzpicture}
  \caption{\ESe}
  \label{esc:lb5}
  \end{subfigure}
  \vspace{1ex}

  \begin{subfigure}[b]{\textwidth}\centering
\begin{tikzpicture}[yscale=0.4,xscale=2.2]
  \node (init) at (0, 1.5) {$\Init$};

  \node (i01) at (-1,  0) {$\mese{1}{1}{1} \erlab{}{x}{0}{}$};
  \node (i02) at (-1, -2) {$\mese{1}{2}{1} \ewlab{}{y}{1}{}$};

  \node (i11) at ( 0,  0) {$\mese{1}{1}{2} \erlab{}{x}{1}{}$};
  \node (i12) at ( 0, -2) {$\mese{1}{2}{2} \ewlab{}{y}{1}{}$};

  \node (i31) at ( 1,  0) {$\mese{2}{1}{} \erlab{}{y}{1}{}$};
  \node (i32) at ( 1, -2) {$\mese{2}{2}{} \ewlab{}{x}{1}{}$};

  \draw[cf] (i01) -- (i11);
  \node at ($.5*(i01) + .5*(i11) - (0, 0.4)$) {\scriptsize$\lCF$};

  \draw[ew] (i02) edge node[below] {\scriptsize$\lEW$} (i12);

  \draw[po] (i01) edge node[right] {} (i02);
  \draw[po] (i11) edge node[right] {} (i12);
  \draw[po] (i31) edge node[right] {} (i32);

  \draw[jf] (i32) edge node[above] {} (i11);
  \draw[jf] (i02) edge node[left] {} (i31);

  \draw[jf,bend right=10] (init) edge node[above] {\scriptsize$\lJF$} (i01);


  \draw[po]   (init) edge node[right] {} (i01);
  \draw[po]   (init) edge node[right] {} (i11);
  \draw[po]   (init) edge node[right] {} (i31);

  \begin{scope}[on background layer]
    \draw[extractStyle] ($(init)  + (-0.42,0.55)$) rectangle ++(1.84,-5.1);
  \end{scope}
\end{tikzpicture}
  \caption{\ESf with execution $\SXf$ selected}
  \label{esc:lb6}
  \end{subfigure}
\end{minipage}

  \caption{A run of \weakestmo witnessing the annotated outcome of \ref{ex:LB}.}
  \label{fig:esc}
\end{figure}

Event structures are constructed incrementally
starting from an event structure consisting only of the initialization events.
Then, events corresponding to the execution of program instructions are added one at a time.
We start by executing the first instruction of a program's thread.
Then, we may execute the second instruction of the same thread or the first instruction of another thread,
and so on.

As an example, \cref{fig:esc} constructs an event structure for \ref{ex:LB}.
\cref{esc:lb1} depicts the event structure $\ESa$ obtained from the initial event structure by executing $\readInst{}{a}{x}$ in~\ref{ex:LB}'s thread 1\@.
As a result of the instruction execution, a read event $\mese{1}{1}{1} \erlab{}{x}{0}{}$ is added.

Whenever the event added is a read, \weakestmo has to justify the returned value
from an appropriate write event.
In this case, there is only one write to $x$---the initialization write---and so
\ESa has a \emph{justified from} edge, denoted $\lJF$, going to $\ese{1}{1}{1}$ in $\ESa$.
This is a requirement of \weakestmo: each read event in an event structure
has to be justified from exactly one write event with the same value and location.
(This requirement is analogous to the \emph{completeness} requirement in $\IMM$-consistency for execution graphs.)
Since events are added in program order and
read events are always justified from existing events in the event structure,
$\lPO \cup \lJF$ is guaranteed to be acyclic by construction.

The next three steps (\cref{esc:lb2,esc:lb3,esc:lb4}) simply add a new event to the event structure.
Notice that unlike \IMM executions, \weakestmo event structures do not track syntactic dependencies,
\eg $\ESd$ in \cref{esc:lb4} does not contain a $\lPPO$ edge between $\ese{2}{1}{}$ and $\ese{2}{2}{}$.
This is precisely what allows \weakestmo to assign the same behavior to \ref{ex:LB} and \ref{ex:LB-fake}: they have exactly the same event structures.
As a programming-language-level memory model, \weakestmo supports optimizations removing fake dependencies.

The next step (\cref{esc:lb5}) is more interesting because it showcases the key
distinction between event structures and execution graphs, namely that
event structures may contain more than one execution for each thread.
Specifically, the transition from \ESd to \ESe reruns the first instruction of thread 1
and adds a new event $\ese{1}{1}{2}$ justified from a different write event.
We say that this new event \emph{conflicts} ($\lCF$) with $\ese{1}{1}{1}$
because they cannot both occur in a single execution.
Because of conflicts, $\lPO$ in event structures does not totally order all events of a thread;
\eg $\ese{1}{1}{1}$ and $\ese{1}{1}{2}$ are not $\lPO$-ordered in \ESe.
Two events of the same thread are conflicted precisely when they are not $\lPO$-ordered.

The final construction step (\cref{esc:lb6}) demonstrates another \weakestmo feature.
Conflicting write events writing the same value to the same location
(\eg $\ese{1}{2}{1}$ and $\ese{1}{2}{2}$ in $\ESf$)
may be declared \emph{equal writes}, \ie connected by an equivalence relation $\lEW$.%
\footnote{In this paper, we take $\lEW$ to be reflexive,
  whereas it is is irreflexive in Chakraborty and Vafeiadis~\cite{Chakraborty-Vafeiadis:POPL19}.
  Our $\lEW$ is the reflexive closure of the one in \cite{Chakraborty-Vafeiadis:POPL19}.}

The $\lEW$ relation is used to define \weakestmo's version of the reads-from relation, $\lRF$,
which relates a read to all (non-conflicted) writes \emph{equal} to the write justifying the read.
For example, $\ese{2}{1}{}$ reads from both $\ese{1}{2}{1}$ and $\ese{1}{2}{2}$.

The \weakestmo's $\lRF$ relation is used for extraction of program executions.
An execution graph $G$ is \emph{extracted} from an event structure $\ES$ denoted $\ES \rhd G$
if $G$ is a maximal conflict-free subset of $\ES$,
it contains only \emph{visible} events (to be defined in \cref{sec:weakestmo}),
and every read event in $G$ reads from some write in $G$ according to $\ES.\lRF$.
Two execution graphs can be extracted from $\ESf$:
$\set{\Init, \ese{1}{1}{1}, \ese{1}{2}{1}, \ese{2}{1}{}, \ese{2}{2}{}}$ and
$\set{\Init, \ese{1}{1}{2}, \ese{1}{2}{2}, \ese{2}{1}{}, \ese{2}{2}{}}$
representing the outcomes $a = 0 \land b = 1$ and $a = b = 1$ respectively.


\subsection{\weakestmo to \IMM Compilation: High-Level Proof Structure}
\label{sec:overview:proofstructure}

\newcommand{\lbShortExGraphTemplate}{
  \node (init) at (0, 1.8) {$\Init$};

  \node (i11) at (-1.1,  0) {$\ese{1}{1}{ }: \erlab{}{x}{1}{}$};
  \node (i12) at (-1.1, -2) {$\ese{1}{2}{ }: \ewlab{}{y}{1}{}$};

  \node (i21) at (1.1,  0) {$\ese{2}{1}{ }: \erlab{}{y}{1}{}$};
  \node (i22) at (1.1, -2) {$\ese{2}{2}{ }: \ewlab{}{x}{1}{}$};

  \draw[po] (init) edge node[left] {} (i11);
  \draw[po] (init) edge node[left] {} (i21);

  \draw[po] (i11) edge node[left] {} (i12);
  \draw[po] (i21) edge (i22);
  \draw[ppo,bend left=8] (i21) edge node[right] {\small$\lPPO$} (i22);

  \draw[rf] (i22) edge node[below] {\small$\lRF$} (i11);
  \draw[rf] (i12) edge node[below] {} (i21);
}
\newcommand{\lbShortExTemplate}[3]{
\begin{subfigure}[b]{.3\linewidth}
\centering
\begin{tikzpicture}[yscale=0.4,xscale=1]
    \lbShortExGraphTemplate

  \begin{scope}[on background layer]
    \nlIssuedCoveredBox{init};
    #3
  \end{scope}
\end{tikzpicture}
  \caption{#1}
  \label{#2}
\end{subfigure}
}

\begin{figure}[t]
    \lbShortExTemplate{\TCa}{tex:lb1}{\nlIssuedBox{i12};} \hfill
    \lbShortExTemplate{\TCb}{tex:lb2}{
      \nlIssuedBox{i12};
      \nlIssuedBox{i22};
    } \hfill
    \lbShortExTemplate{\TCc}{tex:lb3}{
      \nlIssuedBox{i12};
      \nlIssuedBox{i22};
      \nlCoveredBox{i11};
    } \\
    \rule{0pt}{2ex} \\
    \lbShortExTemplate{\TCd}{tex:lb4}{
      \nlIssuedBox{i12};
      \nlIssuedBox{i22};
      \nlCoveredDoubleBox{i11}{i12};
    } \hfill
    \lbShortExTemplate{\TCe}{tex:lb5}{
      \nlIssuedBox{i12};
      \nlIssuedBox{i22};
      \nlCoveredDoubleBox{i11}{i12};
      \nlCoveredBox{i21};
    } \hfill
    \lbShortExTemplate{\TCf}{tex:lb6}{
      \nlIssuedBox{i12};
      \nlIssuedBox{i22};
      \nlCoveredDoubleBox{i11}{i12};
      \nlCoveredDoubleBox{i21}{i22};
    }
\caption{Traversal configurations for $\Glb$.}
\label{fig:lbTravConfigs}
\end{figure}

In this paper, we assume that \weakestmo is defined for the same assembly language
as \IMM (see \cite[Fig.\ 2]{Podkopaev-al:POPL19}) extended with SC accesses and
refer to this language as \lang.
Having that, we show the correctness of the \emph{identity} mapping
as a compilation scheme from \weakestmo to \IMM
in the following theorem.
\begin{theorem}
  \label{thm:main}
  Let $\prog$ be a program in\/ \lang, and $G$ be an \IMM-consistent execution graph of $\prog$.
  Then there exists an event structure $\ES$ of\/ $\prog$ under \weakestmo such that
  $\ES \rhd G$.
\end{theorem}

To prove the theorem, we must show that \weakestmo may construct the needed event structure
in a step by step fashion.
If the \IMM-consistent execution graph $G$ contains no $\lPO \cup \lRF$ cycles,
then the construction is completely straightforward:
$G$ itself is a \weakestmo-consistent event structure (setting $\lJF$ to be just $\lRF$),
and its events can be added in any order extending $\lPO \cup \lRF$.

The construction becomes tricky for \IMM-consistent execution graphs,
such as $\Glb$, that contain $\lPO \cup \lRF$ cycles.
Due to the cycle(s), $G$ cannot be directly constructed as a (conflict-free)
\weakestmo event structure.
We must instead construct a larger event structure $S$ containing multiple executions,
one of which will be the desired graph $G$.
Roughly, for each $\lPO\cup\lRF$ cycle in $G$, we have to construct an immediate conflict
in the event structure.

To generate the event structure $S$,
we rely on a basic property of \IMM-consistent execution graphs
shown by Podkopaev \etal~\cite[\S\S6,7]{Podkopaev-al:POPL19},
namely that execution graphs can be \emph{traversed} in a certain order,
\ie its events can be \emph{issued} and \emph{covered} in that order,
so that in the end all events are covered.
The traversal captures a possible execution order of the program
that yields the given execution.
In that execution order, events are not added according to program order,
but rather according to \emph{preserved program order} ($\lPPO$)
in two steps.
Events are first issued when all their dependencies have been resolved,
and are later covered when all their $\lPO$-prior events have been covered.

In more detail, a traversal of an \IMM-consistent execution graph $G$ is a sequence of traversal steps between \emph{traversal configurations}.
A traversal configuration $\TC$ of an execution graph $G$ is a pair of sets of events,
$\tup{\CoveredSet,\IssuedSet}$, called the \emph{covered} and \emph{issued} set respectively.
As an example, \cref{fig:lbTravConfigs} presents all six traversal configurations
of the execution graph $\Glb$ of \ref{ex:LB} from \cref{fig:lbWeak1}
except for the initial configuration.
The issued set is marked by \smallIssuedBoxText{} and the covered set by \smallCoveredBoxText{}.

A traversal might be seen as an execution of an abstract machine
that can execute write instructions early
but has to execute everything else in order.
The first option corresponds to issuing a write event,
and the second option to covering an event.
The traversal strategy has certain constraints.
To issue a write event, all external reads that it depends upon
must be resolved; \ie they must read from already issued events.
To cover an event, all its $\lPO$-predecessors must also be covered.%
\footnote{For readers familiar with \Promise~\cite{Kang-al:POPL17},
issuing a write event corresponds to promising a message,
and covering an event to normal execution of an instruction.}
For example, in~\cref{fig:lbTravConfigs},
a traversal cannot issue $\mese{2}{2}{} \ewlab{}{\loc}{1}{}$ before issuing $\mese{1}{2}{} \ewlab{}{y}{1}{}$
nor cover $\mese{1}{1}{} \erlab{}{\loc}{1}{}$ before issuing $\mese{2}{2}{} \ewlab{}{\loc}{1}{}$.

According to Podkopaev \etal~\cite[Prop. 6.5]{Podkopaev-al:POPL19},
every \IMM-consistent execution graph $G$ has
a full traversal of the following form:
\begin{equation*}
G \vdash \TCinit{G} \travstep{} \TC_1 \travstep{}
\TC_2 \travstep{} \ldots \travstep{} \TCfinal{G}
\end{equation*}
where the initial configuration, $\TCinit{G} \defeq \tup{G.\lEo, G.\lEo}$,
has issued and covered only $G$'s initial events
and the final configuration, $\TCfinal{G} \defeq \tup{G.\lE, G.\lW}$,
has covered all $G$'s events and issued all its write events.

We construct the event structure $S$ following a full traversal of $G$.
We define a simulation relation, $\simrel(\prog, G, \TC, \ES, \SX)$,
between the program $\prog$,
the current traversal configuration $\TC$ of execution $G$
and the current event structure's state $\tup{\ES,\SX}$,
where $\SX$ is a subset of events corresponding to a particular execution graph extracted from the event structure $\ES$.

Our simulation proof is divided into the following three lemmas,
which state that the initial states are simulated,
that simulation extends along traversal steps,
and that the similation of final states means
that $G$ can be extracted from the generated event structure.

\begin{lemma}[Simulation Start]
  \label{lemma:simstart}
  Let $\prog$ be a program of \lang, and $G$ be an \IMM-consistent execution graph of $\prog$.
  Then $\simrel(\prog, G, \TCinit{G}, \ESinit(\prog), \ESinit(\prog).\lE)$ holds.
\end{lemma}
\begin{lemma}[Weak Simulation Step]
  \label{lemma:simstep}
  If $\simrel(\prog, G, \TC, \ES, \SX)$ and $G \vdash \TC \travstep{} \TC'$ hold,
  then there exist $\ES'$ and $\SX'$ such that
  $\simrel(\prog, G, \TC', \ES', \SX')$ and $\ES \esstepcons{}^* \ES'$ hold.
\end{lemma}
\begin{lemma}[Simulation End]
  \label{lemma:simend}
  If $\simrel(\prog, G, \TCfinal{G}, \ES, \SX)$ holds,
  then the execution graph associated with $\SX$ is isomorphic to $G$.
\end{lemma}

The proof of \cref{thm:main} then proceeds by induction on the length
of the traversal $G \vdash \TCinit{G} \travstep{}^* \TCfinal{G}$.
\Cref{lemma:simstart} serves as the base case,
\cref{lemma:simstep} is the induction step
simulating each traversal step with a number of event structure
construction steps, and \cref{lemma:simend} concludes the proof.

The proofs of \cref{lemma:simstart,lemma:simend} are technical
but fairly straightforward.
(We define $\simrel$ in a way that makes these lemmas immediate.)
In contrast, \cref{lemma:simstep} is much more difficult to prove.
As we will see, simulating a traversal step sometimes requires us
to construct a new branch in the event structure,
\ie to add multiple events (see \cref{sec:simstep}).

\subsection{\weakestmo to \IMM Compilation Correctness by Example}
\label{sec:overview:compile}

Before presenting any formal definitions, we conclude this overview section
by showcasing the construction used in the proof of \cref{lemma:simstep}
on execution graph $\Glb$ in \cref{fig:lbWeak1} following the traversal
of \cref{fig:lbTravConfigs}.
We have actually already seen the sequence of event structures constructed
in \cref{fig:esc}.
Note that, even though \cref{fig:lbTravConfigs,fig:esc} have the same number of steps,
there is no one-to-one correspondence between them as we explain below.

Consider the last event structure $\ESf$ from~\cref{fig:esc}.
A subset of its events $\SXf$ marked by \smallXBoxText{},
which we call a \emph{simulated execution},
is a maximal conflict-free subset of $\ESf$ and
all read events in $\SXf$ read from some write in $\SXf$
(\ie are justified from a write deemed ``equal'' to some write in $\SXf$).
Then, by definition, $\SXf$ is extracted from $\ESf$.
Also, an execution graph induced by $\SXf$ is isomorphic to $\Glb$.
That is, construction of $\ESf$ for \ref{ex:LB} shows that in \weakestmo it is possible to observe the same behavior as $\Glb$.
Now, we explain how we construct $\ESf$ and choose $\SXf$.

During the simulation, we maintain the relation $\simrel(\prog, G, \TC, \ES, \SX)$
connecting a program $\prog$, its execution graph $G$,
its traversal configuration $\TC$, an event structure $\ES$, and a subset of its events $\SX$.
Among other properties (presented in \cref{sec:simrel}),
the relation states that all issued and covered events of $\TC$ have exact counterparts in $\SX$,
and that $\SX$ can be extracted from $\ES$.

The initial event structure and $\SXinit$ consist of only initial events.
Then, following issuing of event $\mese{1}{2}{} \ewlab{}{y}{1}{}$ in~\TCa
(see~\cref{tex:lb1}),
we need to add a branch to the event structure that has $\ewlab{}{y}{1}{}$ in it.
Since \weakestmo requires adding events according to program order,
we first need to add a read event
corresponding to `$\readInst{}{a}{\loc}$' of~\ref{ex:LB}'s thread 1.
Each read event in an event structure has to be justified from somewhere.
In this case, the only write event to location $\loc$ is the initial one.
That is, the added read event $\ese{1}{1}{1}$ is justified from it (see~\cref{esc:lb1}).
In the general case, having more than one option, we would choose a `safe' write event
for an added read event to be justified from,
\ie the one which the corresponding branch is `aware' of already and
being justified from which would not break consistency of the event structure.
After that, a write event $\mese{1}{2}{1} \ewlab{}{y}{1}{}$ can be added $\lPO$-after $\ese{1}{1}{1}$ (see~\cref{esc:lb2}),
and $\simrel(\text{\ref{ex:LB}}, \Glb, \TCa, \ESb, \SXb)$ holds for
$\SXb = \set{\Init, \ese{1}{1}{1}, \ese{1}{2}{1}}$.

Next, we need to simulate the second traversal step (see~\cref{tex:lb2}), which issues $\ewlab{}{\loc}{1}{}$.
As with the previous step, we first need to add a read event
related to the first read instruction of~\ref{ex:LB}'s thread 2 (see~\cref{esc:lb3}).
However, unlike the previous step, the added event $\ese{2}{1}{}$ has to get value $1$, since
there is a dependency between instructions in thread 2.
As we mentioned earlier, the traversal strategy guarantees that
$\mese{1}{2}{} \ewlab{}{y}{1}{}$ is issued at the moment of issuing $\mese{2}{2}{} \ewlab{}{\loc}{1}{}$,
so there is the corresponding event in the event structure to justify the read event $\ese{2}{1}{}$ from.
Now, the write event $\mese{2}{2}{} \ewlab{}{y}{1}{}$ representing $\ese{2}{2}{}$ can be added to
the event structure (see~\cref{esc:lb4})
and $\simrel(\text{\ref{ex:LB}}, \Glb, \TCb, \ESd, \SXd)$ holds for
$\SXd = \set{\Init, \ese{1}{1}{1}, \ese{1}{2}{1}, \ese{2}{1}{}, \ese{2}{2}{}}$.

In the third traversal step (see~\cref{tex:lb3}), the read event $\mese{1}{1}{} \erlab{}{\loc}{1}{}$ is covered.
To have a representative event for $\ese{1}{1}{}$ in the event structure, we
add $\ese{1}{1}{2}$ (see~\cref{esc:lb5}).
It is justified from $\ese{2}{2}{}$, which writes the needed value $1$.
Also, $\ese{1}{1}{2}$ represents an alternative to $\ese{1}{1}{1}$ execution of the first instruction of thread 1,
so the events are in conflict.

However, we cannot choose a simulated execution $\SX$ related to $\TCc$ and $\ESe$ by
the simulation relation since $\SX$ has to contain $\ese{1}{1}{2}$ and
a representative for $\mese{1}{2}{} \ewlab{}{y}{1}{}$
(in $\ESe$ it is represented by $\ese{1}{2}{1}$)
while being conflict-free.
Thus, the event structure has to make one other step (see~\cref{esc:lb6})
and add the new event $\ese{1}{2}{2}$ to represent $\mese{1}{2}{} \ewlab{}{y}{1}{}$.
Now, the simulated execution contains everything needed,
$\SXf = \set{\Init, \ese{1}{1}{2}, \ese{1}{2}{2}, \ese{2}{1}{}, \ese{2}{2}{}}$.

Since $\SXf$ has to be extracted from $\ESf$,
every read event in $\SX$ has to be connected via an $\lRF$ edge to
an event in $\SX$.%
\footnote{Actually, it is easy to show that there could be only one such event
  since equal writes are in conflict and $\SX$ is conflict-free.
}
To preserve the requirement, we connect the newly added event
$\ese{1}{2}{2}$ and $\ese{1}{2}{1}$ via an $\lEW$ edge,
\ie marking them to be equal writes.%
\footnote{Note that we could have left $\ese{1}{2}{2}$ without any outgoing
$\lEW$ edges since the choice of equal writes for newly added events in \weakestmo
is non-deterministic. However, that would not preserve the simulation relation.
}
This induces an $\lRF$ edge between $\ese{1}{2}{2}$ and $\ese{2}{1}{}$.
That is, $\simrel(\text{\ref{ex:LB}}, \Glb, \TCc, \ESf, \SXf)$ holds.


To simulate the remaining traversal steps (\cref{tex:lb4,tex:lb5,tex:lb6}),
we do not need to modify $\ESf$ because it already contains counterparts
for the newly covered events and, moreover,
the execution graph associated with $\SXf$ is isomorphic to $\Glb$.
That is, we just need to show that
$\simrel(\text{\ref{ex:LB}}, \Glb, \TCd, \ESf, \SXf)$,
$\simrel(\text{\ref{ex:LB}}, \Glb, \TCe, \ESf, \SXf)$,
and $\simrel(\text{\ref{ex:LB}}, \Glb, \TCf, \ESf, \SXf)$ hold.

\section{Formal Definition of \weakestmo}
\label{sec:weakestmo}

In this section, we introduce the notation used in the rest of the paper
and define the \weakestmo memory model.  For simplicity, we present only a
minimal fragment of \weakestmo containing only relaxed reads and writes.
For the definition of the full \weakestmo model, we refer the readers to
Chakraborty and Vafeiadis~\cite{Chakraborty-Vafeiadis:POPL19} and to our
Coq development~\cite{appendix}.

  \subparagraph*{Notation}

  Given relations $R_1$ and $R_2$, we write $R_1 \seq R_2$ for their sequential composition.
  Given relation $R$, we write $R^?$, $R^+$ and $R^*$ to denote its
  reflexive, transitive and reflexive-transitive closures.
  We write $\mathtt{id}$ to denote the identity relation (\ie $\mathtt{id} \defeq \set{\tup{x, x}}$).
  For a set $A$, we write $[A]$ to denote the identity relation restricted to $A$ (that is, $[A]\defeq\set{\tup{a,a}\st a\in A}$).
  Hence, for instance, we may write $[A] \seq R\seq  [B]$ instead of $R\cap (A\times B)$.
  We also write $[e]$ to denote $[\set{e}]$ if $e$ is not a set.

  Given a function $f\colon A \to B$, we denote by ${=}_f$ the set of $f$-equivalent elements:
  (${=}_f \defeq \set{\tup{a,b} \in A \times A \st f(a)=f(b)}$).
  In addition, given a relation $R$, we denote by $R\rst{= f}$
  the restriction of $R$ to $f$-equivalent elements
  ($R\rst{= f} \defeq R \cap {=}_f$),
  and by $R\rst{\neq f}$ be the restriction of $R$ to non-$f$-equivalent elements
  ($R\rst{\neq f} \defeq R \setminus {=}_f$).

  \subsection{Events, Threads and Labels}

  \emph{Events}, $e\in\mathsf{Event}$, and \emph{thread identifiers}, $\tid\in\Tid$,
  are represented by natural numbers.
  We treat the thread with identifier $0$ as the \emph{initialization} thread.
  We let $x \in \Loc$ to range over \emph{locations},
  and $v \in \Val$ over \emph{values}.

%
%

  A label, $l \in \Lab$, takes one of the following forms:
  \begin{itemize}
    \item $\lR(x, v)$ --- a read of value $v$ from location $x$. 
    \item $\lW(x, v)$ --- a write of value $v$ to location $x$. 
  \end{itemize}
  Given a label $l$ the functions $\lTYP$, $\lLOC$, $\lVAL$ 
  return (when applicable) its type (\ie $\lR$ or $\lW$), 
  location and value correspondingly. 
  When a specific function assigning labels to events is clear from the context,
  we abuse the notations $\lR$ and $\lW$ 
  to denote the sets of all events labelled with the corresponding type.
  We also use subscripts 
  to further restrict this set to a specific location
  (\eg $\lW_{x}$ denotes the set of write events 
  operating on location $x$.) 

  \subsection{Event Structures}

  An \emph{event structure} $\ES$ is a tuple
  $\tup{\lE, \lTID, \lLAB, \lPO, \lJF, \lEW, \lCO}$ where:
  \begin{itemize}

    \item $\lE$ is a set of events, \ie $\lE \suq \mathsf{Event}$.

    \item $\lTID : \lE \fn \Tid$ is a function assigning a thread identifier to every event.
      We treat events with the thread identifier equal to $0$ as
      \emph{initialization events} and denote them as $\lEo$,
      that is $\lEo \defeq \set{e \in \lE \st \lTID(e) = 0}$.

    \item $\lLAB : \lE \fn \Lab$ is a function assigning a label to every event in $\lE$.

    \item $\lPO \subseteq \lE \times \lE$ is a strict partial order on events,
    called \emph{program order}, that tracks their precedence in the control flow of the program.
    Initialization events are $\lPO$-before all other events,
	whereas non-initialization events can only be $\lPO$-before events from the same thread.

    Not all events of a thread are necessarily ordered by $\lPO$.
    We call such $\lPO$-unordered non-initialization events of the same thread \emph{conflicting} events.
    The corresponding binary relation $\lCF$ is defined as follows:
    \begin{equation*}
      \lCF \defeq ([\lE \setminus \lEo] \seq {=_{\lTID}} \seq [\lE \setminus \lEo]) \setminus (\lPO \cup \lPO^{-1})^?
    \end{equation*}


%
    \item $\lJF \subseteq [\lE \cap \lW] \seq  ({=_{\lLOC}} \cap {=_{\lVAL}}) \seq [\lE \cap \lR]$ is the \emph{justified from} relation,
      which relates a write event to the reads it justifies.
      We require that reads are not justified by conflicting writes
      (\ie ${\lJF \cap \lCF = \emptyset}$)
      and $\lJF^{-1}$ be \emph{functional}
      (\ie whenever $\tup{w_1, r}, \tup{w_2, r} \in \lJF$, then $w_1 = w_2$).

	  We also define the notion of \emph{external} justification: $\lJFE \defeq \lJF \setminus \lPO$.
	  A read event is externally justified from a write if the write is
      not $\lPO$-before the read.

    \item $\lEW \subseteq [\lE \cap \lW] \seq (\lCF \cap {=_{\lLOC}} \cap {=_{\lVAL}})^? \seq [\lE \cap \lW]$ is
      an equivalence relation called the \emph{equal-writes} relation.
      Equal writes have the same location and value,
      and (unless identical) are in conflict with one another.

    \item $\lCO \subseteq [\lE \cap \lW]\seq ({=_{\lLOC}} \setminus \lEW) \seq [\lE \cap \lW]$ is
      the \emph{coherence} order, a strict partial order that relates
      non-equal write events with the same location.
      We require that coherence be closed with respect to equal writes
      (\ie $\lEW \seq \lCO \seq \lEW \subseteq \lCO$)
      and total with respect to $\lEW$ on writes to the same location:
      \begin{equation*}
\forall x \in \Loc \ldotp~ \forall w_1,w_2 \in \lW_{x} \ldotp~
         \tup{w_1, w_2} \in \lEW \cup \lCO \cup \lCO^{-1}
      \end{equation*}

  \end{itemize}

  Given an event structure $S$, we use ``dot notation'' to refer to its components (\eg $S.\lE$, $S.\lPO$).
  For a set $A$ of events, we write $S.A$ for the set $A\cap\ES.\lE$ 
  (for instance, $S.\lW_x = \set{e \in \ES.\lE \st \lTYP(S.\lLAB(e)) = \lW \land \lLOC(S.\lLAB(e)) = x}$).
  Further, for $e\in S.\lE$, we write $S.\lTYP(e)$ to retrieve $\lTYP(S.\lLAB(e))$.
  Similar notation is used for the functions $\lLOC$ and $\lVAL$. 
  Given a set of thread identifiers $T$, we write $S.\thread{T}$ to denote
  the set of events belonging to one of the threads in $T$,
  \ie $S.\thread{T} \defeq \set{e \in S.\lE \mid S.\lTID(e) \in T}$.
  When $T=\set{\thread{t}}$ is a singleton, we often write $S.\thread{t}$ instead of $S.\thread{\set{t}}$.

We define the immediate $\lPO$ and $\lCF$ edges of an event structure as follows:
\begin{equation*}
S.\lPOimm \defeq S.\lPO \setminus (S.\lPO \seq S.\lPO)       \qquad\qquad
S.\lCFimm \defeq S.\lCF \cap (S.\lPOimm^{-1} \seq S.\lPOimm)
\end{equation*}
An event $e_1$ is an immediate $\lPO$-predecessor of $e_2$ if
$e_1$ is $\lPO$-before $e_2$ and there is no event $\lPO$-between them.
Two conflicting events are immediately conflicting if they have the same immediate $\lPO$-predecessor.%
\footnote{Our definition of immediate conflicts differs from that of
  \cite{Chakraborty-Vafeiadis:POPL19} and is easier to work with.
  The two definitions are equivalent if the set of initialization events is non-empty.
}

\subsection{Event Structure Construction}

Given a program $\prog$, we construct its event structures operationally
in a way that guarantees completeness (\ie that every read is justified
from some write) and $\lPO \cup \lJF$ acyclicity.
We start with an event structure containing only the initialization events
and add one event at a time following each thread's semantics.

For the thread semantics,
we assume reductions of the form ${\pstate \threadstep{}{e} \pstate'}$
between thread states $\pstate, \pstate' \in \mathtt{ThreadState}$
and labeled by the event $e \in \lE$ generated by that execution step.
Given a thread $t$ and a sequence of events $e_1, \ldots, e_n \in \ES.\thread{t}$
in immediate $\lPO$ succession (\ie $\tup{e_i,e_{i+1}}\in \ES.\lPOimm$ for $1 \leq i < n$)
starting from a first event of thread $t$ (\ie $\dom{\ES.\lPO; [e_1]} \subseteq \lEo$),
we can add an event $e$ $\lPO$-after that sequence of events provided that
there exist thread states $\pstate_1, \ldots, \pstate_n$ and $\pstate'$ such that
$\prog(t) \threadstep{}{e_1} \pstate_1 \threadstep{}{e_2} \pstate_2 \cdots
 \threadstep{}{e_n} \pstate_n \threadstep{}{e} \pstate'$,
where $\prog(t)$ is the initial thread state of thread $t$ of the program $\prog$.
By construction, this means that the newly added event $e$ will be in conflict
with all other events of thread $t$ besides $e_1,\ldots,e_n$.

Further, when the new event $e$ is a read event,
it has to be justified from an existing write event,
so as to ensure completeness and prevent ``out-of-thin-air'' values.
The write event is picked non-deterministically from all non-conflicting writes
with the same location as the new read event.
Similarly, when $e$ is a write event, its position in $\lCO$ order should be chosen.
It can be done by either picking an $\lEW$ equivalence class and including the new write in it,
or by putting the new write immediately after some existing write in $\lCO$ order.
At each step, we also check for \emph{event structure consistency}
(to be defined in \cref{def:es-cons}):
If the event structure obtained after the addition of the new event is inconsistent,
it is discarded.

\subsection{Event Structure Consistency}

To define consistency, we first need a number of auxiliary definitions.
The \emph{happens-before} order $\ES.\lHB$ is a generalization of the program order.
Besides the program order edges, it includes certain \emph{synchronization} edges
(captured by the \emph{synchronizes with} relation, $\ES.\lSW$).
\begin{equation*}
  \ES.\lHB \defeq (\ES.\lPO \cup \ES.\lSW)^+
\end{equation*}
For the fragment covered in this section, there are no synchronization edges
(\ie $\lSW=\emptyset$), and so $\lHB$ and $\lPO$ coincide.
In the full model,%
\footnote{The full model is presented in \cite{Chakraborty-Vafeiadis:POPL19} and also in our Coq development~\cite{appendix}.}
however, certain justification edges (\eg between
release/acquire accesses) contribute to $\lSW$ and hence to $\lHB$.

The \emph{extended conflict} relation $\ES.\lECF$ extends the notion of conflicting events
to account for $\lHB$;
two events are in extended conflict if they happen after conflicting events.
\begin{equation*}
  \ES.\lECF \defeq (\ES.\lHB^{-1})^? \seq \ES.\lCF \seq \ES.\lHB^?
\end{equation*}

As already mentioned in \cref{sec:overview},
the \emph{reads-from} relation, $\ES.\lRF$, of a \weakestmo event structure is derived.
It is defined as an extension of $\ES.\lJF$ to all $\ES.\lEW$-equivalent writes.
\begin{equation*}
  \ES.\lRF \defeq (\ES.\lEW \seq S.\lJF) \setminus \ES.\lCF
\end{equation*}
Note that unlike $\ES.\lJF^{-1}$, the relation $\ES.\lRF^{-1}$ is not functional.
This does not cause any problems, however, since all the writes from whence a read
reads have the same location and value and are in conflict with one another.

The relation $\ES.\lFR$, called \emph{from-read} or \emph{reads-before},
places read events before subsequent writes.
\begin{equation*}
  \ES.\lFR   \defeq \ES.\lRF^{-1} \seq \ES.\lCO
\end{equation*}
The \emph{extended coherence} $S.\lECO$ is a strict partial order
that orders events operating on the same location.
(It is almost total on accesses to a given location, except that it
does not order equal writes nor reads reading from the same write.)
\begin{equation*}
  \ES.\lECO  \defeq (\ES.\lCO \cup \ES.\lRF \cup \ES.\lFR)^+
\end{equation*}
We observe that in our model, $\lECO$ is equal to $\lRF \cup \lCO \seq \lRF^? \cup \lFR \seq \lRF^?$,
similar to the corresponding definitions about execution graphs in the literature.%
\footnote{This equivalence equivalence does not hold in the original \weakestmo model
 \cite{Chakraborty-Vafeiadis:POPL19}.  To make the equivalence hold, we made $\lEW$ transitive,
 and require $\lEW \seq \lCO \seq \lEW \subseteq \lCO$.}


The last ingredient that we need for event structure consistency is the notion of \emph{visible} events,
which will be used to constrain external justifications.
We define it in a few steps.
Let $e$ be some event in $\ES$.
First, consider all write events used to externally justify $e$ or one of its justification ancestors.
The relation $\ES.\lJFE \seq (\ES.\lPO \cup \ES.\lJF)^*$ defines this connection formally.
Among that set of write events restrict attention to those conflicting with $e$, and call that set $M$.
That is, $M \defeq \dom{\ES.\lCF \cap (\ES.\lJFE \seq (\ES.\lPO \cup \ES.\lJF)^*) \seq [e]}$.
Event $e$ is \emph{visible} if all writes in $M$ have an equal write
that is $\lPO$-related with $e$. Formally,%
\footnote{Note, that in \cite{Chakraborty-Vafeiadis:POPL19} the definition of the visible events is slightly more verbose.
          We proved in Coq~\cite{appendix} that our simpler definition is equivalent to the one given there.}
\begin{equation*}
  \ES.\lVIS \defeq \set{e \in \ES.\lE \mid
  \ES.\lCF \cap (\ES.\lJFE \seq (\ES.\lPO \cup \ES.\lJF)^*) \seq [e] \subseteq \ES.\lEW \seq (\ES.\lPO \cup \ES.\lPO^{-1})^?}
\end{equation*}
Intuitively, visible events cannot depend on conflicting events: for every such
justification dependence, there ought to be an equal non-conflicting write.

\emph{Consistency} places a number of additional constraints on event structures.
First, it checks that there is no redundancy in the event structure:
immediate conflicts arise only because of read events justified from
non-equal writes.
Second, it extends the constraints about $\lCF$ to the extended conflict $\lECF$;
namely that no event can conflict with itself or be justified from a conflicting event.
Third, it checks that reads are justified either from events of the
same thread or from visible events of other threads.
Finally, it ensures \emph{coherence},
\ie that executions restricted to accesses on a single location
do not have any weak behaviors.

\begin{definition}
\label{def:es-cons}

An event structure $\ES$ is said to be \emph{consistent} if the following conditions hold.

\begin{itemize}
  \item $\dom{S.\lCFimm} \subseteq S.\lR$
    \labelAxiom{$\lCFimm$-read}{ax:icf-read}

  \item $S.\lJF \seq S.\lCFimm \seq S.\lJF^{-1} \seq S.\lEW$ is irreflexive.
    \labelAxiom{$\lCFimm$-justification}{ax:icf-jf}

  \item $S.\lECF$ is irreflexive.
   \labelAxiom{$\lECF$-irreflexivity}{ax:ecf-irr}

  \item $S.\lJF \cap S.\lECF = \emptyset$
   \labelAxiom{$\lJF$-non-conflict}{ax:jf-necf}

  \item $\dom{S.\lJFE} \subseteq S.\lVIS$
    \labelAxiom{$\lJFE$-visible}{ax:jf-vis}

  \item $S.\lHB \seq S.\lECO^?$ is irreflexive.
    \labelAxiom{coherence}{ax:es-coh}

\end{itemize}
\end{definition}

\subsection{Execution Extraction}
\label{sec:weakestmo:execution}

The last part of \weakestmo is the extraction of executions from an event structure.
An execution is essentially a conflict-free event structure.

\begin{definition}
\label{def:execution_graph}
  An \emph{execution graph} $G$ is a tuple
  $\tup{\lE, \lTID, \lLAB, \lPO, \lRF, \lCO}$
  where its components are defined similarly as in the case of an event structure with the following exceptions:
  \begin{itemize}
    \item $\lPO$ is required to be total on the set of events from the same thread.
    Thus, execution graphs have no conflicting events, \ie $\lCF = \emptyset$.

    \item The $\lRF$ relation is given explicitly instead of being derived.
    Also, there are no $\lJF$ and $\lEW$ relations.

    \item $\lCO$ totally orders write events operating on the same location.
  \end{itemize}
\end{definition}
  All derived relations are defined similarly as for event structures.
%
Next we show how to extract an execution graph from the event structure.
\begin{definition}
  A set of events $\SX$ is called \emph{extracted from $\ES$} if the following conditions are met:
  \begin{itemize}
    \label{def:extracted}

    \item $\SX$ is conflict-free, \ie $[\SX] \seq \ES.\lCF \seq [\SX] = \emptyset$.

    \item $\SX$ is $\ES.\lRF$-complete, \ie $\SX \cap \ES.\lR \subseteq \codom{[\SX] \seq \ES.\lRF}$.

    \item $\SX$ contains only visible events of $\ES$, \ie $\SX \subseteq \ES.\lVIS$.

    \item $\SX$ is $\lHB$-downward-closed, \ie $\dom{\ES.\lHB \seq [\SX]} \subseteq \SX$.

  \end{itemize}
\end{definition}

Given an event structure $\ES$ and extracted subset of its events $\SX$,
it is possible to associate with $\SX$ an execution graph $G$
simply by restricting the corresponding components of $\ES$ to $\SX$:
\[\def\arraystretch{1.3}
\begin{array}{l@{\;}c@{\;}l@{\qquad}l@{\;}c@{\;}l@{\qquad}l@{\;}c@{\;}l}
 G.\lE & = & \SX &
 G.\lTID & = & S.\lTID\rst{\SX} &
 G.\lLAB & = & S.\lLAB\rst{\SX} \\

 G.\lPO & = & [\SX] \seq \ES.\lPO \seq [\SX] & 
 G.\lRF & = & [\SX] \seq \ES.\lRF \seq [\SX] & G.\lCO & = & [\SX] \seq \ES.\lCO \seq [\SX]    \\
\end{array}
\]
We say that such execution graph $G$ is \emph{associated with} $\SX$
and that it is \emph{extracted} from the event structure: $S \rhd G$.


\weakestmo additionally defines another consistency predicate to further
filter out some of the extracted execution graphs.  In the \weakestmo fragment
we consider, this additional consistency predicate is trivial---every extracted
execution satisfies it---and so we do not present it here.
In the full model, execution consistency checks atomicity of read-modify-write
instructions, and sequential consistency for SC accesses.

%
%
%
%
%
%

\section{Compilation Proof for \weakestmo}
\label{sec:wmoproof}

In this section, we outline our correctness proof for the compilation from
\weakestmo to the various hardware models.
As already mentioned, our proof utilizes \IMM~\cite{Podkopaev-al:POPL19}.
In the following, we briefly present \IMM for the fragment of the model
containing only relaxed reads and writes (\cref{sec:imm}), our simulation
relation (\cref{sec:simrel}) for the compilation from \weakestmo to \IMM,
and outline the argument as to why the simulation relation is preserved
(\cref{sec:simstep}). 
Mapping from \IMM to the hardware models has already been proved correct
by Podkopaev \etal~\cite{Podkopaev-al:POPL19}, so we do not present
this part here.
Later, in \cref{sec:sc}, we will extend the \IMM mapping results to cover SC accesses.

As a further motivating example for this section consider yet another variant
of the load buffering program shown in \cref{fig:LBxyz}.
As we will see, its annotated weak behavior is allowed by \IMM and also by
\weakestmo, albeit in a different way.  The argument for constructing the
\weakestmo event structure that exhibits the weak behavior from the given
\IMM execution graph is non-trivial.

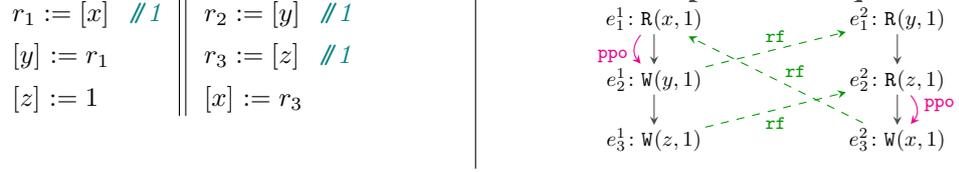
\begin{figure}[t]
\hfill$\inarrII{
  \readInst{}{r_1}{\locx} \comment{1} \\[1mm]
  \writeInst{}{\locy}{r_1}            \\[1mm]
  \writeInst{}{\locz}{1}              \\
}{
  \readInst{}{r_2}{\locy} \comment{1} \\[1mm]
  \readInst{}{r_3}{\locz} \comment{1} \\[1mm]
  \writeInst{}{\locx}{r_3}  \\
}$
\hfill\vrule\hfill
$\inarr{\begin{tikzpicture}[scale=0.8, every node/.style={transform shape}]
  \node (init) at (2,  1)  {$\Init$};
  \node (i11)  at (0,  0)   {$\ese{1}{1}{}\colon \erlab{}{\locx}{1}{}$};
  \node (i12)  at (0, -1)   {$\ese{1}{2}{}\colon \ewlab{}{\locy}{1}{}$};
  \node (i13)  at (0, -2)   {$\ese{1}{3}{}\colon \ewlab{}{\locz}{1}{}$};
  \node (i21)  at (4,  0)   {$\ese{2}{1}{}\colon \erlab{}{\locy}{1}{}$};
  \node (i22)  at (4, -1)   {$\ese{2}{2}{}\colon \erlab{}{\locz}{1}{}$};
  \node (i23)  at (4, -2)   {$\ese{2}{3}{}\colon \ewlab{}{\locx}{1}{}$};
  \draw[rf] (i13) edge node[below] {\small$\lRF$} (i22);
  \draw[rf] (i23) edge node[pos=.4,above] {\small\phantom{j}$\lRF$} (i11);
  \draw[rf] (i12) edge node[above] {\small$\lRF$} (i21);
  \draw[ppo,out=230,in=130] (i11) edge node[left ,pos=0.8] {\small$\lPPO$} (i12);
  \draw[ppo,out=310,in=50 ] (i22) edge node[right,pos=0.3] {\small$\lPPO$} (i23);
  \draw[po] (init) edge (i11);
  \draw[po] (init) edge (i21);
  \draw[po] (i11)  edge (i12);
  \draw[po] (i12)  edge (i13);
  \draw[po] (i21)  edge (i22);
  \draw[po] (i22)  edge (i23);
\end{tikzpicture}}$
\caption{A variant of the load-buffering program (left)
and the \IMM graph $G$ corresponding to its annotated weak behavior (right).}
\label{fig:LBxyz}
\end{figure}

\subsection{The Intermediate Memory Model \IMM}
\label{sec:imm}

In order to discuss the proof, we briefly present a simplified version of
the formal \IMM definition, where we have omitted constraints about RMW
accesses and fences.

\begin{definition}
  An \emph{\IMM execution graph} $G$ is an execution graph (\cref{def:execution_graph}) extended with one additional component:
  the \emph{preserved program order} $\lPPO \subseteq [\lR] \seq \lPO \seq [\lW]$.
\end{definition}

Preserved program order edges correspond to syntactic dependencies guaranteed
to be preserved by all major hardware platforms.
For example, the execution graph in \cref{fig:LBxyz} has two $\lPPO$ edges
corresponding to the data dependencies via registers $r_1$ and $r_3$.
(The full \IMM definition \cite{Podkopaev-al:POPL19} distinguishes between
the different types of dependencies---control, data, adress--and includes
them as separate components of execution graphs. In the full model, $\lPPO$
is actually derived from the more basic dependencies.)

\IMM-consistency checks completeness, coherence, and acyclicity:%
\footnote{Again, this is a simplified presentation for a fragment of the model.
We refer the reader to Podkopaev \etal~\cite{Podkopaev-al:POPL19} for the full definition, 
which further distinguishes between
internal and external $\lRF$ edges.}
\begin{definition}
\label{def:model_IMM}
An \IMM execution graph $G$ is \emph{\IMM-consistent} if
\begin{itemize}
\item $\codom{G.\lRF} = G.\lR$, \labelAxiom{completeness}{ax:complete}
\item $G.\lHB \seq G.\lECO^?$ is irreflexive, and \labelAxiom{coherence}{ax:coh}
\item $G.\lRF \cup G.\lPPO$ is acyclic. \labelAxiom{no-thin-air}{ax:nta}
\end{itemize}
\end{definition}

As we can see, the execution graph $G$ of \cref{fig:LBxyz} is \IMM-consistent
because every read of the graph reads from some write event
and, moreover, the \textsc{coherence} and \textsc{no-thin-air} properties hold.

\subsection{Simulation Relation for \weakestmo to \IMM Proof}
\label{sec:simrel}
In this section, we define the simulation relation $\simrel$
\footnote{A refined version of the simulation relation for the full \weakestmo model can be found in \citeapp{app:simrel}{Appendix A}}, 
which is used for the simulation of a traversal of an \IMM-consistent execution graph
by a \weakestmo event structure presented in \cref{sec:overview:proofstructure}.

The way we define $\simrel(\prog, G, \tup{\CoveredSet, \IssuedSet}, \ES, \SX)$ induces a strong connection between events in 
the execution graph $G$ and the event structure $\ES$.
We make this connection explicit with the function $\ea_{G, \ES} : \ES.\lE \fn G.\lE$,
which maps events of the event structure $S$ into the events of the execution graph $G$,
such that $e$ and $\ea_{G, \ES}(e)$ belong to the same thread and have the same $\lPO$-position in the thread.%
\footnote{Here we assume existence and uniqueness of such a function.
  In our Coq development~\cite{appendix}, 
  we have a different representation of execution graph events (but the same for events of event structures),
  which makes the existence and uniqueness questions trivial.

  More specifically, we follow Podkopaev \etal~\cite[\S 2.2]{Podkopaev-al:POPL19}.
  There each non-initializing event $e$ of an execution graph $G$ is encoded as a pair $\tup{t, n}$
  where $t$ is $e$'s thread and $n$ is a serial number of $e$ in thread $t$, \ie a position of $e$ in $G.\lPO$ restricted to events of thread $t$;
  each initializing event is encoded by the corresponding location---$\tup{{\sf init} \; l}$.

  In this representation, the function $\ea_{G, \ES}$ for an event $e$ returns (i) the $e$'s thread
  and a number of non-initial events which $\ES.\lPO$-preceded $e$ if $e$ is non-initialing
  or (ii) its location if it is initializing:
  $$
  \ea_{G, \ES}(e) \defeq
  \begin{cases}
    \tup{\ES.\lTID(e), |\dom{[\ES.\lE \setminus \ES.\lEo]; \ES.\lPO; [e]}|} & \text{for } e \nin \ES.\lEo \\
    \tup{{\sf init} \; \ES.\lLOC(e)} & \text{for } e \in \ES.\lEo
  \end{cases}
  $$
}
Note that $\ea_{G,\ES}$ is defined for all events $e\in \ES.\lE$,
meaning that the event structure $S$ does not contain any redundant events
that do not correspond to events in the \IMM execution graph $G$.
The function $\ea_{G, \ES}$, however, does not have to be injective:
in particular, events $e$ and $e'$ that are in immediate conflict in $\ES$
have the same $\ea_{G, \ES}$-image in $G$.
In the rest of the paper, whenever $G$ and $\ES$ are clear from the context,
we omit the $G, \ES$ subscript from $\ea$.

In the context of a function $\ea$ (for some $G$ and $\ES$), 
we also use $\fmap{\cdot}$ and $\fcomap{\cdot}$ to lift $\ea$ to sets and relations:
\begin{align*}
\text{for } A_\ES \subseteq \ES.\lE        & : 
  \fmap{A_\ES} \defeq \set{\ea(e) \mid e \in A_\ES} \\
\text{for } A_G \subseteq G.\lE        & : 
  \fcomap{A_G} \defeq \set{e \in \ES.\lE \mid \ea(e) \in A_G} \\[5pt]
\text{for } R_\ES \suq \ES.\lE \times \ES.\lE  & : 
  \fmap{R_\ES} \defeq \set{\tup{\ea(e), \ea(e')} \mid \tup{e, e'} \in R_\ES} \\
\text{for } R_G \suq G.\lE \times G.\lE  & : 
  \fcomap{R_G} \defeq \set{\tup{e, e'} \in \ES.\lE \times \ES.\lE \mid \tup{\ea(e), \ea(e')} \in R_G}
\end{align*}
For example, $\fcomap{\CoveredSet}$ denotes a subset of $\ES$'s events 
whose $\ea$-images are covered events in $G$,
and $\fmap{\ES.\lRF}$ denotes a relation on events in $G$
whose $\ea$-preimages in $\ES$ are related by $\ES.\lRF$.

We define the relation $\simrel(\prog, G, \tup{\CoveredSet, \IssuedSet}, \ES, \SX)$ to hold if the following conditions are met:%
\begin{enumerate}
  \item \label{simrel:gexec}
    $G$ is an \IMM-consistent execution of $\prog$.

  \item \label{simrel:sexec}
    $\ES$ is a \weakestmo-consistent event structure of $\prog$.

  \item \label{simrel:ex}
    $\SX$ is an extracted subset of $\ES$.

  \item \label{simrel:ex-cov-iss}
    $\ES$ and $\SX$ corresponds precisely to all covered and issued events and their $\lPO$-predecessors:
    \begin{itemize}
      \item $\fmap{\ES.\lE} = \fmap{\SX} = \CoveredSet \cup \dom{G.\lPO^? \seq [\IssuedSet]}$
    \end{itemize}
    (Note that $\CoveredSet$ is closed under $\lPO$-predecessors,
    so $\dom{G.\lPO^? \seq [\CoveredSet]} = \CoveredSet$.)

    
  \item \label{simrel:lab}
    Each $\ES$ event has the same thread, type, modifier, and location
    as its corresponding $G$ event.
    In addition, covered and issued events in $\SX$ have the same value
    as their corresponding ones in $G$.
    \begin{enumerate}
      \item $\forall e \in \ES.\lE. \; \ES.\set{\lTID, \lTYP, \lLOC, \lMOD}(e) = G.\set{\lTID, \lTYP, \lLOC, \lMOD}(\ea(e))$
      \item $\forall e \in \SX \cap \fcomap{C \cup I} \ldotp~ \ES.\lVAL(e) = G.\lVAL(\ea(e))$
    \end{enumerate}

  \item \label{simrel:po}
    Program order in $\ES$ corresponds to program order in $G$:
    \begin{itemize}
      \item $\fmap{\ES.\lPO} \subseteq G.\lPO$
    \end{itemize}

  \item \label{simrel:cf}
    Identity relation in $G$ corresponds to identity or conflict relation in $\ES$:
    \begin{itemize}
      \item $\fcomap{\mathtt{id}} \subseteq S.\lCF^?$
    \end{itemize}

  \item \label{simrel:jf}
    Reads in $\ES$ are justified by writes that have already been observed
    by the corresponding events in $G$.
    Moreover, covered events in $X$ are justified by a write corresponding
    to that read from the corresponding read in $G$:
    \begin{enumerate}
      \item \label{simrel:jf-obs}
        $\fmap{\ES.\lJF} \subseteq G.\lRF^?\seq G.\lHB^?$
      \item \label{simrel:jf-cov}
        $\fmap{\ES.\lJF \seq [\SX \cap \fcomap{C}]} \subseteq G.\lRF$
    \end{enumerate}

  \item \label{simrel:jfe-iss}
    Every write event justifying some external read event 
    should be $\ES.\lEW$-equal to some issued write event in $\SX$:
    \begin{itemize}
      \item $\dom{\ES.\lJFE} \subseteq \dom{\ES.\lEW \seq [\SX \cap \fcomap{I}]}$
    \end{itemize}

  \item \label{simrel:ew-id}
    Equal writes in $\ES$ correspond to the same write event in $G$:
    \begin{itemize}
      \item $\fmap{\ES.\lEW} \subseteq \mathtt{id}$
    \end{itemize}

  \item \label{simrel:ew-iss}
    Every non-trivial $\ES.\lEW$ equivalence class contains an issued write in $\SX$:
    \begin{itemize}
      \item $\ES.\lEW \subseteq (\ES.\lEW \seq [\SX \cap \fcomap{I}] \seq \ES.\lEW)^?$
    \end{itemize}

  \item \label{simrel:co}
    Coherence edges in $\ES$ correspond to coherence or identity edges in $G$.
    (We will explain in \cref{sec:simstep} why a coherence edge in $\ES$ might
    correspond to an identity edge in $G$.)
    \begin{itemize}
      \item $\fmap{\ES.\lCO} \subseteq G.\lCO^?$
    \end{itemize}
    


\end{enumerate}

\begin{figure}[t]
$\hfill\inarr{\begin{tikzpicture}[scale=0.8, every node/.style={transform shape}]
  \node (init) at (2,  1)   {$\Init$};
  \node (i11)  at (0,  0)   {$\ese{1}{1}{}\colon \erlab{}{\locx}{1}{}$};
  \node (i12)  at (0, -1)   {$\ese{1}{2}{}\colon \ewlab{}{\locy}{1}{}$};
  \node (i13)  at (0, -2)   {$\ese{1}{3}{}\colon \ewlab{}{\locz}{1}{}$};
  \node (i21)  at (4,  0)   {$\ese{2}{1}{}\colon \erlab{}{\locy}{1}{}$};
  \node (i22)  at (4, -1)   {$\ese{2}{2}{}\colon \erlab{}{\locz}{1}{}$};
  \node (i23)  at (4, -2)   {$\ese{2}{3}{}\colon \ewlab{}{\locx}{1}{}$};
  \node (hh) at (2, -3.5) {$\inarrC{\text{The execution graph } G \text{ and} \\\text{its traversal configuration } \TCa}$};
  \begin{scope}[on background layer]
     \issuedCoveredBox{init};
     \issuedBox{i13};
  \end{scope}
  \draw[rf] (i13) edge node[above] {} (i22);
  \draw[rf] (i23) edge node[above] {} (i11);
  \draw[rf] (i12) edge node[above] {} (i21);
  \draw[ppo,out=230,in=130] (i11) edge node[left ,pos=0.8] {\small$\lPPO$} (i12);
  \draw[ppo,out=310,in=50 ] (i22) edge node[right,pos=0.3] {\small$\lPPO$} (i23);
  \draw[po] (init) edge (i11);
  \draw[po] (init) edge (i21);
  \draw[po] (i11)  edge (i12);
  \draw[po] (i12)  edge (i13);
  \draw[po] (i21)  edge (i22);
  \draw[po] (i22)  edge (i23);
\end{tikzpicture}}
\hfill\vrule\hfill
\inarr{\begin{tikzpicture}[scale=0.8, every node/.style={transform shape}]

  \node (init) at (0, 1)   {$\Init$};

  \node (i111) at (-1.5,  0)   {$\ese{1}{1}{1}\colon \erlab{}{\locx}{0}{}$};
  \node (i121) at (-1.5, -1)   {$\ese{1}{2}{1}\colon \ewlab{}{\locy}{0}{}$};
  \node (i131) at (-1.5, -2)   {$\ese{1}{3}{1}\colon \ewlab{}{\locz}{1}{}$};

  \node (i211) at (0.5,  0)   {\phantom{$\ese{2}{1}{1}\colon \erlab{}{\locy}{0}{}$}};
  \node (i221) at (0.5, -1)   {\phantom{$\ese{2}{2}{1}\colon \erlab{}{\locz}{1}{}$}};
  \node (i231) at (0.5, -2)   {\phantom{$\ese{2}{3}{1}\colon \ewlab{}{\locx}{1}{}$}};

  \draw[jf] (init) edge[bend right] node[above]        {\small{$\lJF$}} (i111);

  \draw[po] (init)  edge (i111);
  \draw[po] (i111)  edge (i121);
  \draw[po] (i121)  edge (i131);

  \begin{scope}[on background layer]
    \draw[extractStyle] (-3, 1.5) rectangle (1,-2.5);
  \end{scope}

  \node (hh) at (0, -3.5) {$\inarrC{\text{The event structure } \ESa \text{ and} \\\text{the selected execution } \SXa}$};
\end{tikzpicture}}\hfill$
\caption{%
The execution graph $G$, 
its traversal configuration $\TCa$,
the related event structure $\ESa$,
and the selected execution $\SXa$.
Covered events are marked by
{\protect\tikz \protect\draw[coveredStyle] (0,0) rectangle ++(0.35,0.35);}
and issued ones by 
{\protect\tikz \protect\draw[issuedStyle] (0,0) rectangle ++(0.35,0.35);}.
Events belonging to the selected execution are marked by 
{\protect\tikz \protect\draw[extractStyle] (0,0) rectangle ++(0.35,0.35);}.
}
\label{fig:simrel-a}
\end{figure}

As an example, consider the execution $G$ from \cref{fig:LBxyz},
the traversal configuration $\TCa \defeq \tup{\set{\Init}, \set{\Init, \ese{1}{3}{}}}$,
and the event structure $\ESa$ shown in \cref{fig:simrel-a}.
We will show that $\simrel(\prog, G, \TCa, \ESa, \SXa)$, where $\SXa \defeq \ESa.\lE$, holds.

Take 
$\ea_{G,\ESa} = \set{ \Init \mapsto \Init, 
\ese{1}{1}{1} \mapsto \ese{1}{1}{},
\ese{1}{2}{1} \mapsto \ese{1}{2}{},
\ese{1}{3}{1} \mapsto \ese{1}{3}{} }$.
Given that $\lCF=\lEW=\emptyset$, the consistency constraints hold immediately.
For example, condition \ref{simrel:jf} holds because $\ese{1}{1}{1}$ is
justified by $\Init$, which happens before it.
Finally, note that only $\ese{1}{3}{1}$ and $\ese{1}{3}{}$
are required to have the same value by constraint \ref{simrel:lab},
the other related thread events only need to have the same type and address.

The definition of the simulation relation $\simrel$ renders the proofs of
\cref{lemma:simstart,lemma:simend} straightforward.
Specifically, for \cref{lemma:simstart},
the initial configuration $\TCinit{G}$ containing only the initialization events
is simulated by the initial event structure $\ESinit$
as all the constraints are trivially satisfied
($\ESinit.\lPO = \ESinit.\lJF = \ESinit.\lEW = \ESinit.\lCO = \emptyset$).

For \cref{lemma:simend}, since $\TCfinal{G}$ covers all events of $G$,
property \ref{simrel:lab} implies that
the labels of the events in $\SX$ are equal to the corresponding events of $G$;
property \ref{simrel:po} means that $\lPO$ is the same between them;
property \ref{simrel:jf} means that $\lRF$ is the same between them;
properties \ref{simrel:cf} and \ref{simrel:co} together mean that $\lCO$ is the same.
Therefore, $G$ and the execution corresponding to $\SX$ are isomorphic.

\subsection{Simulation Step Proof Outline}
\label{sec:simstep}

We next outline the proof of \cref{lemma:simstep},
which states that the simulation relation $\simrel$ can be restored after a traversal step.

Suppose that $\simrel(\prog, G, \TC, \ES, \SX)$ holds
for some $\prog$, $G$, $\TC$, $\ES$, and $\SX$, and we need to simulate a traversal step
$\TC \travstep{} \TC'$ that either covers or issues an event of thread $\tid$.
Then we need to produce an event structure $\ES'$ and a subset of its events $\SX'$
such that $\simrel(\prog, G, \TC', \ES', \SX')$ holds.
Whenever thread $t$ has any uncovered issued write events,
\weakestmo might need to take multiple steps from $\ES$ to $\ES'$
so as to add any missing events $\lPO$-before the uncovered issued writes of thread $t$.
Borrowing the terminology of the ``promising semantics''~\cite{Kang-al:POPL17},
we refer to these steps as constructing a certification branch for the issued write(s).

\begin{figure}[t]
$\hfill\inarr{\begin{tikzpicture}[scale=0.8, every node/.style={transform shape}]
  \node (init) at (2,  1)  {$\Init$};
  \node (i11)  at (0,  0)   {$\ese{1}{1}{}\colon \erlab{}{\locx}{1}{}$};
  \node (i12)  at (0, -1)   {$\ese{1}{2}{}\colon \ewlab{}{\locy}{1}{}$};
  \node (i13)  at (0, -2)   {$\ese{1}{3}{}\colon \ewlab{}{\locz}{1}{}$};
  \node (i21)  at (4,  0)   {$\ese{2}{1}{}\colon \erlab{}{\locy}{1}{}$};
  \node (i22)  at (4, -1)   {$\ese{2}{2}{}\colon \erlab{}{\locz}{1}{}$};
  \node (i23)  at (4, -2)   {$\ese{2}{3}{}\colon \ewlab{}{\locx}{1}{}$};
  \node (hh) at (2, -3.5) {$\inarrC{\text{The traversal configuration } \TCb}$};
  \begin{scope}[on background layer]
     \issuedCoveredBox{init};
     \issuedBox{i13};
     \issuedBox{i23};
  \end{scope}
  \draw[rf] (i13) edge node[above] {} (i22);
  \draw[rf] (i23) edge node[above] {} (i11);
  \draw[rf] (i12) edge node[above] {} (i21);
  \draw[vf] (init) edge[bend right=20]  node[above left, pos=0.9] {$\lVF$} (i11);
  \draw[vf] (init) edge[bend left=20]  node[above right, pos=0.9] {$\lVF$} (i21);
  \draw[vf] (i13)  edge[bend right=20] node[above] {$\lVF$} (i22);
  \draw[ppo,out=230,in=130] (i11) edge node[left ,pos=0.8] {\small$\lPPO$} (i12);
  \draw[ppo,out=310,in=50 ] (i22) edge node[right,pos=0.3] {\small$\lPPO$} (i23);
  \draw[po] (init) edge (i11);
  \draw[po] (init) edge (i21);
  \draw[po] (i11)  edge (i12);
  \draw[po] (i12)  edge (i13);
  \draw[po] (i21)  edge (i22);
  \draw[po] (i22)  edge (i23);
\end{tikzpicture}}
\hfill\vrule\hfill
\inarr{\begin{tikzpicture}[scale=0.8, every node/.style={transform shape}]

  \node (init) at (0, 1)   {$\Init$};

  \node (i111)  at (-1.5,  0)   {$\ese{1}{1}{1}\colon \erlab{}{\locx}{0}{}$};
  \node (i121)  at (-1.5, -1)   {$\ese{1}{2}{1}\colon \ewlab{}{\locy}{0}{}$};
  \node (i131)  at (-1.5, -2)   {$\ese{1}{3}{1}\colon \ewlab{}{\locz}{1}{}$};

  \node (i211)  at (1.5,  0)   {$\ese{2}{1}{1}\colon \erlab{}{\locy}{0}{}$};
  \node (i221)  at (1.5, -1)   {$\ese{2}{2}{1}\colon \erlab{}{\locz}{1}{}$};
  \node (i231)  at (1.5, -2)   {$\ese{2}{3}{1}\colon \ewlab{}{\locx}{1}{}$};

  \draw[jf] (init) edge[bend right] node[above]        {\small{$\lJF$}} (i111);
  \draw[jf] (init) edge[bend left ] node[above]        {\small{$\lJF$}} (i211);
  \draw[jf] (i131) edge             node[pos=.5,below] {\small{$\lJF$}} (i221);

  \draw[po] (init)  edge (i111);
  \draw[po] (i111)  edge (i121);
  \draw[po] (i121)  edge (i131);

  \draw[po] (init)  edge (i211);
  \draw[po] (i211)  edge (i221);
  \draw[po] (i221)  edge (i231);

  \begin{scope}[on background layer]
    \draw[extractStyle] (-3, 1.5) rectangle (3,-2.5);
  \end{scope}

  \node (hh) at (0, -3.5) {$\inarrC{\text{The event structure } \ESb \text{ and} \\\text{the selected execution } \SXb}$};
\end{tikzpicture}}\hfill$
\caption{%
The traversal configuration $\TCb$, the related event structure $\ESb$, and the selected execution $\SXb$.}
\label{fig:simrel-b}
\end{figure}

Before we present the construction, let us return to the example of \cref{fig:LBxyz}.
Consider the traversal step from configuration $\TCa$ to
configuration $\TC_b \defeq \tup{\set{\Init}, \set{\Init, \ese{1}{3}{}, \ese{2}{3}{}}}$
by issuing the event $\ese{2}{3}{}$ (see \cref{fig:simrel-b}).
To simulate this step, we need to show that
it is possible to execute instructions of thread 2 
and extend the event structure with a set of events $\SBrb$ matching these instructions.
As we have already seen, the labels of the new events can differ from their counterparts in $G$---%
they only have to agree for the covered and issued events.
In this case, we set $\SBrb=\set{\ese{2}{1}{1},\ese{2}{2}{1},\ese{2}{3}{1}}$,
and adding them to the event structure $\ESa$ gives us event structure $\ESb$
shown in \cref{fig:simrel-b}.

In more detail, we need to build a run of
thread-local semantics ${\prog(2)
\threadstep{}{\ese{2}{1}{1}}
\threadstep{}{\ese{2}{2}{1}}
\threadstep{}{\ese{2}{3}{1}}
\pstate'}$
such that
(1) it contains events corresponding to all the events of thread 2 up to $\ese{2}{3}{}$
(\ie $\ese{2}{1}{}, \ese{2}{2}{}, \ese{2}{3}{}$)
with the same location, type, and thread identifier
and (2) any events corresponding to covered or issued events (\ie $\ese{2}{3}{}$)
should also have the same value as the corresponding event in $G$.

Then, following the run of the thread-local semantics,
we should extend the event structure $\ESa$ to $\ESb$ by adding new events $\SBrb$,
and ensure that the constructed event structure $\ESb$ is consistent (\cref{def:es-cons})
and simulates the configuration $\TCb$.
In particular, it means that:
\begin{itemize}
  \item for each read event in $\SBrb$ we need to pick a justification write event,
        which is either already present in $\ES$ or $\lPO$-preceed the read event;
  \item for each write event in $\SBrb$ we should determine its position in $\lCO$ order of the event structure.
\end{itemize}
Finally, we need to update the selected execution
by replacing all events of thread 2 by the new events $\SBrb$: 
$\SXb \defeq \SXa \setminus \ES.\thread{\set{2}} \cup \SBrb$.


\subsubsection{Justifying the New Read Events}
\label{sec:simstep-read}

In order to determine whence these read events should be justified (and hence what value they should return),
we have adopted the approach of Podkopaev \etal~\cite{Podkopaev-al:POPL19}
for a similar problem with certifying promises in the compilation proof from \Promise to \IMM.
The construction relies on several auxiliary definitions.


First, given an execution $G$ and a traversal configuration $\tup{\CoveredSet, \IssuedSet}$,
we define the set of \emph{determined} events to be those events of $G$
that must have equal counterparts in $\ES$.
In particular, this means that $\ES$ should assign to these events the same label as $G$,
and thus the same reads-from source for the read events.
\begin{equation*}
  G.\determined_{\tup{\CoveredSet, \IssuedSet}}
    \defeq {} \CoveredSet \cup \IssuedSet \cup
              \dom{(G.\lRF \cap G.\lPO)^? \seq G.\lPPO \seq [\IssuedSet]} \cup
              \codom{[\IssuedSet] \seq (G.\lRF \cap G.\lPO)} 
\end{equation*}
Besides covered and issued events,
the set of determined events also contains
the $\lPPO$-prefixes of issued events, since issued events may depend on their values,
as well as any internal reads reading from issued events,
since their values are also determined by the issued events.

For the graph $G$ and traversal configuration $\TCb$,
the set of determined events contains events $\ese{1}{3}{}$, $\ese{2}{2}{}$, and $\ese{2}{3}{}$.
(The events $\ese{1}{3}{}$ and $\ese{2}{3}{}$ are issued,
whereas $\ese{2}{2}{}$ has a $\lPPO$ edge to $\ese{2}{3}{}$.)
In contrast, events $\ese{1}{1}{}$, $\ese{1}{2}{}$, and $\ese{2}{1}{}$ are not determined,
since their corresponding events in $S$ read/write a different value.


Second, we introduce the \emph{viewfront} relation ($\lVF$)
to contain all the writes that have been observed at a certain point in the graph.
That is, the edge $\tup{w, e} \in G.\lVF_{\TC}$ indicates that
the write $w$ either happens before $e$, is read by a covered event happening before $e$,
or is read by a determined read earlier in the same thread as $e$.
\begin{equation*}
  G.\lVF_{\tup{\CoveredSet,\IssuedSet}} \defeq {}
    [G.\lW] \seq (G.\lRF \seq [\CoveredSet])^? \seq G.\lHB^? \cup
    G.\lRF \seq [G.\determined_{\tup{\CoveredSet, \IssuedSet}}] \seq G.\lPO^?
\end{equation*}
\Cref{fig:simrel-b} depicts three $G.\lVF_{\TCb}$ edges.
Since ${G.\lVF_{\TC} \seq G.\lPO \subseteq G.\lVF_{\TC}}$,
the other incoming viewfront edges to thread 2 can be derived.
Note that there is no edge from $\ese{1}{2}{}$ to thread 2,
since $\ese{1}{2}{}$ neither happens before any event in thread 2
nor is read by any determined read.

Finally, we construct the \emph{stable justification} relation ($\lSRF$)
that helps us justify the read events in $\SBrb$ in the event structure:
\begin{equation*}
  G.\lSRF_{\TC} \defeq
    ([G.\lW] \seq (G.\lVF_{\TC} \cap {=}_{G.\lLOC}) \seq [G.\lR])
    \setminus (G.\lCO \seq G.\lVF_{\TC})
\end{equation*}
It relates a read event $r$ to the $\lCO$-last `observed' write event with same location.
Assuming that $G$ is \IMM-consistent, it can be shown that
$G.\lSRF$ agrees with $G.\lRF$ on the set of determined reads.
\begin{equation*}
  G.\lSRF_{\TC} ; [G.\determined_{\TC}] \subseteq G.\lRF
\end{equation*}
For the graph $G$ and traversal configuration $\TCb$ shown in \cref{fig:simrel-b}
the $\lSRF$ relation coincides with the depicted $\lVF$ edges: \ie we have
$\tup{\Init, \ese{1}{1}{}}, \tup{\Init, \ese{2}{1}{}},
 \tup{\ese{1}{3}{}, \ese{2}{2}{}} \in G.\lSRF_{\TCb}$.

Having $\lSRF_{\TCb}$ as a guide for values read by instructions in the certification run,
we construct the steps of the thread-local operational semantics
$\prog(2) \threadstep{}{}^* \pstate'$
using the receptiveness property of the thread's semantics,
which essentially says that given an execution trace $\tau=e_1,\ldots,e_n$ of the thread semantics,
and a subset of events $K\subseteq\set{e_1,\ldots,e_{n-1}}$ along that trace that have no $\lPPO$-successors in the graph,
we arbitrarily change the values of read events in $K$,
and there exist values for the write events in $K$ such that
the updated execution trace is also a trace of the thread semantics.%
\footnote{The formal definition of the receptiveness property is quite elaborate.
  For the detailed definition we refer the reader to the Coq development of \IMM~\cite{IMMrepo}. }

The relation $\lSRF_{\TCb}$ is also used to pick justification writes for the read events in $\SBrb$.
We have proved that each $\lSRF$ edge either starts in some issued event (of the previous traversal configuration)
or it connects two events that are related by $\lPO$:
\begin{equation*}
  G.\lSRF_{\TCb} \subseteq [\IssuedSet_{\mathsf{a}}] \seq G.\lSRF_{\TCb} \cup G.\lPO
\end{equation*}
In the former case, thanks to the property \ref{simrel:ex-cov-iss} of our simulation relation,
we can pick a write event from $\SX_{\mathsf{a}}$ corresponding to the issued write
(\eg for \cref{fig:simrel-b}, it is the event $\ese{1}{3}{1}$, corresponding to the issued write $\ese{1}{3}{}$).
In the latter case, we pick either the initial write or some $\ESb.\lPO$ preceding write belonging to $\SBrb$.

\subsubsection{Ordering the New Write Events}
\label{sec:simstep-write}

In order to pick the $S_b.\lCO$ position of the new write events in the updated event structure,
we generally follow the original $G.\lCO$ order of the \IMM graph.
Because of the conflicting events, however, it is not
always possible to preserve the inclusion between the relations.
This is why we relax the inclusion to $\fmap{\ES.\lCO} \subseteq G.\lCO^?$
in property \ref{simrel:co} of the simulation relation.

\begin{figure}[t]
\hfill$\inarr{\begin{tikzpicture}[scale=0.8, every node/.style={transform shape}]
  \node (init) at (2,  1)  {$\Init$};
  \node (i11)  at (0,  0)   {$\ese{1}{1}{}\colon \erlab{}{\locx}{1}{}$};
  \node (i12)  at (0, -1)   {$\ese{1}{2}{}\colon \ewlab{}{\locy}{1}{}$};
  \node (i13)  at (0, -2)   {$\ese{1}{3}{}\colon \ewlab{}{\locz}{1}{}$};
  \node (i21)  at (4,  0)   {$\ese{2}{1}{}\colon \erlab{}{\locy}{1}{}$};
  \node (i22)  at (4, -1)   {$\ese{2}{2}{}\colon \erlab{}{\locz}{1}{}$};
  \node (i23)  at (4, -2)   {$\ese{2}{3}{}\colon \ewlab{}{\locx}{1}{}$};
  \node (hh) at (2, -3) {$\inarrC{\text{The traversal configuration } \TCc}$};
  \begin{scope}[on background layer]
     \issuedCoveredBox{init};
     \issuedBox{i13};
     \issuedBox{i23};
     \coveredBox{i11};
  \end{scope}
  \draw[rf] (i13) edge node[above] {} (i22);
  \draw[rf] (i23) edge node[above] {} (i11);
  \draw[rf] (i12) edge node[above] {} (i21);
  \draw[ppo,out=230,in=130] (i11) edge node[left ,pos=0.8] {\small$\lPPO$} (i12);
  \draw[ppo,out=310,in=50 ] (i22) edge node[right,pos=0.3] {\small$\lPPO$} (i23);
  \draw[po] (init) edge (i11);
  \draw[po] (init) edge (i21);
  \draw[po] (i11)  edge (i12);
  \draw[po] (i12)  edge (i13);
  \draw[po] (i21)  edge (i22);
  \draw[po] (i22)  edge (i23);
\end{tikzpicture}}
\hfill\vrule\hfill
\inarr{\begin{tikzpicture}[scale=0.8, every node/.style={transform shape}]
  \node (init) at (3, 1)   {$\Init$};

  \node (i111)  at (0,  0)   {$\ese{1}{1}{1}\colon \erlab{}{\locx}{0}{}$};
  \node (i121)  at (0, -1)   {$\ese{1}{2}{1}\colon \ewlab{}{\locy}{0}{}$};
  \node (i131)  at (0, -2)   {$\ese{1}{3}{1}\colon \ewlab{}{\locz}{1}{}$};

  \node (i112)  at (3,  0)   {$\ese{1}{1}{2}\colon \erlab{}{\locx}{1}{}$};
  \node (i122)  at (3, -1)   {$\ese{1}{2}{2}\colon \ewlab{}{\locy}{1}{}$};
  \node (i132)  at (3, -2)   {$\ese{1}{3}{2}\colon \ewlab{}{\locz}{1}{}$};

  \node (i211)  at (6,  0)   {$\ese{2}{1}{1}\colon \erlab{}{\locy}{0}{}$};
  \node (i221)  at (6, -1)   {$\ese{2}{2}{1}\colon \erlab{}{\locz}{1}{}$};
  \node (i231)  at (6, -2)   {$\ese{2}{3}{1}\colon \ewlab{}{\locx}{1}{}$};

  \draw[jf] (init) edge[bend right] node[above]        {} (i111);
  \draw[jf] (init) edge[bend left ] node[above]        {} (i211);
  \draw[jf] (i131) edge             node[pos=.5,below] {} (i221);
  \draw[jf] (i231) edge             node[pos=.5,below] {} (i112);

  \draw[cf] (i111) -- (i112);
  \node at ($.5*(i111) + .5*(i112) - (0, 0.2)$) {\small$\lCF$};

  \draw[co] (i122) edge node[pos=.5,below] {\small$\lCO$} (i121);
  \draw[ew] (i131) edge node[pos=.5,below] {\small$\lEW$} (i132);

  \draw[po] (init)  edge (i111);
  \draw[po] (i111)  edge (i121);
  \draw[po] (i121)  edge (i131);

  \draw[po] (init)  edge (i112);
  \draw[po] (i112)  edge (i122);
  \draw[po] (i122)  edge (i132);

  \draw[po] (init)  edge (i211);
  \draw[po] (i211)  edge (i221);
  \draw[po] (i221)  edge (i231);

  \begin{scope}[on background layer]
    \draw[extractStyle] (2, 1.5) rectangle (7,-2.5);
  \end{scope}

  \node (hh) at (3, -3.5) {$\inarrC{\text{The event structure } \ESc \text{ and} \\\text{the selected execution } \SXc}$};
\end{tikzpicture}}$\hfill
\caption{The traversal configuration $\TCc$, the related event structure $\ESc$, and the selected execution $\SXc$.}
\label{fig:simrel-c}
\end{figure}

To see the problem let us return to the example.
Suppose that the next traversal step covers the read $\ese{1}{1}{}$.
To simulate this step, we build an event structure $\ESc$
(see \cref{fig:simrel-c}).
It contains the new events
$\SBrc \defeq \set{\ese{1}{1}{2}, \ese{1}{2}{2}, \ese{1}{3}{2}}$.

Consider the write events $\ese{1}{2}{1}$ and $\ese{1}{2}{2}$ of the event structure.
Since the events have different labels, we cannot make them $\lEW$-equivalent.
And since $\ESc.\lCO$ should be total among all writes to the same location (with respect to $\ESc.\lEW$),
we must put a $\lCO$ edge between these two events in one direction or another.
Note that events $\ese{1}{2}{1}$ and $\ese{1}{2}{2}$ correspond to the same event $\ese{1}{2}{}$ in the graph,
thus we cannot use the coherence order of the graph $G.\lCO$ to guide our decision.

In fact, the $\lCO$-order between these two events does not matter,
so we could pick either direction.
For the purposes of our proofs, however, we found it more convenient
to always put the new events earlier in the $\lCO$ order
(thus we have $\tup{\ese{1}{2}{2}, \ese{1}{2}{1}} \in \ESc.\lCO$).
Thereby we can show that the $\lCO$ edges of the event structure ending in
the new events, have corresponding edges in the graph:
$\fmap{\ESc.\lCO \seq [\SBrc]} \subseteq G.\lCO$.

Now consider the events $\ese{1}{3}{1}$ and $\ese{1}{3}{2}$.
Since these events have the same label and correspond to the same event in $G$,
we make them $\lEW$-equivalent.
In fact, this choice is necessary for the correctness of our construction.
Otherwise, the new events $\SBrc$ would be deemed invisible,
because of the $\ESc.\lCF \cap (\ESc.\lJFE \seq (\ESc.\lPO \cup \ESc.\lJF)^*)$
path between $\ese{1}{3}{1}$ and $\ese{1}{1}{2}$.
Recall that only the visible events can be used to extract an execution
from the event structure (\cref{def:extracted}).

In general, assuming that $\simrel(\prog, G, \tup{\CoveredSet, \IssuedSet}, \ES, \SX)$ holds,
we attach the new write event $e$ to an $\ES.\lEW$ equivalence class represented
by the write event $w$, \sth (i) $w$ has the same $\ea$ image as $e$, \ie $\ea(w) = \ea(e)$;
(ii) $w$ belongs to $\SX$ and its $\ea$ image is issued, that is $w \in \SX \cap \fcomap{I}$.
If there is no such an event $w$,
we put $e$ $\ES.\lCO$-after events such that
their $\ea$ images are ordered $G.\lCO$-before $\ea(e)$,
and $\ES.\lCO$-before events such that
their $\ea$ images are equal to $\ea(e)$ or ordered $G.\lCO$-after it.
Note that thanks to property \ref{simrel:jfe-iss} of the simulation relation,
that is $\dom{\ES.\lJFE} \subseteq \dom{\ES.\lEW \seq [\SX \cap \fcomap{I}]}$,
our choice of $\lEW$ guarantees that all new events will be visible.


\subsubsection{Construction Overview}

To sum up, to prove \cref{lemma:simstep},
we consider the events of $G.\thread{\set{\tid}}$ where $\tid$ is the thread
of the event issued or covered by the traversal step $\TC \travstep{} \TC'$,
together with the $\lSRF$ relation determining the values of the read events.
At this point, we can show that $\simrel$-conditions for the new configuration
$\TC'$ hold for all events except for those in thread $\tid$.

Because of receptiveness, there exists a sequence of the thread steps
$\prog(\tid) \threadstep{}{}^* \pstate'$ for some thread state $\pstate'$
such that the labels on this sequence match the events $G.\thread{\set{\tid}}$
with the labels determined by $\lSRF$, and include an event with the same label
as the one issued or covered by the traversal step $\TC \travstep{} \TC'$.

We then do an induction on this sequence of steps, and add each event
to the event structure $\ES$ and to its selected subset of events $\SX$
(unless already there), showing along the way that the $\simrel$-conditions
also hold for the updated event structure, selected subset, and the events added.
At the end, when we have considered all the events generated by the step sequence,
we will have generated the event structure $\ES'$ and execution $\SX'$
such that $\simrel(\prog,G,\TC',\ES',\SX')$ holds.

\section{Handling SC Accesses}
\label{sec:sc}

In this section, we briefly describe the changes needed in order to handle the
compilation of \weakestmo's sequentially consistent (SC) accesses.
The purpose of SC accesses is to guarantee sequential consistency
for the simple programming pattern that uses exclusively SC accesses to communicate between threads.
As Lahav \etal~\cite{Lahav-al:PLDI17} showed, however,
their semantics is quite complicated
because they can be freely mixed with non-SC accesses.

We first define an extension of \IMM, which we call \IMMsc.
Its consistency extends that of \IMM with an additional acyclicity requirement
concerning SC accesses, which is taken directly from
\RC-consistency~\cite[Definition 1]{Lahav-al:PLDI17}.
\begin{definition}
\label{def:model_IMM_SC}
An execution graph $G$ is \emph{\IMMsc-consistent} if
it is \IMM-consistent \cite[Definition 3.11]{Podkopaev-al:POPL19}
and $G.\lPSCB \cup G.\lPSCF$ is acyclic, where:%
\footnote{In \IMMsc, event labels include an ``access mode'', where $\sco$ denotes an SC access.
The sets $G.\lE^\sco$ consists of all SC accesses (reads, writes and fences) in $G$,
and $G.\lF^\sco$ consists of all SC fences in $G$.}
    \begin{align*}
      G.\lSCB & \defeq
                G.\lPO \cup
                G.\lPO\rst{\neq G.\lLOC} \seq G.\lHB \seq G.\lPO\rst{\neq G.\lLOC} \cup
                G.\lHB\rst{= \lLOC} \cup
                G.\lCO \cup G.\lFR \\
      G.\lPSCB &\defeq
                 ([G.\lE^\sco] \cup{} [G.\lF^\sco]\seq G.\lHB^?) \seq G.\lSCB \seq ([G.\lE^\sco] \cup{} G.\lHB^?\seq[G.\lF^\sco]) \\
      G.\lPSCF &\defeq [G.\lF^\sco];(G.\lHB \cup G.\lHB;G.\lECO;G.\lHB);[G.\lF^\sco]
    \end{align*}
\end{definition}

The $\lSCB$, $\lPSCB$ and $\lPSCF$ relations were carefully designed
by Lahav \etal~\cite{Lahav-al:PLDI17} (and recently adopted by the C++ standard), 
so that they provide strong enough guarantees for programmers
while being weak enough to support the intended compilation of SC accesses to commodity hardware.
In particular, a previous (simpler) proposal in~\cite{sc-atomics}, which essentially includes $G.\lHB$ between SC accesses in the 
relation required to be acyclic, is too strong for efficient compilation to the POWER architecture.
Indeed, the compilation schemes to POWER do not enforce a strong barrier on $\lHB$-paths between SC accesses, 
but rather on $G.\lPO \seq G.\lHB \seq G.\lPO$-paths between SC accesses.

\begin{remark}
The full \IMM model (\ie including release/acquire accesses and SC fences,
as defined by Podkopaev \etal~\cite{Podkopaev-al:POPL19})
forbids cycles in $\lRFE \cup \lPPO \cup \lBOB \cup \lPSCF$,
where $\lBOB$ is (similar to $\lPPO$) a subset of the program order
that must be preserved due to the presence of a memory fence or release/acquire access.
Since $\lPSCF$ is already included in \IMM's acyclicity constraint,
one may consider the natural option of including $\lPSCB$ in that acyclicity constraint as well.
However, it leads to a model that is too strong, as it forbids the following behavior:
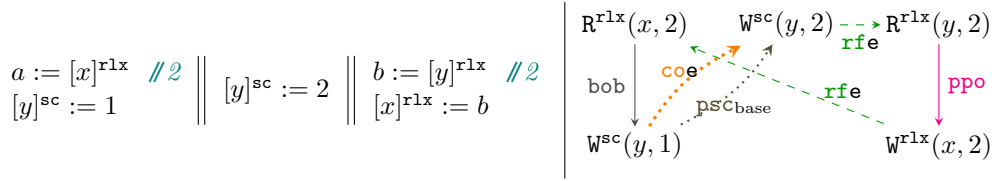
\[
\inarrIII{
  \readInst{\rlx}{a}{x} \comment{2} \\
  \writeInst{\sco}{y}{1}
}{
  \writeInst{\sco}{y}{2}
}{
  \readInst{\rlx}{b}{y} \comment{2} \\
  \writeInst{\rlx}{x}{b} \\
}
\quad\vrule\;
\inarr{\begin{tikzpicture}[yscale=0.8,xscale=1]
  \node (i11) at (0,  0) {$\erlab{\rlx}{x}{2}{}$};
  \node (i12) at (0, -2) {$\ewlab{\sco}{y}{1}{}$};

  \node (i21) at (2,  0) {$\ewlab{\sco}{y}{2}{}$};

  \node (i31) at (4,  0) {$\erlab{\rlx}{y}{2}{}$};
  \node (i32) at (4, -2) {$\ewlab{\rlx}{x}{2}{}$};

  \draw[po] (i11) edge node[left] {\small$\lBOB$} (i12);

  \draw[mo,bend left =15] (i12) edge node[above,pos=.4] {\small$\lCOE$} (i21);
  \draw[sc,bend right=15] (i12) edge node[below] {\small$\lPSCB$} (i21);

  \draw[rf] (i21) edge node[below] {\small$\lRFE$} (i31);
  \draw[ppo] (i31) edge node[right] {\small$\lPPO$} (i32);

  \draw[rf] (i32) edge node[above,pos=.2] {\small$\lRFE$} (i11);
\end{tikzpicture}}
\]
This behavior is allowed by \POWER (using any of the two intended compilation schemes
for SC accesses; see \cref{sec:power}).
\end{remark}

Adapting the compilation from \weakestmo to \IMMsc to cover SC accesses
is straightforward because the full definition of \weakestmo~\cite{Chakraborty-Vafeiadis:POPL19}
does not have any additional constraints about SC accesses at the level of event structures.
It only has an SC constraint at the level of extracted executions
which is actually the same as in \RC, which we took as is for \IMMsc.

\subsection{Compiling \texorpdfstring{\IMMsc}{IMM-SC} to Hardware}

In this section, we establish describe the extension of the results of~\cite{Podkopaev-al:POPL19}
to support SC accesses with their intended compilation schemes to the different architectures.

As was done in~\cite{Podkopaev-al:POPL19}, 
since \IMMsc and the models of hardware we consider are all defined in the same declarative framework
(using execution graphs), we formulate our results
on the level of execution graphs.
Thus, we actually consider the mapping of \IMMsc execution graphs to target architecture execution graphs
that is induced by compilation of \IMMsc programs to machine programs.
Hence, roughly speaking, for each architecture $\alpha\in\set{\TSO,\POWER,\ARMs,\ARMe}$, 
our (mechanized) result takes the following form:
\begin{quote}
If the $\alpha$-execution-graph $G_\alpha$ corresponds to the $\IMMsc$-execution-graph $G$,
then $\alpha$-consistency of $G_\alpha$ implies $\IMMsc$-consistency of $G$.
\end{quote}
Since the mapping from \weakestmo to \IMMsc (on the program level) is the \emph{identity mapping}
(\cref{thm:main}), we obtain as a corollary the correctness of the compilation from \weakestmo
to each architecture $\alpha$ that we consider.
The exact notions of correspondence between $G_\alpha$ and $G$
are presented in~\citeapp{app:imm_to_arm,,app:imm_to_tso,,app:imm_to_power}{Appendicies B, C and D}.

The mapping of $\IMMsc$ to each architecture
follows the intended compilation scheme of C/C++11
\cite{www:mappings,Lahav-al:PLDI17},
and extends the corresponding mappings of $\IMM$ from Podkopaev \etal~\cite{Podkopaev-al:POPL19}
with the mapping of SC reads and writes.
Next, we schematically present these extensions.

\subsubsection{\TSO}

There are two alternative sound mappings of SC accesses to \Intel-\TSO:
\begin{center}
$\begin{array}{@{}l@{}c@{{}\defeq{}}l@{\qquad\qquad}l@{}c@{{}\defeq{}}l@{}}
\multicolumn{3}{@{}c@{\qquad\qquad}}{
\text{{Fence after SC writes}}} & 
\multicolumn{3}{@{}c@{}}{
\text{{Fence before SC reads}}} \\
\hline
\compile{\lR^\sco} && \texttt{mov} &
\compile{\lR^\sco} && \texttt{mfence;mov} \\ 
\compile{\lW^\sco} && \texttt{mov;mfence} &
\compile{\lW^\sco} && \texttt{mov} \\ 
\compile{\lRMWc^\sco} && \texttt{(lock) xchg} &
\compile{\lRMWc^\sco} && \texttt{(lock) xchg} \\
\end{array}$
\end{center}
The first, which is implemented in mainstream compilers, inserts an $\texttt{mfence}$
after every SC write; whereas the second inserts an $\texttt{mfence}$ before every SC read.
Importantly, one should \emph{globally} apply one of the two mappings
to ensure the existence of an $\texttt{mfence}$ between every SC write and following SC read.

\subsubsection{\POWER}
\label{sec:power}

There are two alternative sound mappings of SC accesses to \POWER:
\begin{center}
$\begin{array}{@{}l@{}c@{{}\defeq{}}l@{\qquad\qquad}l@{}c@{{}\defeq{}}l@{}}
\multicolumn{3}{@{}c@{\qquad\qquad}}{
\text{{ Leading \texttt{sync}}}} & 
\multicolumn{3}{@{}c@{}}{
\text{{ Trailing \texttt{sync}}}} \\
\hline
\compile{\lR^\sco} && \texttt{sync;}\compile{\lR^\acq} &
\compile{\lR^\sco} &&  \texttt{ld;sync} \\ 
\compile{\lW^\sco} && \texttt{sync;st} &
\compile{\lW^\sco} &&  \compile{\lW^\rel}\texttt{;sync} \\ 
\compile{\lRMWc^\sco} && \texttt{sync;} \compile{\lRMWc^\acq} &
\compile{\lRMWc^\sco} &&  \compile{\lRMWc^\rel}\texttt{;sync}
\end{array}$
\end{center}
The first scheme inserts a $\texttt{sync}$ before every SC access,
while the second inserts an $\texttt{sync}$ after every SC access.
Importantly, one should \emph{globally} apply one of the two mappings
to ensure the existence of a $\texttt{sync}$ between every two SC accesses.

Observing that $\texttt{sync}$ is the result of mapping an SC-fence to \POWER,
we can reuse the existing proof for the mapping of \IMM to \POWER.
To handle the leading $\texttt{sync}$ (respectively, trailing $\texttt{sync}$) 
scheme we introduce a preceding step,
in which we prove that splitting in the whole execution graph
each SC access to a pair of an SC fence followed (preceded) by a release/acquire access
is a sound transformation under \IMMsc.
That is, this global execution graph transformation cannot make an
inconsistent execution consistent:

\begin{theorem}
  \label{thm:sc-access-to-fence}
Let $G$ be an execution graph such that
$$[\lR^\sco \cup \lW^\sco]\seq (G.\lPO' \cup G.\lPO'\seq G.\lHB\seq G.\lPO')\seq [\lR^\sco \cup \lW^\sco] \subseteq G.\lHB\seq [\lF^\sco]\seq G.\lHB,$$
where $G.\lPO' \defeq G.\lPO \setminus G.\lRMW$.
Let $G'$ be the execution graph obtained from $G$ by weakening the access modes
of SC write and read events to release and acquire modes respectively.
Then, \IMMsc-consistency of $G$ follows from \IMM-consistency of $G'$.
\end{theorem}

Having this theorem, we can think about mapping of \IMMsc to \POWER as if it consists 
of three steps. We establish the correctness of each of them separately.

\begin{enumerate}
\item At the \IMMsc level, we globally split each SC-access to an SC-fence and release/acquire access.
Correctness of this step follows by \cref{thm:sc-access-to-fence}.

\item We map \IMM to \POWER, whose correctness follows by the existing results
of~\cite{Podkopaev-al:POPL19}, since we do not have SC accesses at this stage.

\item We remove any redundant fences introduced by the previous step.
Indeed, following the leading $\texttt{sync}$ scheme, we will obtain $\texttt{sync;lwsync;st}$ for an SC write.
The $\texttt{lwsync}$ is redundant here since $\texttt{sync}$ provides stronger guarantees than $\texttt{lwsync}$
and can be removed.
Similarly, following the trailing $\texttt{sync}$ scheme, 
we will obtain $\texttt{ld;cmp;bc;isync;sync}$ for an SC read.
Again, the $\texttt{sync}$ makes other synchronization instructions redundant.

\end{enumerate}

\subsubsection{\ARMs}

The \ARMs model \cite{Alglave-al:TOPLAS14} is very similar to the \POWER model
with the main difference being that it has a weaker preserved program order than \POWER.
However, Podkopaev \etal~\cite{Podkopaev-al:POPL19} proved \IMM to \POWER compilation correctness
without relying on \POWER's preserved program order explicitly,
but assuming the weaker version of \ARMs's order.
Thus, their proof also establishes correctness of compilation from \IMM to \ARMs.

Extending the proof to cover SC accesses follows the same scheme discussed for \POWER,
since two intended mappings of SC accesses for \ARMs are the same except for
replacing \POWER's $\texttt{sync}$ fence with \ARMs's $\texttt{dmb}$:
\begin{center}
$\begin{array}{@{}l@{}c@{{}\defeq{}}l@{\qquad\qquad}l@{}c@{{}\defeq{}}l@{}}
\multicolumn{3}{@{}c@{\qquad\qquad}}{
\text{{ Leading \texttt{dmb}}}} & 
\multicolumn{3}{@{}c@{}}{
\text{{ Trailing \texttt{dmb}}}} \\
\hline
\compile{\lR^\sco} && \texttt{dmb;}\compile{\lR^\acq} &
\compile{\lR^\sco} &&  \texttt{ldr;dmb} \\ 
\compile{\lW^\sco} && \texttt{dmb;str} &
\compile{\lW^\sco} &&  \compile{\lW^\rel}\texttt{;dmb} \\ 
\compile{\lRMWc^\sco} && \texttt{dmb;} \compile{\lRMWc^\acq} &
\compile{\lRMWc^\sco} &&  \compile{\lRMWc^\rel}\texttt{;dmb}
\end{array}$
\end{center}

\subsubsection{\ARMe}
\label{sec:sc-arm8}

Since \ARMe has added dedicated instructions to support C/C++-style
SC accesses, we have established the correctness of a mapping employing
these new instructions:
\begin{center}
$\begin{array}{@{}l@{}c@{{}\defeq{}}l@{}}
\compile{\lR^\sco} && \texttt{LDAR} \\
\compile{\lW^\sco} && \texttt{STLR} \\
\compile{\texttt{FADD}^\sco} && \texttt{L:LDAXR;STLXR;BC L} \\
\compile{\texttt{CAS}^\sco}  && \texttt{L:LDAXR;CMP;BC Le;STLXR;BC L;Le:} \\
\end{array}$
\end{center}

We note that in this mapping, we follow Podkopaev \etal~\cite{Podkopaev-al:POPL19}
and compile RMW operations to loops with load-linked and store-conditional
instructions (\texttt{LDX}/\texttt{STX}).
An alternative mapping for RMWs would be to use single hardware instructions,
such as \texttt{LDADD} and \texttt{CAS}, that directly implement the required
functionality.
Unfortunately, however, due to a limitation of the current \IMM setup and
unclarity about the exact semantics of the \texttt{CAS} instruction,
we are not able to prove the correctness of the alternative mapping employing
these instructions.
The problem is that \IMM assumes that every $\lPO$-edge from a RMW instruction
is preserved, which holds for the mapping of \texttt{CAS} using the
aforementioned loop, but not necessarily using the single instruction.
 
%
%

\section{Related Work}
\label{sec:related}

While there are several memory model definitions both for hardware
architectures~\cite{Alglave-al:TOPLAS14,Flur-al:POPL16,Owens-al:TPHOL09,Pulte-al:POPL18,Pulte-al:PLDI19}
and programming languages~\cite{Batty-al:POPL11,Bender-Palsberg:OOPSLA19,Jeffrey-Riely:LICS16,Manson-al:POPL05,Pichon-al:POPL16,Podkopaev-al:CoRR16}
in the literature, there are relatively few compilation correctness results
\cite{Chakraborty-Vafeiadis:POPL19,Dolan-al:PLDI18,Kang-al:POPL17,Lahav-al:PLDI17,Podkopaev-al:POPL19,compcerttso}.

Most of these compilation results do not tackle any of the problems caused by
$\lPO\cup\lRF$ cycles, which are the main cause of complexity in establishing
correctness of compilation mappings to hardware architectures.
A number of papers (\eg \cite{Chakraborty-Vafeiadis:POPL19,Kang-al:POPL17,compcerttso})
consider only hardware models that forbid such cycles,
such as x86-TSO~\cite{Owens-al:TPHOL09} and ``strong POWER''~\cite{Lahav-Vafeiadis:FM16},
while others (\eg \cite{Dolan-al:PLDI18})
consider compilation schemes that introduce fences and/or dependencies
so as to prevent $\lPO\cup\lRF$ cycles.
The only compilation results where there is some non-trivial interplay of dependencies
are by Lahav \etal~\cite{Lahav-al:PLDI17} and by Podkopaev \etal~\cite{Podkopaev-al:POPL19}.

The former paper~\cite{Lahav-al:PLDI17} defines the \RC model (repaired C11),
and establishes a number of results about it, most of which are not related to
compilation.  The only relevant result is its pencil-and-paper correctness proof
of a compilation scheme from \RC to POWER
that adds a fence between relaxed reads and subsequent relaxed writes,
but not between non-atomic accesses.
As such, the only $\lPO\cup\lRF$ cycles possible under the compilation scheme
involve a racy non-atomic access.
Since non-atomic races have undefined semantics in RC11,
whenever there is such a cycle,
the proof appeals to receptiveness to construct a different acyclic execution
exhibiting the race.

The latter paper~\cite{Podkopaev-al:POPL19} introduced \IMM and used it to
establish correctness of compilation from the ``promising semantics'' (\Promise) \cite{Kang-al:POPL17}
to the usual hardware models.
As already mentioned, \IMM's definition catered precisely for the needs of
the \Promise compilation proof, and so did not include important features
such as sequentially consistent (SC) accesses.
Our compilation proof shares some infrastructure with that proof---namely, the
definition of \IMM and traversals---but also has substantial differences
because \Promise is quite different from \weakestmo.
The main challenges in the \Promise proof were (1) to encode the various orders
of the \IMM execution graphs with the timestamps of the \Promise machine,
and (2) to construct the certification runs for each outstanding promise.
In contrast, the main technical challenge in the \weakestmo compilation proof
is that event structures represent several possible executions of the program
together, and that \weakestmo consistency includes constraints that
correlate these executions, allowing one execution to affect the consistency of
another.


\section{Conclusion}
\label{sec:conclusion}

In this paper, we presented the first correctness proof of mapping
from the \weakestmo memory model to a number of hardware architectures.
As a way to show correctness of \weakestmo compilation to hardware,
we employed \IMM~\cite{Podkopaev-al:POPL19}, which we extended with SC
accesses, from which compilation to hardware follows.

Although relying on \IMM modularizes the compilation proof
and makes it easy to extend to multiple architectures, it does have one limitation.
As was discussed in \cref{sec:sc-arm8},
\IMM enforces ordering between RMW events and subsequent memory accesses,
while one desirable alternative compilation mapping of RMWs to \ARMe does not
enforce this ordering, which means that we cannot prove soundness of that mapping via
the current definition of \IMM.
We are investigating whether one can weaken the corresponding \IMM constraint,
so that we can establish correctness of the alternative \ARMe mapping as well.

Another way to establish correctness of this alternative mapping to \ARMe
may be to use the recently developed Promising-ARM model~\cite{Pulte-al:PLDI19}.
Indeed, since Promising-ARM is closely related to \Promise~\cite{Kang-al:POPL17},
it should be relatively easy to prove the correctness of compilation from \Promise to Promising-ARM.
Establishing compilation correctness of \weakestmo to Promising-ARM, however,
would remain unresolved because \weakestmo and \Promise are incomparable \cite{Chakraborty-Vafeiadis:POPL19}.
Moreover, a direct compilation proof would probably also be quite difficult
because of the rather different styles in which these models are defined.


 \bigskip \noindent\textbf{Acknowledgments.}
  Evgenii Moiseenko and Anton Podkopaev were supported by RFBR (grant number 18-01-00380).
  Ori Lahav was supported by the Israel Science Foundation
  (grant number 5166651), by Len Blavatnik and the Blavatnik Family foundation,
  and by the Alon Young Faculty Fellowship.

\bibliographystyle{plainurl}
\bibliography{main}

\iftoggle{fullpaper}{%
  \newpage
  \appendix
  \section{Simulation Relation for the complete \weakestmo model}
\label{app:simrel}

Here we present the simulation relation $\simrel_\TidSet(\prog, G, \TC, \ES, \SX)$ and 
the auxiliary relation $\simrelCert(\prog, G, \tup{\CoveredSet, \IssuedSet}, \tup{\CoveredSet', \IssuedSet'}, \ES, \SX, \tid, \SBr, \pstate, \pstate')$
for the complete \weakestmo memory model. 
In addition to the relaxed accesses the full versions of the relations handle 
fences, read-modify-write pairs, release, acquire and sequentially consistent accesses. 

We define the relation $\simrel_\TidSet(\prog, G, \tup{\CoveredSet, \IssuedSet}, \ES, \SX)$ to hold if the following conditions are met:%

\begin{enumerate}
  \item \label{full-simrel:gexec}
    $G$ is an $\IMMsc$-consistent execution of $\prog$.

  \item \label{full-simrel:sexec}
    $\ES$ is a \weakestmo-consistent event structure of $\prog$.

  \item \label{full-simrel:ex}
    $\SX$ is an extracted subset of $\ES$.

  \item \label{full-simrel:ex-cov-iss}
    The $\ea$-image of $\SX$ is equal to the union of the covered and issued events and the events
    which $\lPO$-precede the issued ones:
    \begin{itemize}
      \item $\fmap{\SX} = \CoveredSet \cup \dom{G.\lPO^? \seq [\IssuedSet]}$
    \end{itemize}

  \item \label{full-simrel:tid-cov-iss}
    The $\ea$-image of the event from the thread $t \in \TidSet$ lies in $\CoveredSet \cup \dom{G.\lPO^? \seq [\IssuedSet]}$.
    \begin{itemize}
      \item $\fmap{\ES.\thread{\TidSet}} \subseteq \CoveredSet \cup \dom{G.\lPO^? \seq [\IssuedSet]}$
    \end{itemize}
    
  \item \label{full-simrel:lab}
    The $\ea$-image of $\ES$'s event has the same thread, type, modifier, and location.
    Additionally, the $\ea$-image of $\SX$'s event which is covered or issued has the same value:
    \begin{enumerate}
      \item $\forall e \in \ES.\lE. \; \ES.\set{\lTID, \lTYP, \lLOC, \lMOD}(e) = G.\set{\lTID, \lTYP, \lLOC, \lMOD}(\ea(e))$
      \item $\forall e \in \SX \cap \fcomap{C \cup I} \ldotp~ \ES.\lVAL(e) = G.\lVAL(\ea(e))$
    \end{enumerate}

  \item \label{full-simrel:po}
    The $\ea$-image of $\ES.\lPO$ is a subset of the $G.\lPO$ relation:
    \begin{itemize}
      \item $\fmap{\ES.\lPO} \subseteq G.\lPO$
    \end{itemize}

  \item \label{full-simrel:cf}
    Identity relation in $G$ corresponds to identity or conflict relation in $\ES$:
    \begin{itemize}
      \item $\fcomap{\mathtt{id}} \subseteq S.\lCF^?$
    \end{itemize}

  \item \label{full-simrel:jf}
    The $\ea$-image of a justification edge is included in paths in $G$ representing observation of the corresponding thread.  
    The $\ea$-image of a justification edge is in $G.\lRF$ if the edge ends either in domain of $\ES.\lRMW$, an acquire access, or followed by an acquire fence.
    Moreover, the $\ea$-image of $\ES.\lJF$ ending in $\SX$ matches the simulation reads-from relation:
    \begin{enumerate}
      \item \label{full-simrel:jf-obs}
        $\fmap{\ES.\lJF} \subseteq G.\lRF^?\seq (G.\lHB\seq [G.\lF^\sco])^? \seq G.\lPSCF^? \seq G.\lHB^?$
      \item \label{full-simrel:jf-rmw}
        $\fmap{\ES.\lJF \seq \ES.\lRMW} \subseteq G.\lRF \seq G.\lRMW$
      \item \label{full-simrel:jf-acq}
        $\fmap{\ES.\lJF \seq (\ES.\lPO \seq [\ES.\lF])^? \seq [\ES.\lE^{^{\sqsupseteq\acq}}]}
          \subseteq G.\lRF \seq (G.\lPO \seq [\ES.\lF])^? \seq [G.\lE^{^{\sqsupseteq\acq}}]$
      \item \label{full-simrel:jf-sim-rf}
        $\fmap{\ES.\lJF \seq [\SX]} \subseteq G.\lSRF(\TC)$
    \end{enumerate}
    Using the last property it is possible to derive that $\fmap{\ES.\lJF \seq [\SX \cap \fcomap{C}]} \subseteq G.\lRF$.

  \item \label{full-simrel:jfe-iss}
    Each write event in $\ES$ which justifies some read event externally
    should be $\ES.\lEW$-equal to a write event in $\SX$ whose $\ea$-image is issued:
    \begin{itemize}
      \item $\dom{\ES.\lJFE} \subseteq \dom{\ES.\lEW \seq [\SX \cap \fcomap{I}]}$
    \end{itemize}

  \item \label{full-simrel:ew-id}
    The $\ea$-image of $\ES.\lEW$ is a subset of the identity relation:
    \begin{itemize}
      \item $\fmap{\ES.\lEW} \subseteq \mathtt{id}$
    \end{itemize}

  \item \label{full-simrel:ew-iss}
    Let $w$ and $w'$ be different events in one $\ES.\lEW$ equivalence class.
    Then, there is $w''$ in this equivalence class s.t. $w''$ is in $\SX$ and $\ea(w'')$ is issued:
    \begin{itemize}
      \item $\ES.\lEW \subseteq (\ES.\lEW \seq [\SX \cap \fcomap{I}] \seq \ES.\lEW)^?$
    \end{itemize}

  \item \label{full-simrel:co}
    The $\ea$-image of $\ES.\lCO$ lies in the reflexive closure of $G.\lCO$.
    Additionally, $\ea$-images of $\ES.\lCO$-edges ending in $\SX \cap \ES.\thread{\TidSet}$
    lay in $G.\lCO$:
    \begin{enumerate}
      \item \label{full-simrel:co-weak}
        $\fmap{\ES.\lCO} \subseteq G.\lCO^?$
      \item \label{full-simrel:co-ex}
        $\fmap{\ES.\lCO \seq [\SX \cap \ES.\thread{\TidSet}} \subseteq G.\lCO$
    \end{enumerate}

  \item \label{full-simrel:rmw}
    The $\ea$-image of $\ES.\lRMW$ is in $G.\lRMW$.
    Vice versa, $G.\lRMW$ ending in the covered set is in the $\ea$-image of $\ES.\lRMW$ ending in $\SX$.
    \begin{enumerate}
      \item $\fmap{\ES.\lRMW} \subseteq G.\lRMW$
      \item $G.\lRMW \seq [\CoveredSet] \subseteq \fmap{\ES.\lRMW \seq [\SX]}$
    \end{enumerate}

  \item \label{full-simrel:release}
    Let $e$, $w$, and $w'$ be events in $\ES$ s.t.
    (i) $\tup{e,w}$ is an $\ES.\lRELEASE$ edge,
    (ii) $w$ and $w'$ is in the same $\ES.\lEW$ equivalence class,
    (iii) $w'$ is in $\SX$, and (iv) $\ea(w')$ is issued.
    Then $e$ is in $\SX$:
    \begin{itemize}
      \item $\dom{\ES.\lRELEASE \seq \ES.\lEW \seq [\SX \cap \fcomap{I}]} \subseteq \SX$
    \end{itemize}
    This property is needed to show that $\dom{\ES.\lHB \setminus \ES.\lPO}$ is included in $\SX$.

  \item Let $r$, $r'$, $w$, and $w'$ be events in $\ES$ s.t.
    (i) $r$ and $r'$ are in immediate conflict and justified from $w$ and $w'$ respectively,
    and (ii) $r'$ is in $\SX$ and its thread is in $\TidSet$.
    Then $\ea(w)$ is $G.\lCO$-less than $\ea(w')$:
    \begin{itemize}
      \item \label{full-simrel:icf}
        $\fmap{\ES.\lJF \seq \ES.\lCFimm \seq
         [\SX \cap \ES.\thread{\TidSet}] \seq \ES.\lJF^{-1}} \subseteq G.\lCO$
    \end{itemize}
    This property is needed to prove \ref{ax:icf-jf} on the simulation step.

  \item \label{full-simrel:state-cov}
    For all $\tid \in \TidSet$ there exists $\pstate$ 
    \sth ${\ES.\lCONT_{\CoveredSet}(\tid) \threadstep{\tid}{}^* \pstate}$ and 
    the thread-local execution graph $\sigma.G$ is equivalent modulo $\lRF$ and $\lCO$ components
    to the restriction of $G$ to the thread $\tid$.

\end{enumerate}

In addition to $\simrel$ we also define a version of the simulation realtion 
which holds during the construction of a certification branch $\simrelCert$. 

We define the relation $\simrelCert(\prog, G, \tup{\CoveredSet, \IssuedSet}, \tup{\CoveredSet', \IssuedSet'}, \ES, \SX, \tid, \SBr, \pstate, \pstate')$ to hold if the following conditions are met:%

\begin{enumerate}

  \item \label{full-simrelcert:simrel}
    $\simrel_{\TidSet \setminus \{ \tid \}}(\prog, G, \tup{\CoveredSet, \IssuedSet}, \ES, \SX)$ holds.

  \item \label{full-simrelcert:trav-step}
    $G \vdash \tup{\CoveredSet, \IssuedSet} \travstep{t} \tup{\CoveredSet', \IssuedSet'}$ holds.

  \item \label{full-simrelcert:thread-steps}
    $\pstate$ and $\pstate'$ are thread states \sth 
    $\pstate'$ is reachable from $\pstate$,
    $\pstate$ corresponds to the $\ES.\lPO$-last event in $\SBr$
    and the partial execution graph of $\pstate'$ contains
    covered and issued events up to the $G.\lPO$-last issued write in the thread $\tid$:
    \begin{enumerate}
      \item \label{full-simrelcert:thread-steps-reach}
        $\pstate \threadstep{\tid}{}^* \pstate'$
      \item \label{full-simrelcert:thread-g}
        $\pstate.G.\lE = \fmap{\SBr}$
      \item \label{full-simrelcert:thread-cov-iss}
        $\pstate'.G.\lE = G.\thread{\tid} \cap (\CoveredSet' \cup \dom{G.\lPO^? \seq [\IssuedSet']})$
    \end{enumerate}

  \item \label{full-simrelcert:br} 
    The set $\SBr$ consists of the events from the thread $\tid$
    and covered prefixes of $\SBr$ and $\SX$ restricted to thread $\tid$ coincide:
    \begin{enumerate}
      \item $\SBr \subseteq S.\thread{\tid}$ 
      \item $\SBr \cap \fcomap{\CoveredSet} = \SX \cap S.\thread{\tid} \cap \fcomap{\CoveredSet}$
    \end{enumerate}

  \item \label{full-simrelcert:lab}
    The partial execution graph of $\pstate'$ assigns 
    same thread identifier, type, location and mode as the full execution graph $G$ does.
    Additionally, it assigns the same value as $G$ to determined events.
    \begin{enumerate}
      \item 
        $\forall e \in \pstate'.G.\lE \ldotp~ \pstate'.G.\set{\lTID, \lTYP, \lLOC, \lMOD}(e) = G.\set{\lTID, \lTYP, \lLOC, \lMOD}(e)$
      \item 
        $\forall e \in \pstate'.G.\lE \cap G.\determined(\tup{\CoveredSet', \IssuedSet'}) \ldotp~
            \pstate'.G.\lVAL(e) = G.\lVAL(e)$
    \end{enumerate}
    
  \item \label{full-simrelcert:jf}
    The $\ea$-image of the $\lJF$ relation ending in $\SBr$ is included in $G.\lSRF(\tup{\CoveredSet', \IssuedSet'})$:
    \begin{itemize}
      \item $\fmap{\ES.\lJF \seq [\SBr]} \subseteq G.\lSRF(\tup{\CoveredSet', \IssuedSet'})$ 
    \end{itemize}

  \item \label{full-simrelcert:ew}
    For every issued event from $\SBr$ there exists an $\ES.\lEW$-equivalent in $\SX$.
    And, symmetrically, every issued event from $\SX$ within the processed part of the certification branch 
    has an $\ES.\lEW$-equivalent in $\SBr$.
    \begin{enumerate}
      \item $\SBr \cap \fcomap{\IssuedSet} \subseteq \dom{\ES.\lEW \seq [\SX]}$ 
      \item $\SX \cap \fcomap{\IssuedSet \cap \pstate.G.\lE} \subseteq \dom{\ES.\lEW \seq [\SBr]}$ 
    \end{enumerate}

  \item \label{full-simrelcert:co}
    The $\ea$-image of $\ES.\lCO$ ending in $\SBr$ lies in $G.\lCO$
    The $\ea$-image of $\ES.\lCO$ ending in $\SX \cap S.\thread{t}$ 
    and not in the processed part of the certification branch lies in $G.\lCO$.
    \begin{enumerate}
      \item $\fmap{\ES.\lCO \seq [\SBr]} \subseteq G.\lCO$
      \item $\fmap{\ES.\lCO \seq [\SX \cap S.\thread{t} \setminus \fcomap{\pstate.G.\lE}]} \subseteq G.\lCO$ 
    \end{enumerate}

  \item \label{full-simrelcert:rmw}
    Each $G.\lRMW$ edge ending in the processed part of the certification branch
    is the $\ea$-image of some $\ES.\lRMW$ edge ending in $\SBr$.
    \begin{itemize}
      \item $G.\lRMW \seq [C' \cap \pstate.G.\lE] \subseteq \fmap{\ES.\lRMW \seq [\SBr]}$
    \end{itemize}

  \item \label{full-simrelcert:icf}
    Suppose $w$, $w'$, $r$, and $r'$ are $\ES$'s events \sth
    (i) $r$ and $r'$ are justified from $w$ and $w'$ respectively,
    and (ii) $r$ and $r'$ are in immediate conflict and belong to thread $t$.
    Then $\ea(w')$ is $G.\lCO$-greater than $\ea(w)$ if either
    $r'$ is in $\SBr$:
    \begin{itemize}
      \item $\fmap{\ES.\lJF \seq \ES.\lCFimm \seq [\SBr] \seq \ES.\lJF^{-1}} \subseteq G.\lCO$ 
    \end{itemize}
    or $r$ is not in $\SBr$ and $r'$ is in $\SX \cap S.\thread{t}$:
    \begin{itemize}
      \item $\fmap{\ES.\lJF \seq [\ES.\lE \setminus \SBr] \seq \ES.\lCFimm \seq [\SX \cap S.\thread{t}] \seq \ES.\lJF^{-1}}
         \subseteq G.\lCO$ 
    \end{itemize}

\end{enumerate}

  \section{From \IMMsc to \ARMe}
\label{app:imm_to_arm}

\begin{figure}[t]
\small
\centering
$\begin{array}{@{}r@{}c@{}l@{\qquad \qquad \qquad \;\;}r@{}c@{{}\approx{}}l@{}}
\compile{\readInst{\rlx}{r}{e}} & {}\approx{} &  ``\texttt{ldr}" & 
\compile{\writeInst{\rlx}{e_1}{e_2}} &&  ``\texttt{str}" \\
\compile{\readInst{\sqsupseteq\acq}{r}{e}} & {}\approx{} & ``\texttt{ldar}"  &
\compile{\writeInst{\sqsupseteq\rel}{e_1}{e_2}} && ``\texttt{stlr}" \\ 
\compile{\fenceInst{\acq}}  & {}\approx{} & ``\texttt{dmb.ld}" &
\compile{\fenceInst{\neq\acq}} && ``\texttt{dmb.sy}" \\
\compile{\incInst{o}{}{r}{e_1}{e_2}{}} & {}\approx{} &
\multicolumn{4}{@{}l@{}}{
  ``\textsf{L:}" \doubleplus {\rm ld}(o) \doubleplus
  {\rm st}(o) \doubleplus ``\texttt{bc  }\textsf{L}"
}\\ 
\compile{\casInst{o}{}{r}{e}{e_\lR}{e_\lW}{}} & {}\approx{} &
\multicolumn{4}{@{}l@{}}{  
  ``\textsf{L:}" \doubleplus {\rm ld}(o) \doubleplus ``\texttt{cmp;bc } \textsf{Le}
  \texttt{;}" \doubleplus {\rm st}(o) \doubleplus
  ``\text{bc  }\textsf{L}\texttt{;}\textsf{Le:}"
} 
\end{array}$ \\
$\begin{array}{@{}l@{}l@{\quad \qquad}l@{}}
{\rm ld}(o) & \defeq o \sqsupseteq \acq \; ? \; ``\texttt{ldaxr;}" \; : \; ``\texttt{ldxr;}" &
{\rm st}(o) \defeq o \sqsupseteq \rel \; ? \; ``\texttt{stlxr.;}" \; : \; ``\texttt{stxr.;}" \\
\end{array}$
\caption{Compilation scheme from \IMMsc to ARMv8.}
\label{fig:compArm}
\end{figure}

The intended mapping of \IMM to \ARMe is presented schematically in \cref{fig:compArm} and follows~\cite{www:mappings}.
Note that acquire and SC loads are compiled to the same instruction (\texttt{ldar})
as well as release and SC stores (\texttt{stlr}). In \ARM assembly RMWs are represented as pairs of instructions---%
\emph{exclusive} load (\texttt{ldxr}) followed by \emph{exclusive} store (\texttt{stxr}),
and these instructions are also have their stronger (SC) counterparts---\texttt{ldaxr} and \texttt{stlxr}.

We use \ARMe declarative model~\cite{ARMv82model} (see also \cite{Pulte-al:POPL18}).%
\footnote{We only describe the fragment of the model that is needed for mapping of \IMMsc,
thus excluding \texttt{isb} fences.}
Its labels are given by:
\begin{itemize}
\item \ARM read label: $\erlab{o_\lR}{\loc}{v}{}$ where $\loc\in\Loc$, $v\in\Val$, $o_\lR\in\set{\rlx,\lQ,\lA}$, and $\rlx \sqsubset \lQ \sqsubset \lA$.
\item \ARM write label: $\ewlab{o_\lW}{\loc}{v}{}$ where $\loc\in\Loc$, $v\in\Val$, $o_\lW\in\set{\rlx,\lL}$, and $\rlx \sqsubset \lL$.
\item \ARM fence label: $\flab{o_\lF}$ where $o_\lF\in\set{\ld,\full}$ and $\ld \sqsubset \full$.
\end{itemize}
In turn, \ARM's execution graphs are defined as \IMMsc's ones,
except for the CAS dependency, $\lRMWDEP$, which is not present in \ARM executions.


The definition of \ARMe-consistency requires the following derived relations
(see~\cite{Pulte-al:POPL18} for further explanations and details):

\begin{align*}
\lOBS &\defeq  \lRFE \cup \lFRE \cup \lCOE  \tag{\emph{observed-by}} \\
\lDOB &\defeq \inarr{(\lADDR \cup \lDATA); \lRFI^? \cup 
	(\lCTRL \cup \lDATA); [\lW]; \lCOI^? \cup 
	\lADDR; \lPO; [\lW] } \tag{\emph{dependency-ordered-before}} \\
\lAOB &\defeq \lRMW \cup [\lW^\isex]; \lRFI; [\lR^{\sqsupseteq\lQ}] \tag{\emph{atomic-ordered-before}} \\
\lBOB &\defeq \inarr{\lPO; [\lDMBSY]; \lPO  \cup 
	{}[\lR]; \lPO; [\lDMBLD]; \lPO  \cup 
	{}[\lR^{\sqsupseteq\lQ}]; \lPO \cup 
	\lPO; [\lW^\lL]; \lCOI^? \cup 
        {}[\lW^\lL]; \lPO; [\lR^\lA]
        \tag{\emph{barrier-ordered-before}} }
\end{align*}

\begin{definition}
\label{def:arm}
An \ARMe execution graph $G_a$ is called \ARMe-consistent if the following hold:
\begin{itemize}
\item $\codom{G_a.\lRF} = G_a.\lR$. 
\item For every location $\loc\in\Loc$, $G_a.\lCO$ totally orders $G_a.\ewlab{}{\loc}{}{}$. 
\item $G_a.\lPO\rst{\lLOC} \cup G_a.\lRF \cup G_a.\lFR \cup G_a.\lCO$ is acyclic.
 \labelAxiom{sc-per-loc}{ax:A-sc_per_loc_arm}
\item $G_a.\lRMW \cap (G_a.\lFRE; G_a.\lCOE)=\emptyset$. 
\item $G_a.\lOBS \cup G_a.\lDOB \cup G_a.\lAOB \cup G_a.\lBOB$ is acyclic. \labelAxiom{external}{ax:external}
\end{itemize}
\end{definition}

We interpret the intended compilation on execution graphs:

\begin{definition}
\label{def:cmp_exec_arm}
Let $G$ be an \IMM execution graph.
An \ARM execution graph $G_a$ \emph{corresponds} to $G$ if the following hold:
\begin{itemize}
\item 
$G_a.\lE = G.\lE$ and $G_a.\lPO = G.\lPO$
\item  
$G_a.\lLAB = \set{ e \mapsto \compile{G.\lLAB(e)} \st e\in G.\lE}$ where:
$$\begin{array}{r@{\;}l@{\qquad\quad}r@{\;}l}
\compile{\erlab{\rlx}{x}{v}{s}} & \defeq \rlab{\rlx}{x}{v}  &  \compile{\ewlab{\rlx}{x}{v}{}} & \defeq \wlab{\rlx}{x}{v} \\
\compile{\erlab{\acq}{x}{v}{s}} & \defeq \rlab{\lQ}{x}{v}  & \compile{\ewlab{\sqsupseteq\rel}{x}{v}{}} & \defeq \wlab{\lL}{x}{v} \\
\compile{\erlab{\sco}{x}{v}{s}} & \defeq \rlab{\lA}{x}{v} \\
\compile{\flab{\acq}} & \defeq \flab{\ld}& \compile{\flab{\rel}} =\compile{\flab{\acqrel}} = \compile{\flab{\sco}} & \defeq \flab{\full}
\end{array}$$
\item
$G.\lRMW = G_a.\lRMW$, $G.\lDATA = G_a.\lDATA$, and $G.\lADDR = G_a.\lADDR$
\\ (the compilation does not change RMW pairs and data/address dependencies)
\item
$G.\lCTRL \suq G_a.\lCTRL$
\\ (the compilation only adds control dependencies)
\item
$[G.\erlab{}{}{}{\isex}] \mathbin{;}G.\lPO \suq G_a.\lCTRL \cup G_a.\lRMW \cap G_a.\lDATA $
\\ (exclusive reads entail a control dependency to any future event, 
except for their immediate exclusive write successor if arose from an atomic increment)
\item
$G.\lRMWDEP \mathbin{;} G.\lPO \suq G_a.\lCTRL$
\\ (CAS dependency to an exclusive read entails a control dependency to any future event)
\end{itemize}
\end{definition}

We state our theorem that ensures \IMMsc-consistency if the corresponding \ARMe
execution graph is \ARMe-consistent.

\begin{theorem}
  \label{thm:imm-to-arm}
Let $G$ be an \IMM execution graph with whole serial numbers ($\lSN[G.\lE] \suq \N$),
and let $G_a$ be an \ARMe execution graph that corresponds to $G$.
Then, \ARMe-consistency of $G_a$ implies \IMMsc-consistency of $G$. 
\end{theorem}
\begin{proof}[Outline]
  \IMM-consistency of $G$ follows from~\cite[Theorem 4.5]{Podkopaev-al:POPL19}.
  That is, we only need to show that acyclicity of $G.\lPSCB \cup G.\lPSCF$ holds.
  We start by showing that
  $G_a.\lOBS' \cup G_a.\lDOB \cup G_a.\lAOB \cup G_a.\lBOB'$ is acyclic, where
  \[\begin{array}{@{}l@{}l@{}}
    \lOBS' & {} \defeq \lRFE \cup \lFR \cup \lCO \\
    \lBOB' & {} \defeq \lBOB \cup [\lR]; \lPO; [\lF^{\ld}] \cup \lPO; [\lF^{\full}] \cup [\lF^{\sqsupseteq\ld}]; \lPO
  \end{array}\]
  Then, we finish the proof by showing that $G_a.\lPSCB \cup G_a.\lPSCF$ is included in 
  $(G_a.\lOBS' \cup G_a.\lDOB \cup G_a.\lAOB \cup G_a.\lBOB')^+$. \qedhere
\end{proof}

  \section{From \IMMsc to \TSO}
\label{app:imm_to_tso}

\begin{figure}[t]
\small
\centering
$\begin{array}{@{}r@{{}\approx{}}l@{}r@{{}\approx{}}l@{}}
\compile{\readInst{\neq\sco}{r}{e}} &  ``\texttt{mov}" &
\compile{\fenceInst{\neq\sco}}  & ``" \\
\compile{\writeInst{\neq\sco}{e_1}{e_2}} & ``\texttt{mov}" &
\compile{\fenceInst{\sco}} & ``\texttt{mfence}" \\
\quad \compile{\incInst{o}{}{r}{e_1}{e_2}{}} &
  ``\texttt{(lock) xadd}" \\
\quad \compile{\casInst{o}{}{r}{e}{e_\lR}{e_\lW}{}} &
  ``\texttt{(lock) cmpxchg}" \\
\hline
\textbf{Alt. 1:} \quad
\compile{\readInst{\sco}{r}{e}} & \multicolumn{1}{@{}l|@{}}{``\texttt{mov}"} &
\textbf{Alt. 2:} \quad
\compile{\readInst{\sco}{r}{e}} &``\texttt{mfence;mov}" \\

\compile{\writeInst{\sco}{e_1}{e_2}} &
  \multicolumn{1}{@{}l|@{}}{``\texttt{mov;mfence}"} &
\compile{\writeInst{\sco}{e_1}{e_2}} &``\texttt{mov}" \\ 
\end{array}$ \\
\caption{Compilation scheme from \IMMsc to \TSO.} 
\label{fig:compTSO}
\end{figure}

The intended mapping of \IMMsc to \TSO is presented schematically in \cref{fig:compTSO}.
There are two possible alternatives for compiling SC accesses
(see the bottom of \cref{fig:compTSO}): to compile an SC store to a store followed by a fence
or to compile an SC load to a load preceded by a fence.
Both of the schemes guarantee that in compiled code
there is a fence between every store and load instructions originated from SC accesses.
Regarding compilation schemes of SC accesses,
our proof of the compilation correctness from \IMMsc to \TSO depends only on this property.
That is, in this section, we concentrate only on the compilation alternative which
compiles SC stores using fences.

As a model of the \TSO architecture,
we use a declarative model from~\cite{Alglave-al:TOPLAS14}.
Its labels are given by:
\begin{itemize}
\item \TSO read label: $\erlab{}{\loc}{v}{}$ where $\loc\in\Loc$ and $v\in\Val$.
\item \TSO write label: $\ewlab{}{\loc}{v}{}$ where $\loc\in\Loc$ and $v\in\Val$.
\item \TSO fence label: $\lMFENCE$.
\end{itemize}
In turn, \TSO's execution graphs are defined as \IMMsc's ones.
Below, we interpret the compilation on execution graphs.

\begin{definition}
\label{def:cmp_exec_tso}
Let $G$ be an \IMM execution graph with whole identifiers ($G.\lE \suq \N$).
A \TSO execution graph $G_t$ \emph{corresponds} to $G$ if the following hold:
\begin{itemize}
\item 
$G_t.\lE = G.\lE \setminus G.\lF^{\neq\sco} \cup \set{n \cshift \st n \in G.\ewlab{\sco}{}{}{}}$ 
\\ (non-SC fences are removed)
\item $G_t.\lTID(e) = G.\lTID(\floor{e} \cshift)$ for all $e$ in $G_t$
\item $G_t.\lPO = {}$ \\
  $[G_t.\lE] \seq (G.\lPO \cup
  \set{\tup{a, n \cshift} \mid \tup{a, n} \in G.\lPO^?} \cup
  \set{\tup{n \cshift, a} \mid \tup{n, a} \in G.\lPO}) \seq [G_t.\lE]
$
\\ (new events are added after SC writes)
\item  
$G_t.\lLAB = \set{e \mapsto \compile{G.\lLAB(e)} \st e\in G.\lE\setminus G.\lF^{\neq\sco}} \cup
  \set{ e\mapsto  \lMFENCE \st e\in G_t.\lE \setminus G.\lE}$ where:
$$\begin{array}{r@{\;}l@{\qquad\quad}r@{\;}l@{\qquad\quad}r@{\;}l}
\compile{\erlab{o_\lR}{x}{v}{s}}     & \defeq \prlab{x}{v} &
\compile{\ewlab{o_\lW}{x}{v}{}} & \defeq \pwlab{x}{v} &
\compile{\flab{\sco}}  &\defeq \lMFENCE \\
\end{array}$$
\item  
$G.\lRMW = G_t.\lRMW$, $G.\lDATA = G_t.\lDATA$, and $G.\lADDR = G_t.\lADDR$
\\ (the compilation does not change RMW pairs and data/address dependencies)
\item
$G.\lCTRL; [G.\lE \setminus G.\lF^{\neq\sco}] \suq G_t.\lCTRL$
\\ (the compilation only adds control dependencies)
\end{itemize}
\end{definition}

The following derived relations are used to define the \TSO-consistency predicate.
\begin{align*}
\lPPOTSO &\defeq [\lR \cup \lW] ; \lPO ; [\lR \cup \lW] \setminus [\lW]; \lPO; [\lR] \\
 \lFENCETSO &\defeq [\lR \cup \lW];\lPO; [\lMFENCE]; \lPO;[\lR \cup \lW] \\
 \lIFENCE  & \defeq [\lW]; \lPO; [\dom{\lRMW}] \cup
                    [\codom{\lRMW}]; \lPO; [\lR] \\
 \lHBTSO   & \defeq \lPPOTSO \cup \lFENCETSO \cup \lIFENCE \cup \lRFE
   \cup \lCO \cup \lFR
\end{align*}

\begin{definition}
\label{def:tso}
$G$ is called \emph{\TSO-consistent} if the following hold:
\begin{itemize}
\item $\codom{G.\lRF} = G.\lR$. \labelAxiom{$\lRF$-completeness}{ax:comp}
\item For every location $\loc\in\Loc$, $G.\lCO$ totally orders $G.\ewlab{}{\loc}{}{}$. \labelAxiom{$\lCO$-totality}{ax:total}
\item $\lPO\rst{\lLOC} \cup \lRF \cup \lFR \cup \lCO$ is acyclic.
 \labelAxiom{sc-per-loc}{ax:P-sc_per_loc}
\item $G.\lRMW \cap (G.\lFRE\mathbin{;}G.\lCOE) = \emptyset$.  \labelAxiom{atomicity}{ax:at}
\item $G.\lHBTSO$ is acyclic. \labelAxiom{tso-no-thin-air}{ax:tnta}
\end{itemize}
\end{definition}

Next, we state our theorem that ensures \IMMsc-consistency if the corresponding \TSO
execution graph is \TSO-consistent.

\begin{theorem}
  \label{thm:imm-to-tso}
Let $G$ be an \IMMsc execution graph with whole identifiers ($G.\lE \suq \N$),
and let $G_t$ be an \TSO execution graph that corresponds to $G$.
Then, \TSO-consistency of $G_t$ implies \IMMsc-consistency of $G$. 
\end{theorem}
\begin{proof}[Outline]
  Since $G_t$ corresponds to $G$, we know that
  $$
  [G.\lW^\sco]; G.\lPO; [G.\lR^\sco] \subseteq G_t.\lPO; [G_t.\lMFENCE]; G_t.\lPO
  $$
  as the aforementioned property of the compilation scheme.
  We show that
  $$
  G_t.\lEHBTSO \defeq G_t.\lHBTSO \cup [G_t.\lMFENCE]; G_t.\lPO \cup [G_t.\lMFENCE]; G_t.\lPO
  $$
  is acyclic. Then, we show that $G.\lPSCB \cup G.\lPSCF$ is included in $G_t.\lEHBTSO^+$.
  It means that acyclicity of $G.\lPSCB \cup G.\lPSCF$ holds, and it leaves us to prove that
  $G$ is \IMM-consistent. That is done by standard relational techniques
  (see \cite{IMMrepo}).
  \qedhere


\end{proof}

  \newpage
\section{From \IMMsc to \POWER}
\label{app:imm_to_power}

\begin{figure}[t]
\small
\centering
$\begin{array}{@{}r@{}c@{{}\approx{}}l@{}r@{}c@{{}\approx{}}l@{}}
\compile{\readInst{\rlx}{r}{e}} &&  ``\texttt{ld}" & 
\compile{\writeInst{\rlx}{e_1}{e_2}} &&  ``\texttt{st}" \\
\compile{\readInst{\acq}{r}{e}} && ``\texttt{ld;cmp;bc;isync}"  &
\compile{\writeInst{\rel}{e_1}{e_2}} && ``\texttt{lwsync;st}" \\ 
\compile{\fenceInst{\neq \sco}}  && ``\texttt{lwsync}" &
\compile{\fenceInst{\sco}} && ``\texttt{sync}" \\
\compile{\incInst{o}{}{r}{e_1}{e_2}{}} &&
\multicolumn{4}{@{}l@{}}{ 
  {\rm wmod}(o) \doubleplus
  ``\textsf{L:}\texttt{lwarx;stwcx.;bc  }\textsf{L}"
  \doubleplus {\rm rmod}(o)
}
\\ 
\compile{\casInst{o}{}{r}{e}{e_\lR}{e_\lW}{}} &&
\multicolumn{4}{@{}l@{}}{
  {\rm wmod}(o) \doubleplus
  ``\textsf{L:}\texttt{lwarx;cmp;bc } \textsf{Le}\texttt{;stwcx.;bc  }\textsf{L}\texttt{;}\textsf{Le:}"
  \doubleplus {\rm rmod}(o)
} \\
\multicolumn{6}{@{}l@{}}{
{\rm wmod}(o) \defeq o \sqsupseteq \rel \; ? \; ``\texttt{lwsync;}" \; : \; ``" \qquad \qquad \qquad \qquad
{\rm rmod}(o) \defeq o \sqsupseteq \acq \; ? \; ``\texttt{;isync}" \; : \; ``" 
} \\
\hline
\textbf{Leading }\texttt{sync}\textbf{:} &
\multicolumn{2}{@{}l|@{}}{} &
\quad
\textbf{Trailing }\texttt{sync}\textbf{:} \\
\compile{\readInst{\sco}{r}{e}} &&
\multicolumn{1}{@{}l|@{}}{``\texttt{sync;ld;cmp;bc;isync}"} &
\compile{\readInst{\sco}{r}{e}} && ``\texttt{ld;sync}" \\

\compile{\writeInst{\sco}{e_1}{e_2}} &&
\multicolumn{1}{@{}l|@{}}{``\texttt{sync;st}"} &
\compile{\writeInst{\sco}{e_1}{e_2}} &&
``\texttt{lwsync;st;sync}" \\ 
\end{array}$
\caption{Compilation scheme from \IMM to \POWER.}
\label{fig:compPower}
\end{figure}


Here we use the same mapping of \IMM to \POWER (see~\cref{fig:compPower}) as
in~\cite{Podkopaev-al:POPL19} for all instructions except for SC accesses.
For the latter, there are two standard compilations schemes~\cite{www:mappings}
presented in the bottom of~\cref{fig:compPower}:
with leading and trailing $\texttt{sync}$ fences.

The next definition presents the correspondence between \IMM execution graphs and
their mapped \POWER ones following the leading compilation scheme
in \cref{fig:compPower} with
elimination of the aforementioned redundancy of SC write compilation.

\begin{definition}
\label{def:cmp_exec_power}
Let $G$ be an \IMM execution graph with whole identifiers 
($G.\lE \suq \N$).
A \POWER execution graph $G_p$ \emph{corresponds} to $G$ if the following hold:
\begin{itemize}
\item
  $\begin{array}[t]{@{}l@{}l@{}l@{}}
G_p.\lE = G.\lE & {} \cup
            \{n \cshift \st n \in {} & (G.\erlab{\sqsupseteq \acq}{}{}{}
                 \setminus \dom{G.\lRMW}) \cup {} \\
            & & \codom{[G.\erlab{\sqsupseteq \acq}{}{}{}]\mathbin{;}G.\lRMW} \} \\
 & {} \cup \{n \dshift \st
                   n \in {} & (G.\lE^{\sqsupseteq \rel} \setminus \dom{G.\lRMW}) \cup {} \\
            & & \dom{G.\lRMW\mathbin{;}[G.\lW^{\sqsupseteq \rel}]}\}
\end{array}$ \\
(new events are added after acquire reads and acquire RMW pairs and
    before SC accesses and SC RMW pairs)
\item $G_p.\lTID(e) = G.\lTID(\floor{e\cshift})$ for all $e$ in $G_p$
\item
$\begin{array}[t]{@{}l@{}l@{}l@{}}
G_p.\lPO = G.\lPO \cup (&(& G_p.\lE \times G_p.\lE) \cap {} \\
 &(& \set{\tup{a, n \dshift} \mid \tup{a, n} \in G.\lPO} \cup {} \\
 && \set{\tup{n \dshift, a} \mid \tup{n, a} \in G.\lPO^?} \cup {} \\
 && \set{\tup{a, n \cshift} \mid \tup{a, n} \in G.\lPO^?} \cup {} \\
 && \set{\tup{n \cshift, a} \mid \tup{n, a} \in G.\lPO}))
\end{array}$
\item  
$\begin{array}[t]{@{}l@{}l@{}}
G_p.\lLAB = &
  \set{e \mapsto \compile{G.\lLAB(e)} \st e\in G.\lE} \cup {} \\
  & \begin{array}[t]{@{}l@{}l@{}}
    \{ n \cshift \mapsto  \flab{\isync} & {} \st {}
          n \cshift \in G_p.\lE \land n \in \N \} \cup {} \\
    \{ n \dshift \mapsto  \flab{\lwsync} & {} \st {}
      \begin{array}[t]{@{}l@{}}
      n \dshift \in G_p.\lE \land n \in \N \land {} \\
      n \nin G.\lE^\sco \cup \dom{G.\lRMW\seq[G.\lW^\sco]} \}
      \cup {}
      \end{array}
     \\
      \{ n \dshift \mapsto  \flab{\sync} & {} \st {}
      \begin{array}[t]{@{}l@{}}
      n \dshift \in G_p.\lE \land n \in \N \land {} \\
      n \in G.\lE^\sco \cup \dom{G.\lRMW\seq[G.\lW^\sco]} \}
      \end{array} \\
   \end{array} \\
\end{array}$ \\
  where:
$$\begin{array}{r@{\;}l@{\qquad\quad}r@{\;}l}
\compile{\erlab{o_\lR}{x}{v}{s}}         & \defeq \prlab{x}{v}                     
& \compile{\flab{\acq}} = \compile{\flab{\rel}} = \compile{\flab{\acqrel}} &\defeq \flab{\lwsync} \\
\compile{\ewlab{o_\lW}{x}{v}{}} & \defeq \pwlab{x}{v}  & 
\compile{\flab{\sco}}  &\defeq \flab{\sync}
\end{array}$$

\item  
$G.\lRMW = G_p.\lRMW$, $G.\lDATA = G_p.\lDATA$, and $G.\lADDR = G_p.\lADDR$
\\ (the compilation does not change RMW pairs and data/address dependencies)

\item
$G.\lCTRL \suq G_p.\lCTRL$
\\ (the compilation only adds control dependencies)

\item 
$[G.\erlab{\sqsupseteq \acq}{}{}{}] \mathbin{;} G.\lPO \suq G_p.\lRMW \cup G_p.\lCTRL$
\\ (a control dependency is placed from every acquire or SC read)

\item
$[G.\erlab{}{}{}{\isex}] \mathbin{;}G.\lPO \suq G_p.\lCTRL \cup G_p.\lRMW \cap G_p.\lDATA $
\\ (exclusive reads entail a control dependency to any future event, 
except for their immediate exclusive write successor if arose from an atomic increment)

\item
$G.\lDATA \mathbin{;} [\codom{G.\lRMW}] \mathbin{;} G.\lPO \suq G_p.\lCTRL$
\\ (data dependency to an exclusive write entails a control dependency to any future event)

\item
$G.\lRMWDEP \mathbin{;} G.\lPO \suq G_p.\lCTRL$
\\ (CAS dependency to an exclusive read entails a control dependency to any future event)
\end{itemize}
\end{definition}

The correspondence between \IMM and \POWER execution graphs which follows the trailing compilation scheme
  may be presented similarly with two main difference.
  First, obviously, SC accesses are compiled to release and acquire accesses followed by SC fences:
  $$\begin{array}[t]{@{}l@{}l@{}}
G_p.\lE = G.\lE & {} \cup
            \{n \cshift \st n \in {}
            \begin{array}[t]{@{}l@{}}
            \{(G.\lE^{\sqsupseteq \acq}{}{}{}
                 \setminus \dom{G.\lRMW}) \cup {} \\
             \codom{[G.\erlab{\sqsupseteq \acq}{}{}{}]\mathbin{;}G.\lRMW}\}
            \end{array} \\
 & {} \cup \{n \dshift \st
                   n \in {}
            \begin{array}[t]{@{}l@{}}
              (G.\lW^{\sqsupseteq \rel} \setminus \dom{G.\lRMW}) \cup {} \\
              \dom{G.\lRMW\mathbin{;}[G.\lW^{\sqsupseteq \rel}]}\}
            \end{array} \\
\\
G_p.\lLAB = &
  \set{e \mapsto \compile{G.\lLAB(e)} \st e\in G.\lE} \cup {} \\
  & \begin{array}[t]{@{}l@{}l@{}}
    \{ n \cshift \mapsto  \flab{\isync} & {} \st {}
      \begin{array}[t]{@{}l@{}}
      n \cshift \in G_p.\lE \land n \in \N \land {} \\
      n \in G.\lR^\acq \cup \codom{[G.\lR^\acq]\mathbin{;}G.\lRMW} \}
      \cup {}
      \end{array}
     \\
      \{ n \cshift \mapsto  \flab{\sync} & {} \st {}
      \begin{array}[t]{@{}l@{}}
      n \cshift \in G_p.\lE \land n \in \N \land {} \\
      n \in G.\lE^\sco \cup \codom{[G.\lR^\sco]\mathbin{;}G.\lRMW} \}
      \cup {}
      \end{array} \\
    \{ n \dshift \mapsto  \flab{\lwsync} & {} \st {}
          n \dshift \in G_p.\lE \land n \in \N^+ \} \\
   \end{array} \\
\end{array}$$
Second, $[G.\erlab{\sqsupseteq \acq}{}{}{}] \mathbin{;} G.\lPO$ has to be included in
$G_p.\lRMW \cup G_p.\lCTRL \cup G_p.\lPO \mathbin{;} [G_p.\flab{\lwsync}] \mathbin{;} G_p.\lPO^?$, not just
in $G_p.\lRMW \cup G_p.\lCTRL$, to allow for elimination of the aforementioned SC read compilation redundancy.


The next theorem ensures \IMMsc-consistency if the corresponding \POWER
execution graph is \POWER-consistent.

\begin{theorem}
  \label{thm:imm-to-power}
Let $G$ be an \IMM execution graph with whole identifiers ($G.\lE \suq \N$),
and let $G_p$ be a \POWER execution graph that corresponds to $G$.
Then, \POWER-consistency of $G_p$ implies \IMMsc-consistency of $G$. 
\end{theorem}
\begin{proof}[Outline]
  We construct an \IMM execution graph $G'$ by inserting SC fences before SC accesses in $G$.
  We also construct $G_{\texttt{NoSC}}$ from $G'$ by replacing SC write and read accesses of $G'$ with release write and
  acquire read ones respectively.
  Obviously, \IMMsc-consistency of $G$ follows from \IMMsc-consistency of $G'$, which, in turn,
  follows from \IMM-consistency of $G_{\texttt{NoSC}}$ by~\cref{thm:sc-access-to-fence}.
  We construct an \IMM execution graph $G''$ from $G_{\texttt{NoSC}}$ by inserting release fences before release writes,
  and then an \IMM execution graph $G_{\texttt{NoRel}}$ from $G''$ by weakening the access modes of release write events to a relaxed mode.
  As on a previous proof step, \IMM-consistency of $G_{\texttt{NoSC}}$ follows from \IMM-consistency of $G''$, which in turn follows
  from \IMM-consistency of $G_{\texttt{NoRel}}$ by~\cite[Theorem 4.1]{Podkopaev-al:POPL19}.
  
  Thus to prove the theorem we need to show that $G_{\texttt{NoRel}}$ is \IMM-consistent.
  Note that $G_p$---the \POWER execution graph corresponding to $G$---also corresponds to $G_{\texttt{NoRel}}$
  by construction of $G_{\texttt{NoRel}}$. That is, \IMM-consistency of $G_{\texttt{NoRel}}$ follows from
  \POWER-consistency of $G_p$ by~\cite[Theorem 4.3]{Podkopaev-al:POPL19} since $G_{\texttt{NoRel}}$ does not contain
  SC read and write access events as well as release write access events.
  \qedhere
\end{proof}

}{}

\end{document}